\documentclass[a4paper,11pt,fleqn]{article}

\usepackage{amsmath, amssymb, amsthm}
\usepackage{mathtools, thm-restate}
\usepackage[margin=1in]{geometry}
\usepackage{xcolor}
\usepackage{url}
\usepackage{float}
\usepackage{array}
\usepackage{longtable}
\usepackage{comment}
\usepackage{tikz}
\usepackage{pifont}
\usepackage[shortlabels]{enumitem}
\usepackage[justification=centering]{caption}
\usepackage[bookmarksnumbered=true,hypertexnames=false]{hyperref}
\usepackage{algorithm, algpseudocode}
\usepackage[capitalize,sort]{cleveref}
\usetikzlibrary{patterns}
\usetikzlibrary{decorations.pathreplacing,arrows.meta,calc}

\hypersetup{
    colorlinks,
    linkcolor={red!50!black},
    citecolor={red!50!black},
    urlcolor={blue!50!black}
}

\newtheorem{theorem}{Theorem}
\newtheorem{definition}{Definition}

\newtheorem{observation}[theorem]{Observation}

\newtheorem{lemma}[theorem]{Lemma}
\newtheorem{claim}[theorem]{Claim}

\makeatletter
\g@addto@macro{\UrlBreaks}{%
\do\/%
\do\a\do\b\do\c\do\d\do\e\do\f\do\g\do\h\do\i\do\j\do\k\do\l\do\m%
\do\n\do\o\do\p\do\q\do\r\do\s\do\t\do\u\do\v\do\w\do\x\do\y\do\z%
\do\A\do\B\do\C\do\D\do\E\do\F\do\G\do\H\do\I\do\J\do\K\do\L\do\M%
\do\N\do\O\do\P\do\Q\do\R\do\S\do\T\do\U\do\V\do\W\do\X\do\Y\do\Z%
\do\0\do\1\do\2\do\3\do\4\do\5\do\6\do\7\do\8\do\9%
}
\makeatother

\usepackage{catchfile}

\ifx\input{"|kpsewhich -var-value RELEASE"}\empty
\def\rem#1{{\marginpar{\raggedright\scriptsize #1}}}
\newcommand*{\todo}[1]{\textcolor{violet}{TODO: #1}}
\newcommand*{\aricolor}{red}
\newcommand*{\eklcolor}{blue}
\newcommand*{\kvncolor}{green!50!black}

\newcommand*{\arir}[1]{\rem{\textcolor{\aricolor}{$\bullet$ #1}}}
\newcommand*{\eklr}[1]{\rem{\textcolor{\eklcolor}{$\bullet$ #1}}}
\newcommand*{\kvnr}[1]{\rem{\textcolor{\kvncolor}{$\bullet$ #1}}}
\else
\newcommand*{\todo}[1]{}

\newcommand{\arir}[1]{}
\newcommand{\eklr}[1]{}
\newcommand{\kvnr}[1]{}
\fi

\newenvironment*{tightemize}{\begin{itemize}[noitemsep]}{\end{itemize}}

\algnewcommand{\LineComment}[1]{\State \textcolor{gray}{\texttt{//} \textit{#1}}}

\algblockdefx{RepeatN}{EndRepeatN}[1]{\textbf{repeat} #1 \textbf{times}}{\textbf{end repeat}}
\newcolumntype{L}{>{$\displaystyle}l<{$}}

\newtheorem{property}[definition]{Property}
\newtheorem{transformation}[definition]{Transformation}
\crefname{claim}{Claim}{Claims}
\crefname{property}{Property}{Properties}
\crefname{transformation}{Transformation}{Transformations}

\newcommand*{\geomvec}[2]{(#1, #2)}
\newcommand*{\geomvecdim}[2]{(#1, #2)-dimensional}

\let\epsilon\varepsilon
\let\eps\epsilon
\newcommand*{\defeq}{\coloneqq}

\newcommand*{\floor}[1]{\left\lfloor #1 \right\rfloor}
\newcommand*{\ceil}[1]{\left\lceil #1 \right\rceil}
\newcommand*{\smallceil}[1]{\lceil #1 \rceil}
\newcommand*{\abs}[1]{\left\lvert #1 \right\rvert}
\newcommand*{\norm}[1]{\left\lVert #1 \right\rVert}
\newcommand*{\smallnorm}[1]{\lVert #1 \rVert}
\newcommand*{\Th}{^{\textrm{th}}}
\DeclareMathOperator*{\E}{\mathbb{E}}

\DeclareMathOperator*{\argmin}{argmin}

\newcommand*{\WLoG}{W.l.o.g.}
\newcommand*{\wLoG}{w.l.o.g.}

\DeclareMathOperator*{\support}{support}
\DeclareMathOperator*{\opt}{opt}

\DeclareMathOperator{\vol}{vol}

\DeclareMathOperator{\poly}{poly}
\DeclareMathOperator{\round}{\mathtt{round}}

\DeclareMathOperator{\Span}{span}
\DeclareMathOperator{\simplePack}{\mathtt{simple-pack}}
\DeclareMathOperator{\betterSimplePack}{\mathtt{better-simple-pack}}
\DeclareMathOperator{\cbPack}{\mathtt{cb-pack}}

\newcommand*{\hdhIV}{\mathrm{HDH}_4}
\newcommand*{\hdhIVunit}{\textrm{HDH-unit-pack}_4}
\newcommand*{\last}{\mathrm{last}}
\DeclareMathOperator{\type}{type}
\newcommand*{\config}{configuration}
\newcommand*{\Config}{Configuration}
\newcommand*{\LP}{\mathrm{LP}}
\newcommand*{\degree}{^{\circ}}

\newcommand*{\xhat}{\widehat{x}}
\newcommand*{\Acal}{\mathcal{A}}
\newcommand*{\Ihat}{\widehat{I}}
\newcommand*{\Jhat}{\widehat{J}}
\newcommand*{\Itild}{\widetilde{I}}
\newcommand*{\Ktild}{\widetilde{K}}
\newcommand*{\Jtild}{\widetilde{J}}
\newcommand*{\Stild}{\widetilde{S}}
\newcommand*{\Xtild}{\widetilde{X}}
\newcommand*{\Ctild}{\widetilde{C}}
\newcommand*{\best}{\mathrm{best}}
\newcommand*{\epsLP}{\eps}

\newcommand*{\Imed}{I_{\mathrm{med}}}
\DeclareMathOperator{\complexPack}{\mathtt{complex-pack}}
\DeclareMathOperator{\cLP}{configLP}
\DeclareMathOperator{\cLPsolve}{\mathtt{solve-config-lp}}
\DeclareMathOperator{\fscLP}{fsconfigLP}
\DeclareMathOperator{\fscIP}{fsconfigIP}
\DeclareMathOperator{\unround}{\mathtt{unround}}
\DeclareMathOperator{\fsopt}{fsopt}
\DeclareMathOperator{\uncovered}{uncovered}
\DeclareMathOperator{\rnaPack}{\mathtt{rna-pack}}

\newcommand*{\precImm}{\prec_{\mathrm{imm}}}
\newcommand*{\epsLarge}{\eps_1}
\newcommand*{\epsSmall}{\eps_2}
\newcommand*{\deltaDense}{\delta_d}
\newcommand*{\lwc}{\mathrm{lwc}}  %
\newcommand*{\hwc}{\mathrm{hwc}}  %
\newcommand*{\bwc}{\mathrm{bwc}}  %
\newcommand*{\wwc}{\mathrm{wwc}}  %
\newcommand*{\swc}{\mathrm{swc}}  %
\newcommand*{\itild}{\tilde{\imath}}
\newcommand*{\nlwcHyp}{\hyperref[lem:round-up-light-n]{n_{\lwc}}}
\newcommand*{\nhwcHyp}{\hyperref[lem:round-up-heavy-n]{n_{\hwc}}}
\newcommand*{\nbwcHyp}{\hyperref[lem:num-weight-classes]{n_{\bwc}}}
\newcommand*{\nwwcHyp}{\hyperref[lem:num-weight-classes]{n_{\wwc}}}
\newcommand*{\nswcHyp}{\hyperref[lem:num-weight-classes]{n_{\swc}}}
\newcommand*{\deltaLG}{\delta_{\mathrm{lg}}}
\newcommand*{\nWCont}{n_{\mathrm{wcont}}}
\newcommand*{\nTCont}{n_{\mathrm{tcont}}}
\newcommand*{\epsCont}{\eps_{\mathrm{cont}}}
\newcommand*{\nBig}{n_{\mathrm{big}}}
\newcommand*{\nNConfs}{n_{\mathrm{nconfs}}}
\newcommand*{\Pcal}{\mathcal{P}}
\newcommand*{\Pcalhat}{\widehat{\mathcal{P}}}
\DeclareMathOperator{\leader}{leader}
\DeclareMathOperator*{\cartProd}{\prod}
\DeclareMathOperator{\outputAdd}{\mathtt{outputs.add}}
\DeclareMathOperator{\remMed}{remove-medium}
\DeclareMathOperator{\iterFineParts}{\mathtt{iter-fine-parts}}
\DeclareMathOperator{\iterFinePartsHyp}{\hyperref[algo:ifp]{\iterFineParts}}
\DeclareMathOperator{\partBig}{\mathtt{part-big}}
\let\partBigHyp\partBig
\DeclareMathOperator{\partWide}{\mathtt{part-wide}}
\let\partWideHyp\partWide
\DeclareMathOperator{\roundHyp}{\hyperref[algo:round]{\round}}
\DeclareMathOperator{\FP}{FP}

\DeclareMathOperator{\fpack}{\mathtt{fpack}}
\DeclareMathOperator{\ipack}{\mathtt{ipack}}
\DeclareMathOperator{\ipackHyp}{\hyperref[algo:ipack]{\ipack}}
\DeclareMathOperator{\packSmall}{\mathtt{pack-small}}
\DeclareMathOperator{\packSmallHyp}{\hyperref[algo:pack-small]{\packSmall}}

\DeclareMathOperator{\remMedHyp}{\hyperref[defn:remmed]{\remMed}}
\newcommand*{\weightRoundingHyp}{\hyperref[trn:wround]{weight-rounding}}
\newcommand*{\compartmentalHyp}{\hyperref[defn:compartmental]{compartmental}}
\newcommand*{\finePartHyp}{\hyperref[defn:fine-part]{fine partitioning}}
\newcommand*{\semiStrucHyp}{\hyperref[defn:semi-struct]{semi-structured}}

\newenvironment{omittedproof}[1]
{\begin{proof}[Proof of \Cref{#1}]}
{\end{proof}}

\newcommand{\acknowledgements}[1]{\par\noindent\textbf{Acknowledgements:} #1}

\title{Geometry Meets Vectors: Approximation\texorpdfstring{\\}{ }Algorithms for Multidimensional Packing}
\author{Arindam Khan $^*$ \and Eklavya Sharma $^*$ \and K. V. N. Sreenivas
\thanks{Department of Computer Science and Automation, Indian Institute of Science, Bengaluru, India.
{\tt arindamkhan@iisc.ac.in}, {\tt eklavyas@iisc.ac.in}, {\tt venkatanaga@iisc.ac.in}}}
\date{\empty}

\begin{document}

\maketitle
\setlength{\parskip}{0.25em}

\begin{abstract}
We study the generalized multidimensional bin packing problem (GVBP) that generalizes
both geometric packing and vector packing.
Here, we are given $n$ rectangular items where the $i\Th$ item has
width $w(i)$, height $h(i)$, and $d$ nonnegative weights $v_1(i), v_2(i), \ldots, v_{d}(i)$.
Our goal is to get an axis-parallel non-overlapping packing of the items into square bins so that
for all $j \in [d]$, the sum of the $j\Th$ weight of items in each bin is at most 1.
This is a natural problem arising in logistics, resource allocation, and scheduling.
Despite being well studied in practice, surprisingly,
approximation algorithms for this problem have rarely been explored.

We first obtain two simple algorithms for GVBP having asymptotic approximation ratios
$6(d+1)$ and $3(1 + \ln(d+1) + \eps)$.
We then extend the Round-and-Approx (R\&A) framework
\texorpdfstring{\cite{bansal2009new,bansal2014binpacking}}{[Bansal-Khan, SODA'14]}
to wider classes of algorithms, and show how it can be adapted to GVBP.
Using more sophisticated techniques, we obtain better approximation algorithms for GVBP,
and we get further improvement by combining them with the R\&A framework.
This gives us an asymptotic approximation ratio of $2(1+\ln((d+4)/2))+\eps$ for GVBP,
which improves to $2.919+\eps$ for the special case of $d=1$.
We obtain further improvement when the items are allowed to be rotated.
We also present algorithms for a generalization of GVBP
where the items are high dimensional cuboids.

\end{abstract}

\section{Introduction}
\label{sec:intro}

Bin packing and knapsack problems are classical NP-hard optimization problems.
Two classical generalizations of these problems: geometric packing and vector packing
have been well-studied from the 1980s \cite{coffman1980performance,vl81}.
Geometric packing considers the packing of rectangular items, whereas,
in vector packing items are multidimensional vectors.
\cref{fig:geom-vec-diff} illustrates the difference between geometric packing and
vector packing.
However, often in practice, we encounter a mixture of geometric and vector constraints.
Consider the following airlines cargo problem \cite{paquay2016mixed}:
We have boxes to load in an airline cargo container.
In addition to the geometric constraint that all the boxes must fit within the container,
we also have a constraint that the total weight of the loaded boxes
is within a specified capacity. Thus, in this problem,
three dimensions are geometric and the weight is a vector constraint.

\begin{figure}[htb]
\centering

\begin{tikzpicture}[
    scale = 2.4,
    every node/.style = {scale=0.7}
]
    \newcommand*{\shiftline}{0.05}
    \draw[thick] (0,0) rectangle +(1,1);
    \draw[fill=red!50] (0,0) rectangle +(0.8,0.2);
    \draw[fill=blue!50] (0,0.2) rectangle +(0.4,0.4);
    \draw[fill=green!50] (0.4,0.2) rectangle +(0.4,0.4);
    \draw[<->] (-\shiftline,0) --node[midway,left]{\small{$0.2$}}
                            +(0,0.2);
    \draw[<->] (-\shiftline,0.2) --node[midway,left]{\small{$0.4$}}
                            +(0,0.4);
    \draw[<->] (0,-\shiftline) --node[midway,below]{\small{$0.8$}}
                            +(0.8,0);
    \draw[<->] (0,0.6+\shiftline)
                --node[midway,above]{\small{$0.4$}} +(0.4,0);
    \draw[<->] (0.4,0.6+\shiftline)
                --node[midway,above]{\small{$0.4$}} +(0.4,0);

    \tikzset{shift={(1.5,0)}}
    \draw[thick] (0,0) rectangle +(1,1);
    \draw[->,red,thick] (0,0) -- +(0.8,0.2);
    \draw[->,blue,thick] (0.8,0.2) -- +(0.4,0.4);
    \draw[->,green,thick] (1.2,0.6) -- +(0.4,0.4);
    \draw[dotted] (0,0) rectangle +(0.8,0.2);
    \draw[dotted] (0.8,0.2) rectangle +(0.4,0.4);
    \draw[dotted] (1.2,0.6) rectangle +(0.4,0.4);
    \draw[<->] (0,-\shiftline) --node[midway,below]{\small{$0.8$}}
                            +(0.8,0);
    \draw[<->] (0.8,0.2-\shiftline) --node[midway,below right]{\small{$0.4$}}
                            +(0.4,0);
    \draw[<->] (-\shiftline,0) --node[midway,left]{\small{$0.2$}}
                            +(0,0.2);
    \draw[<->] (0.8-\shiftline,0.2)
                --node[midway,left]{\small{$0.4$}} +(0,0.4);
    \draw[<->] (1.2,0.6-\shiftline)
                --node[midway,below]{\small{$0.4$}} +(0.4,0);
    \draw[<->] (1.6+\shiftline,0.6)
                --node[midway,right]{\small{$0.4$}} +(0,0.4);

\end{tikzpicture}

\caption{Geometric \& vector packing: We can pack rectangles
$(0.8,0.2)$, $(0.4,0.4)$, $(0.4,0.4)$ into unit geometric bin but can't
pack vectors $(0.8,0.2)$, $(0.4,0.4)$, $(0.4,0.4)$ into unit vector bin.}
\label{fig:geom-vec-diff}
\end{figure}
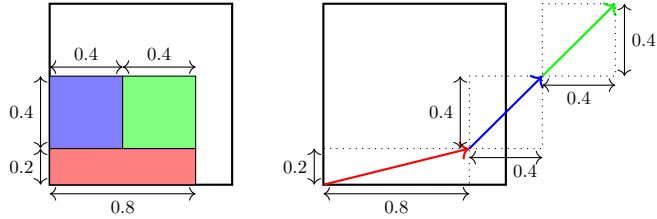

Weight has been an important constraint to consider for packing in logistics and supply chain
management, e.g., cranes and other equipment can be damaged by the bins being too heavy,
or by a skewed weight distribution \cite{alonso2017mathematical}.
While the container loading problems mostly consider packing items inside a container,
the container stowage planning problem considers the stacking of the containers
onto and off cargo ships \cite{monaco2014terminal}.
Even when different cargoes are packed into a fleet of aircraft for transport,
one needs the individual cargoes to be not too heavy to ensure stability
and less fuel consumption \cite{amiouny1992balanced}.
Similar problems find applications in vehicle routing with loading constraints \cite{bortfeldt2013constraints}.
Many practical heuristics \cite{sorset2019heuristic,taylor2017three}
have been proposed for these kinds of problems.
Several companies (such as Driw, Boxify, Freightcom) and practical packages \cite{yang2017gbp}
have considered the problem. In many cases, we also want to ensure a limit on other attributes,
such as the amount of magnetism, radioactivity, or toxicity.
Each of these properties can be considered as additional vector dimensions.

Such multidimensional packing problems are also getting attention due to their
connections with fair resource allocation \cite{fairknap}.
In recent years, a considerable amount of research has focused on
group fairness \cite{JosephKMR16, TsangWRTZ19} such that the algorithms are
not biased towards (or against) some groups or categories.
One such notion of fairness is {\em restricted dominance} \cite{BeraCFN19},
which upper bounds the number (or size) of items from a category.
These different categories can be considered as dimensions.
E.g., in a container packing problem for flood relief,
one needs to ensure that the money spent on a container is fairly distributed among
different types of items (such as medicine, food, garments).
Hence, for each category, there is an upper bound on the value that can go into a container.

Formally, we are given $n$ items $I\defeq\{1, 2, \dots, n\}$
that are \geomvecdim{$d_g$}{$d_v$}, i.e., item $i$ is a
$d_g$-dimensional cuboid of lengths $\ell_1(i), \ell_2(i), \ldots, \ell_{d_g}(i)$
and has $d_v$ non-negative weights $v_1(i), v_2(i), \ldots, v_{d_v}(i)$.
A \geomvecdim{$d_g$}{$d_v$} bin is a $d_g$-dimensional cuboid of length 1
in each geometric dimension and weight capacity 1 in each of the $d_v$ vector dimensions.
A feasible packing of items into a bin is a packing where items are packed parallel to the axes
without overlapping, and for all $j \in [d_v]$,
the sum of the $j\Th$ weight-dimension of the items in the bin is at most 1.
In the \geomvec{$d_g$}{$d_v$} bin packing problem (BP),
we have to feasibly pack all items into the minimum number of bins.
In the \geomvec{$d_g$}{$d_v$} knapsack problem (KS),
each item $i$ also has an associated nonnegative profit $p(i)$,
and we have to feasibly pack a maximum-profit subset of the items into a single bin
(also called `knapsack').
\geomvec{$d_g$}{$d_v$} packing problems generalize both $d_g$-dimensional geometric packing
(when $d_v=0$) and $d_v$-dimensional vector packing (when $d_g=0$).
Already for vector packing, if $d_v$ is part of the input, there is an approximation
hardness of ${d_v}^{1-\eps}$ unless NP=ZPP \cite{BansalE016}.
Thus, throughout the paper we assume both $d_g, d_v$ to be constants.

\subsection{Our Results}
\label{subsec:ourcont}

We study the first approximation algorithms for general \geomvec{$d_g$}{$d_v$} BP,
with a focus on $d_g = 2$.
We give two simple algorithms for \geomvec{2}{$d$} BP, called $\simplePack$
and $\betterSimplePack$, having asymptotic approximation ratios (AARs) of
$6(d+1)$ and $3(1+\ln(d+1))+\eps$, respectively, for any $\eps > 0$.
For $d = 1$, $\betterSimplePack$'s AAR improves to $\approx 4.216 + \eps$.

Next, we modify the Round-and-Approx (R\&A) framework \cite{bansal2014binpacking}
so that it works for \geomvec{$d_g$}{$d_v$} BP.
We combine R\&A with the $\simplePack$ algorithm to get an AAR of
$2(1 + \ln(3(d+1))) + \eps$ for \geomvec{2}{$d$} BP.
This improves upon the AAR of $\betterSimplePack$ for $d \ge 3$.

Next, we obtain a more sophisticated algorithm for \geomvec{2}{$d$} BP, called $\cbPack$,
that fits into the R\&A framework and has an even better AAR.
\begin{theorem}
\label{thm:2-d-gvbp}
For any $\eps > 0$, there is a polynomial-time algorithm for \geomvec{2}{$d$} BP,
called $\cbPack$, having AAR $2(1+\ln((d+4)/2))+\eps$
(improves to $2(1+\ln((d+3)/2))+\eps$ when items can be rotated by $90\degree$).
For $d=1$, the AAR improves to $2(1+\ln(19/12))+\eps$ $\approx 2.919+\eps$
(further improves to $2(1+\ln(3/2))+\eps$ $\approx 2.811 + \eps$ when items can be rotated).
\end{theorem}

\begin{table}[!ht]
\centering
\caption{Comparison of asymptotic approximation ratios (AARs) of
algorithms for \geomvec{2}{$d$} BP.}
\begin{tabular}{|l|c|c|}
\hline & \geomvec{2}{$d$} BP & \geomvec{2}{1} BP
\\ \hline $\simplePack$
    & $6(d+1)$
    & $12$
\\ \hline $\betterSimplePack$
    & $3(1 + \ln(d+1)) + \eps$
    & $3(1 + \ln(3/2)) + \eps \approx 4.216+\eps$
\\ \hline $\simplePack$ with R\&A
    & $2(1 + \ln(3(d+1))) + \eps$
    & $2(1 + \ln 6) + \eps$ $\approx 5.5835+\eps$
\\ \hline $\cbPack$ (without rotation)
    & {$2(1+\ln(\frac{d+4}{2})) + \eps$}
    & $2(1+\ln(19/12))+\eps$ $\approx 2.919+\eps$
\\ \hline $\cbPack$ (with rotation)
    & ${2(1+\ln(\frac{d+3}{2})) + \eps}$
    & $2(1+\ln(3/2))+\eps$ $\approx 2.811+\eps$
\\ \hline
\end{tabular}
\label{table:gvbp-aar}
\end{table}

We also show how to extend $\simplePack$ and $\betterSimplePack$ to \geomvec{$d_g$}{$d_v$} BP
to obtain AARs $2b(d_v+1)$ and $b(1 + \ln(d_v + 1) + \eps)$, respectively,
where $b \defeq 4^{d_g} + 2^{d_g}$.
We also give a similar algorithm for \geomvec{$d_g$}{$d_v$} KS
having approximation ratio $b(1+\eps)$.

\subsection{Related Work}
\label{subsec:prior}

The bin packing problem (BP) has been the cornerstone of approximation algorithms \cite{Johnson}.
The standard performance measure for BP algorithms is the asymptotic approximation ratio (AAR).
An asymptotic polynomial time approximation scheme (APTAS) for BP was given by Fernandez de la Vega and Lueker \cite{vl81}, using linear grouping.
Note that 1-D BP can be considered as \geomvec{$1$}{$0$} BP as well as \geomvec{$0$}{$1$} BP.
The present best approximation algorithm for 1-D BP returns a packing
in at most $\opt+O(\log\opt)$ bins, where $\opt$ is the optimal number of bins \cite{HobergR17}.
Knapsack problem (KS) is one of Karp's 21 NP-complete problems.
Lawler gave an FPTAS \cite{Lawler77} for KS.
For surveys on BP and KS, we refer the readers to \cite{coffman2013bin,kellererBook}.

In \cite{vl81}, a $(d+\eps)$-asymptotic approximation algorithm was given
for the $d$-dimensional vector bin packing ($d$-D VBP).
Chekuri and Khanna gave a $O(\log d)$ approximation for $d$-D VBP \cite{ChekuriK04}.
The study of 2-D geometric bin packing (2-D GBP) was initiated by \cite{coffman1980performance}.
Caprara gave a $T_{\infty}^{d-1}$-asymptotic approximation Harmonic-Decreasing-Height (HDH)
algorithm for $d$-D GBP \cite{caprara2002packing}, where $T_{\infty} \approx 1.6901$.
This is still the best known approximation for $d$-D GBP for $d\ge3$.
Both 2-D GBP and 2-D VBP were shown to be APX-hard \cite{bansal2006bin,Woeginger97}.

Bansal, Caprara and Sviridenko \cite{bansal2009new} introduced the
{\em Round-and-Approx (R\&A)} framework to obtain improved approximations
for both 2-D GBP and $d$-D VBP. The R\&A framework is a two stage process.
First, a (possibly exponential-sized) set covering LP relaxation (called configuration LP)
is solved approximately. Then, a randomized rounding procedure is applied for a few steps
to pack a subset of items, after which only a `small' fraction of items
(called the \emph{residual instance}) remain unpacked.
In the second step, the residual instance is packed using a {\em subset-oblivious} algorithm.
Intuitively, given a {\em random} subset $S$ of $I$ where each element occurs
with probability about $1/k$, a $\rho$-approximate subset-oblivious algorithm
produces a packing of $S$ in approximately $\rho \opt(I)/k$ bins.
In the R\&A framework, one can obtain a $(1+\ln\rho)$-approximation algorithm using
a $\rho$-approximate subset oblivious algorithm.
Two algorithms, 1-D BP APTAS \cite{vl81} and HDH \cite{caprara2002packing}
were shown to be subset-oblivious based on various properties of dual-weighting functions.
This led to an AAR of $(1+\ln(1.69))$ and $(1+\ln d)$
for 2-D GBP and $d$-D VBP, respectively.
However, it was cumbersome to extend subset-obliviousness to wider classes of algorithms.

Bansal and Khan \cite{bansal2014binpacking} later extended the R\&A framework for 2-D GBP to
{\em rounding-based} algorithms, where the large dimensions are rounded up to $O(1)$ values
and the packing of items is container-based, i.e. each bin contains a constant number of
rectangular regions called containers and items are packed in a special way into containers.
For 2-D GBP, they used a 1.5-asymptotic-approximation algorithm \cite{jansen2016new}
to obtain the present best $(1+\ln 1.5) \approx 1.405$-asymptotic-approximation.
For $d$-D VBP, Bansal et al.~\cite{BansalE016} used the R\&A framework
combined with a multi-objective budgeted matching problem,
to obtain the present best AAR of $(0.81+o_d(1)+\ln d)$.

Multidimensional knapsack is also well-studied.
For $d$-D vector knapsack ($d$-D VKS), Frieze and Clarke gave a PTAS \cite{FriezeClarke84}.
For 2-D geometric knapsack (GKS), Jansen and Zhang \cite{jansen2004maximizing}
gave a 2-approximation algorithm, while the present best approximation ratio is 1.89 \cite{l-packing}.
It is not even known whether 2-D GKS is APX-hard or not.
There are many other related important geometric packing problems,
such as strip packing \cite{GGIK16,Galvez20} and
maximum independent set of rectangles \cite{KhanR20}.
For surveys on multidimensional packing, see \cite{CKPT17,Khan16}.

\subsection{Technical Contribution}
\label{sec:intro:tech-contrib}

One of our main contributions is the enhancement of R\&A framework \cite{bansal2014binpacking} to wider applications.

First, R\&A framework now also works with \geomvecdim{$d_g$}{$d_v$} items,
unifying the approach for geometric and vector packing.
To use R\&A, we need to solve the \config{} LP of the corresponding bin packing problem.
All previous applications (d-D VBP and 2-D GBP) of R\&A solved the \config{} LP within
$(1+\eps)$ factor using a $(1+O(\eps))$-approximate solution to KS.
Due to the unavailability of a PTAS for \geomvec{2}{$d$} KS, we had to adapt and use a different
linear programming algorithm \cite{eku-pst} that uses an
$\eta$-approximation algorithm for KS to
$(1+\eps) \eta$-approximately solve the \config{} LP of
the corresponding BP problem, for any constants $1 < \eta$, $0<\eps < 1$.

Second, we introduce more freedom in choosing the packing structure.
Unlike the R\&A framework in \cite{bansal2014binpacking}
that worked only for container-based packing,
we allow either relaxing the packing structure to non-container-based
(like in $\simplePack$)
or imposing packing constraints in addition to being container-based
(like in $\cbPack$).
This generalization can help in finding better algorithms for other variants of BP.

Finally, we allow rounding items in ways other than rounding up, if we can find
a suitable way of \emph{unrounding} a packing of rounded items.
In $\cbPack$, we round down the width and height of
some items to 0, and in $\simplePack$,
we round each \geomvecdim{2}{$d$} item $i$ to an item of width 1,
height $x$ and each vector dimension $x$, where $x$ is a value depending
on the original dimensions of $i$.
As it was shown in \cite{Khan16}, if the large coordinates of items are rounded to $O(1)$ types,
we cannot obtain better than $d$ and $4/3$-approximation for $d$-D VBP and 2-D GBP, respectively.
However, as we now allow rounding down, we may be able to use the R\&A framework
with algorithms having better approximation ratios.

We also fix a minor error in the R\&A framework of \cite{Khan16}
(see \cref{sec:rna-extra:bugfix} for details).

In \cite{bansal2009new}, it was mentioned:
``One obstacle against the use of R\&A for other problems is the difficulty in
deriving subset-oblivious algorithms (or proving that existing algorithms are subset oblivious).''
We expect that our progress will help in understanding the power of R\&A to
extend it to other set-cover type problems,
e.g. round-SAP and round-UFP \cite{ElbassioniGGKN12}.

Our another major contribution is handling of the \geomvec{2}{$d$} BP problem.
This problem presents additional challenges over pure geometric BP,
and our algorithm $\cbPack$ demonstrates how to circumvent those challenges.
For example, in geometric packing, items having low total area can be packed into a small number of bins
using the Next-Fit-Decreasing-Height (NFDH) algorithm
(see \cite{coffman1980performance} and \cref{lem:nfdh-strip,lem:nfdh-small} in \cref{sec:next-fit}).
This need not be true when items have weights, since the geometric dimensions
can be small but the vector dimensions may be large.
To handle this, we divide the items into different classes based on density (i.e., weight/area).
We use the facts that items of low density and low total area can be packed into a small number of bins,
and items of high density that fit into a bin have low total area.
Also, in geometric packing, we can sometimes move items with similar geometric dimensions
across different bins (like in \emph{linear grouping} \cite{vl81,pradel-thesis}).
Vector dimensions again create problems here, since the items can have very different weights.
To handle this, we only move items of similar geometric dimensions and density.
This leads to a more {\em structured} packing and we show how to
find such a near-optimal structured packing efficiently.

\section{Preliminaries and Notation}
A valid \emph{packing} of a given set of
\geomvecdim{$d_g$}{$d_v$} items into a \geomvecdim{$d_g$}{$d_v$} bin
is an arrangement of the items in the bin such that:
\emph{(i)} All items are packed in an axis parallel manner, i.e.,
    each item has its faces parallel to the faces of the bin.
    Formally, to pack an item $i$ into a bin, we need to decide its position
    $(x_1(i), x_2(i), \ldots, x_{d_g}(i))$ in the bin,
    and the item will be packed in the cuboidal region
    $\prod_{j=1}^{d_g} [x_j(i), x_j(i)+\ell_j(i)]$.
\emph{(ii)} The items are non-overlapping, i.e., the interiors of any two items
    do not intersect. Formally, for any two items $i_1$ and $i_2$ in the same bin, the sets
    $\prod_{j=1}^{d_g} (x_j(i_1), x_j(i_1)+\ell_j(i_1))$ and
    $\prod_{j=1}^{d_g} (x_j(i_2), x_j(i_2)+\ell_j(i_2))$ do not intersect.
\emph{(iii)} All items are packed completely inside the bin,
    i.e., for each item $i$ and each $j \in [d_g]$,
    $x_j(i) \ge 0$ and $x_j(i) + \ell_j(i) \le 1$.
\emph{(iv)} The total weight packed in each of the $d_v$ vector dimensions is at most one.

Let $\mathcal P$ be a minimization problem and let $\mathcal I$ denote the set of
all input instances of $\mathcal P$. Consider a polynomial-time approximation algorithm
$\mathcal A$ for $\mathcal{P}$. For any input $I \in \mathcal I$, let
$\opt(I)$ denote the cost of the optimal solution and let
$\mathcal A(I)$ denote the cost of $\Acal$'s output on $I$.
Then the quantity
\begin{align*}
\rho_{\mathcal A} \defeq \sup_{I\in\mathcal I}\frac{\mathcal A(I)}{\opt(I)}
\end{align*}
is called the \emph{approximation ratio} of $\mathcal A$ and the
\emph{asymptotic approximation ratio (AAR)} of $\mathcal A$ is defined as
\begin{align*}
\rho_{\mathcal A}^\infty \defeq
\lim_{z\to\infty}\sup_{I\in\mathcal I}
    \left\{\frac{\mathcal A(I)}{\opt(I)}\Big\vert\opt(I)=z\right\}.
\end{align*}
Intuitively, AAR denotes the algorithm's performance for inputs whose optimum value is large.

Let $[n] \defeq \{1, 2, \ldots, n\}$.
Let $\poly(n)$ be the set of polynomial and sub-polynomial functions of $n$.
Define $v_{\max}$, $\vol$, and $\Span$ as follows:
$v_{\max}(i) \defeq \max_{j=1}^{d_v} v_j(i)$,
$\vol(i) \defeq \prod_{j=1}^{d_g} \ell_j(i)$,
$\Span(i) \defeq \max(\vol(i), v_{\max}(i))$.
$\Span(i)$ is, intuitively, the measure of \emph{largeness} of item $i \in [n]$.
For convenience, let $v_0(i) \defeq \vol(i)$.
Assume \wLoG{}\ that $\vol(i) = 0$ implies $(\forall j \in [d_g], \ell_j(i) = 0)$.
For a set $I$ of items, given a function $f: I \mapsto \mathbb{R}$,
for $S \subseteq I$, define $f(S) \defeq \sum_{i \in S} f(i)$.
This means, e.g., $\vol(S) \defeq \sum_{i \in S} \vol(i)$.
For any bin packing algorithm $\Acal$, let $\Acal(I)$ be the resulting bin packing of items $I$,
and let $|\Acal(I)|$ be the number of bins in $\Acal(I)$.
Define $\opt(I)$ as the minimum number of bins needed to pack $I$.

\begin{lemma}
\label{lem:span-lb-opt}
For \geomvec{$d_g$}{$d_v$} items $I$, $\ceil{\Span(I)} \le (d_v+1)\opt(I)$.
\end{lemma}
\begin{proof}
Let there be $m$ bins in an optimal packing of \geomvecdim{$d_g$}{$d_v$} items $I$.
In this optimal packing, let $J_j$ be the items in the $j\Th$ bin. Then
\begin{align*}
\ceil{\Span(I)} &= \ceil{\sum_{k=1}^m \sum_{i \in J_k} \max_{j=0}^{d_v} v_j(i)}
\le \ceil{\sum_{k=1}^m \sum_{i \in J_k} \sum_{j=0}^{d_v} v_j(i)}
\le \ceil{\sum_{k=1}^m \sum_{j=0}^{d_v} 1}
= (d_v+1)m  \qedhere
\end{align*}
\end{proof}

For $d_g = 2$, let $w(i) \defeq \ell_1(i)$, $h(i) \defeq \ell_2(i)$
be the width and height of item $i$, respectively.
The area of item $i$ is $a(i) \defeq w(i)h(i) = \vol(i)$.
The items in \geomvec{$2$}{$0$} BP are called `rectangles'.

\subsection{\Config{} LP}

For a bin packing instance $I$ containing $n$ items,
a \config{} of $I$ is a packing of a subset of items of $I$ into a bin
(see \cref{fig:22-config}).
We denote the set of all \config{}s of $I$ by $\mathcal{C}$.
The \config{} matrix of $I$ is a matrix $A \in \{0, 1\}^{n \times |\mathcal{C}|}$ where
$A[i, C]$ is 1 if \config{} $C$ contains item $i$ and 0 otherwise.
To solve the bin packing problem, it is sufficient to decide
the number of bins of each \config{}.
This gives us the following linear programming relaxation, called a \config{} LP:
\[ \min_{x \in \mathbb{R}^{|\mathcal{C}|}} \sum_{C \in \mathcal{C}} x_C
\quad\textrm{where}\quad Ax \ge \mathbf{1} \textrm{ and } x \ge 0 \]
A feasible solution $x$ to the \config{} LP is a vector in $\mathbb{R}^{|\mathcal{C}|}$,
and $|\mathcal{C}|$, the number of \config{}s, can be exponential in $n$.
But if $|\support(x)| \in \poly(n)$, then $x$ may have a polynomial-sized representation.

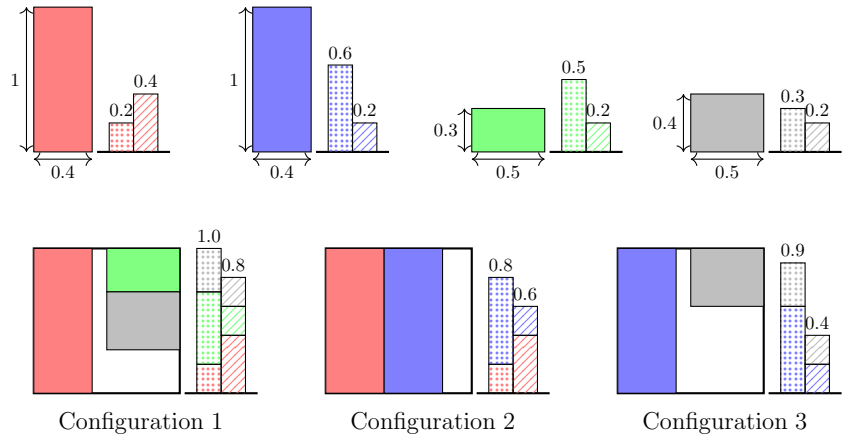
\begin{figure}[htb]
\centering
\begin{tikzpicture}[
    every node/.style={scale=0.64},
    scale=0.64,
]
\pgfdeclarepatternformonly{densedots}{
    \pgfqpoint{-1pt}{-1pt}}{\pgfqpoint{1pt}{1pt}}{\pgfqpoint{2pt}{2pt}
}
{
  \pgfpathcircle{\pgfqpoint{0pt}{0pt}}{.5pt}
  \pgfusepath{fill}
}
\pgfkeys{/pgf/number format/precision=1}
\newcommand*{\shifthist}{0.25cm}
\newcommand*{\shiftlabel}{0.15}
\newcommand*{\histwidth}{0.5cm}
\newcommand*{\scalehist}{3}
\newcommand*{\scaleitem}{3}
\newcommand*{\fillop}{0.5}
\newcommand*{\truew}[1]{
    \ifnum#1=1
        0.4
    \else
        \ifnum#1=2
            0.4
        \else
            \ifnum#1=3
                0.5
            \else
                0.5
            \fi
        \fi
    \fi
}
\newcommand*{\trueh}[1]{
    \ifnum#1=1
        1
    \else
        \ifnum#1=2
            1
        \else
            \ifnum#1=3
                0.3
            \else
                0.4
            \fi
        \fi
    \fi
}
\newcommand*{\truevone}[1]{
    \ifnum#1=1
        0.2
    \else
        \ifnum#1=2
            0.6
        \else
            \ifnum#1=3
                0.5
            \else
                0.3
            \fi
        \fi
    \fi
}
\newcommand*{\truevtwo}[1]{
    \ifnum#1=1
        0.4
    \else
        \ifnum#1=2
            0.2
        \else
            \ifnum#1=3
                0.2
            \else
                0.2
            \fi
        \fi
    \fi
}
\newcommand*{\w}[1]{\truew{#1}*\scaleitem}
\newcommand*{\h}[1]{\trueh{#1}*\scaleitem}
\newcommand*{\vone}[1]{\truevone{#1}*\scalehist}
\newcommand*{\vtwo}[1]{\truevtwo{#1}*\scalehist}
\newcommand*{\col}[1]{
    \ifnum#1=1 \colorlet{code}{red!50}
    \else
        \ifnum#1=2
            \colorlet{code}{blue!50}
        \else
            \ifnum#1=3
                \colorlet{code}{green!50}
            \else
                \colorlet{code}{gray!50}
            \fi
        \fi
    \fi
}
\newcommand*{\histone}[2]{
    \col{#2};
    \draw[thick,black] #1 -- +(2*\histwidth+2*\shifthist,0);
    \filldraw[pattern=densedots, pattern color=code]
        ($#1+(\shifthist,0)$) rectangle +(\histwidth,\vone{#2});
    \filldraw[pattern=north east lines, pattern color=code]
        ($#1+(\shifthist+\histwidth,0)$) rectangle +(\histwidth,\vtwo{#2});

    \foreach[evaluate={\height=\truevone{#2}}] \x in {0}{
    \node at ($#1+(\shifthist+0.5*\histwidth,\vone{#2})$) [above]
                                            {\pgfmathroundtozerofill{\height}\pgfmathresult};
                                        }
    \foreach[evaluate={\height=\truevtwo{#2}}] \x in {0}{
    \node at ($#1+(\shifthist+1.5*\histwidth,\vtwo{#2})$) [above]
                                            {\pgfmathroundtozerofill{\height}\pgfmathresult};
                                        }
}
\newcommand*{\histtwo}[3]{
    \col{#2};
    \draw[thick,black] #1 -- +(2*\histwidth+2*\shifthist,0);
    \filldraw[pattern=densedots, pattern color=code]
        ($#1+(\shifthist,0)$) rectangle +(\histwidth,\vone{#2});
    \filldraw[pattern=north east lines, pattern color=code]
        ($#1+(\shifthist+\histwidth,0)$) rectangle +(\histwidth,\vtwo{#2});

    \col{#3};
    \filldraw[pattern=densedots, pattern color=code]
        ($#1+(\shifthist,\vone{#2})$) rectangle +(\histwidth,\vone{#3});
    \filldraw[pattern=north east lines, pattern color=code]
        ($#1+(\shifthist+\histwidth,\vtwo{#2})$) rectangle +(\histwidth,\vtwo{#3});

    \foreach[evaluate={\height=(\truevone{#2}+\truevone{#3})}] \x in {0}{
    \node at ($#1+(\shifthist+0.5*\histwidth,\height*\scalehist)$) [above]
                                            {\pgfmathroundtozerofill{\height}\pgfmathresult};
                                        }
    \foreach[evaluate={\height=(\truevtwo{#2}+\truevtwo{#3})}] \x in {0}{
    \node at ($#1+(\shifthist+1.5*\histwidth,\height*\scalehist)$) [above]
                                            {\pgfmathroundtozerofill{\height}\pgfmathresult};
                                        }
}
\newcommand*{\histthree}[4]{
    \draw[thick,black] #1 -- +(2*\histwidth+2*\shifthist,0); %

    \col{#2};
    \filldraw[pattern=densedots, pattern color=code]
        ($#1+(\shifthist,0)$) rectangle +(\histwidth,\vone{#2});
    \filldraw[pattern=north east lines, pattern color=code]
        ($#1+(\shifthist+\histwidth,0)$) rectangle +(\histwidth,\vtwo{#2});

    \col{#3};
    \filldraw[pattern=densedots, pattern color=code]
        ($#1+(\shifthist,\vone{#2})$) rectangle +(\histwidth,\vone{#3});
    \filldraw[pattern=north east lines, pattern color=code]
        ($#1+(\shifthist+\histwidth,\vtwo{#2})$) rectangle +(\histwidth,\vtwo{#3});

    \col{#4};
    \filldraw[pattern=densedots, pattern color=code]
        ($#1+(\shifthist,\vone{#2}+\vone{#3})$) rectangle +(\histwidth,\vone{#4});
    \filldraw[pattern=north east lines, pattern color=code]
        ($#1+(\shifthist+\histwidth,\vtwo{#2}+\vtwo{#3})$) rectangle +(\histwidth,\vtwo{#4});

    \foreach[evaluate={\height=(\truevone{#2}+\truevone{#3}+\truevone{#4})}] \x in {0}{
    \node at ($#1+(\shifthist+0.5*\histwidth,\height*\scalehist)$) [above]
                                            {\pgfmathroundtozerofill{\height}\pgfmathresult};
                                        }
    \foreach[evaluate={\height=(\truevtwo{#2}+\truevtwo{#3}+\truevtwo{#4})}] \x in {0}{
    \node at ($#1+(\shifthist+1.5*\histwidth,\height*\scalehist)$) [above]
                                            {\pgfmathroundtozerofill{\height}\pgfmathresult};
                                        }
}
\histone{(\w{1}+0.1,0)}{1};
\draw[fill=code] (0,0) rectangle +(\w{1},\h{1});
\draw[<->] (0,-\shiftlabel) --node[midway,below]{\truew{1}} +(\w{1},0);
\draw[<->] (-\shiftlabel,0) --node[midway,left]{\trueh{1}} +(0,\h{1});

\tikzset{shift={(4.5,0)}};
\histone{(\w{2}+0.1,0)}{2};
\draw[fill=code] (0,0) rectangle +(\w{2},\h{2});
\draw[<->] (0,-\shiftlabel) --node[midway,below]{\truew{2}} +(\w{2},0);
\draw[<->] (-\shiftlabel,0) --node[midway,left]{\trueh{2}} +(0,\h{2});

\tikzset{shift={(4.5,0)}};
\histone{(\w{3}+0.1,0)}{3};
\draw[fill=code] (0,0) rectangle +(\w{3},\h{3});
\draw[<->] (0,-\shiftlabel) --node[midway,below]{\truew{3}} +(\w{3},0);
\draw[<->] (-\shiftlabel,0) --node[midway,left]{\trueh{3}} +(0,\h{3});

\tikzset{shift={(4.5,0)}};
\histone{(\w{4}+0.1,0)}{4};
\draw[fill=code] (0,0) rectangle +(\w{4},\h{4});
\draw[<->] (0,-\shiftlabel) --node[midway,below]{\truew{4}} +(\w{4},0);
\draw[<->] (-\shiftlabel,0) --node[midway,left]{\trueh{4}} +(0,\h{4});
\tikzset{shift={(-13.5,-5)}};
\draw[thick] (0,0) rectangle +(\scaleitem,\scaleitem);
\col{1};
\draw[fill=code] (0,0) rectangle +(\w{1},\h{1});
\col{3};
\draw[fill=code] (\scaleitem,\scaleitem) rectangle +(-\w{3},-\h{3});
\col{4};
\draw[fill=code] (\scaleitem,\scaleitem-\h{3}) rectangle +(-\w{4},-\h{4});
\node at (\scaleitem-0.8,-0.25) [below] {\Large{Configuration 1}};
\histthree{(\scaleitem+0.1,0)}{1}{3}{4}

\tikzset{shift={(6,0)}};
\draw[thick] (0,0) rectangle +(\scaleitem,\scaleitem);
\col{1};
\draw[fill=code] (0,0) rectangle +(\w{1},\h{1});
\col{2};
\draw[fill=code] (\w{1},0) rectangle +(\w{2},\h{2});
\node at (\scaleitem-0.8,-0.25) [below] {\Large{Configuration 2}};
\histtwo{(\scaleitem+0.1,0)}{1}{2}

\tikzset{shift={(6,0)}};
\draw[thick] (0,0) rectangle +(\scaleitem,\scaleitem);
\col{2};
\draw[fill=code] (0,0) rectangle +(\w{2},\h{2});
\col{4};
\draw[fill=code] (\scaleitem,\scaleitem) rectangle +(-\w{4},-\h{4});
\node at (\scaleitem-0.8,-0.25) [below] {\Large{Configuration 3}};
\histtwo{(\scaleitem+0.1,0)}{2}{4}
\end{tikzpicture}

\caption{Top row contains four \geomvecdim{2}{2} items (histograms show vector dimensions).
Three configurations are shown below it.
Placing all items into one bin would violate both geometric and
vector constraints. Placing the green item in
configuration 3 violates vector constraints.}
\label{fig:22-config}
\end{figure}

\section{Simple Algorithms}
\label{sec:simple-algos}
\label{sec:span-pack}

In this section, we look at simple algorithms for \geomvec{2}{$d$} BP.
They are based on Steinberg's algorithm for rectangle packing \cite{steinberg1997strip},
which has the following corollary:

\begin{lemma}[Proof in \cref{sec:omitted}]
\label{corr:steinberg}
Let $I$ be a set of rectangles where $a(I) \le 1$.
Then there is a $O(n\log^2 n/\log\log n)$-time algorithm to pack $I$ into
3 square bins of side length 1.
\end{lemma}

Let $I$ be a \geomvec{2}{$d$} BP instance.
Let $\Ihat \defeq \{\Span(i): i \in I\}$, i.e., $\Ihat$ is a 1-D BP instance.
The algorithm $\simplePack(I)$ first runs the Next-Fit algorithm on $\Ihat$.
Let $[\Jhat_1,$ $\Jhat_2,$ $\ldots,$ $\Jhat_m]$
be the resulting bin packing of $\Ihat$ into $m$ bins.
For each $\Jhat_j \subseteq \Ihat$, let $J_j$ be the corresponding items from $I$.
Then $\forall k \in [d_v],$ $v_k(J_j) \le 1$ and $\vol(J_j) \le 1$.
$\simplePack$ then uses Steinberg's algorithm to pack each $J_j$ into at most 3 bins,
giving a packing of $I$ into at most $3m$ bins.
By the property of Next-Fit (see \cref{claim:next-fit} in \cref{sec:next-fit}),
we get that $m \le \ceil{2\Span(I)}$.
By \cref{lem:span-lb-opt}, we get $3\ceil{2\Span(I)} \le 6(d+1)\opt(I)$.
This gives us the following theorem.
\begin{theorem}
\label{thm:span-pack}
For \geomvec{2}{$d$} BP, $\simplePack$ uses at most $3\ceil{2\Span(I)}$ bins,
so $\simplePack$ is a $6(d+1)$-approximation algorithm.
$\simplePack$ runs in $O(nd + n\log^2 n/\log\log n)$ time.
\end{theorem}

The algorithm $\betterSimplePack$ first computes $\Itild$,
which is a $(d+1)$-D VBP instance obtained by
replacing the geometric dimensions of each item $i \in I$ by a single vector dimension $a(i)$.
It computes a bin packing of $\Itild$ using any algorithm $\Acal$.
It then uses Steinberg's algorithm to obtain a bin packing
of $I$ into at most $3|\Acal(\Itild)|$ bins.

Note that $\opt(\Itild) \le \opt(I)$. If $\Acal$ has AAR $\alpha$,
then $|\Acal(\Itild)| \le \alpha\opt(\Itild) + O(1)$.
Therefore, $\betterSimplePack$ has AAR $3\alpha$.
The $(d+1)$-D VBP algorithm by \cite{bansal2009new}
(parametrized by a constant $\eps > 0$) gives $\alpha = 1+\ln(d+1)+\eps$
and the algorithm by \cite{BansalE016} gives $\alpha = 1.5+\ln((d+2)/2)+\eps$
(improves to $\alpha = 1+\ln(1.5)+\eps$ for $d=1$).

Using similar ideas, we can get an algorithm for \geomvec{2}{$d$} KS
(see \cref{sec:simple-ks}).

Although $\simplePack$ has a worse AAR than $\betterSimplePack$,
the number of bins used by $\simplePack$ is upper-bounded in terms of $\Span$,
which is a useful property. Because of this,
we will use it as a subroutine in other algorithms (like $\cbPack$).

The algorithms for \geomvec{2}{$d$} packing can be extended to \geomvec{$d_g$}{$d_v$} packing.
We just need an algorithm for the following problem:
\textsl{given a set $J$ of $d_g$-dimensional cuboids where $\vol(J) \le 1$,
pack $J$ into a small number of bins.}
We used Steinberg's algorithm when $d_g = 2$.
When $d_g = 3$, we can use the algorithm of \cite[Section 2]{diedrich2008approximation}
to pack $J$ into at most 9 bins.
For $d_g > 3$, we can use a variant of the $\hdhIV$ algorithm \cite{caprara2008}
to pack $J$ into at most $4^{d_g} + 2^{d_g}$ bins (details in \cref{sec:hdhIV}).
Therefore, $\simplePack$ will use
$b\ceil{2\Span(I)}$ bins, where $b \defeq 3$ when $d_g=2$,
$b \defeq 9$ when $d_g = 3$,
and $b \defeq 4^{d_g}+2^{d_g}$ when $d_g > 3$.
Therefore, $\simplePack$ is $2b(d_v+1)$-approximate.
Similarly, $\betterSimplePack$ has AAR $b(1 + \ln(d_v + 1) + \eps)$.

\section{Round-and-Approx Framework}
\label{sec:rna}

We enhance the R\&A framework as a general outline for designing approximation algorithms
for the generalized bin packing problem.
We denote the algorithm for the R\&A framework as $\rnaPack(I, \beta, \eps)$,
which takes as input a set $I$ of \geomvecdim{$d_g$}{$d_v$} items
and parameters $\beta \ge 1$ and $\eps \in (0, 1)$. The steps of the algorithm are as follows
(see \cref{algo:rna-pack} for a more formal description).

\begin{algorithm}[!ht]
\caption{$\rnaPack(I, \beta, \eps)$: Computes a bin packing of $I$.
$I$ is a set of \geomvecdim{$d_g$}{$d_v$} items and $\beta \ge 1$}
\label{algo:rna-pack}
\begin{algorithmic}[1]
\State $\xhat = \cLPsolve(I)$
\RepeatN{$T \defeq \ceil{(\ln\beta)\smallnorm{\xhat}_1}$}
    \State \label{alg-line:rna-pack:rround-choose}Select
        a \config{} $C$ with probability $\xhat_C/\smallnorm{\xhat}_1$.
    \State \label{alg-line:rna-pack:rround-pack}Pack a bin according to $C$.
\EndRepeatN
\State Let $S$ be the unpacked items from $I$.
    \Comment{$S$ is called the set of \emph{residual items}.}
\State Initialize $J_{\best}$ to \texttt{null}.
\For{$(\Itild, D) \in \round(I)$}
    \Comment{$\round(I)$ outputs a set of pairs.}
    \State $J_D = \simplePack(S \cap D)$
    \State Let $\pi$ be a bijection from $I-D$ to $\Itild$.
        Let $\Stild \defeq \{\pi(i): i \in S-D\}$.
    \State $\Jtild = \complexPack(\Stild)$
    \State $J = \unround(\Jtild)$
    \If{$J_{\best}$ is \texttt{null} or $|J_D \cup J| < |J_{\best}|$}
        \State $J_{\best} = J_D \cup J$
    \EndIf
\EndFor
\State Pack $S$ according to $J_{\best}$.
\end{algorithmic}
\end{algorithm}

\begin{enumerate}
\item \textbf{Solve the \Config{} LP} of $I$.
Let $\xhat$ be a $\mu$-asymptotic-approximate solution to the configuration LP.
Note that each index of $\xhat$ corresponds to a \config{}.
In all previous applications of the R\&A framework, $\mu = 1+\eps$.
However, in some of our applications we will have $\mu$ to be a large constant.

\item \textbf{Randomized rounding of \config{} LP}:
For $T \defeq \ceil{(\ln\beta)\smallnorm{\xhat}_1}$ steps do the following:
select a \config{} $C$ with probability $\xhat_C/\smallnorm{\xhat}_1$.
Pack $T$ bins according to each of these selected $T$ \config{}s.
Let $S$ be the remaining items that are not packed, called the \emph{residual instance}.

\item \textbf{Rounding of items}:
We define a subroutine $\round$ that takes items $I$ and parameter $\eps$ as input%
\footnote{The input to $\round$ is $I$ instead of $S$ because $S$ is random
and we want to round items deterministically, i.e.,
the rounding of each item $i \in S$ should
not depend on which other items from $I$ lie in $S$.
In fact, this is where the old R\&A framework \cite{Khan16} introduced an error.
See \cref{sec:rna-extra:bugfix} for details.}.
It \emph{discards} a set $D \subseteq I$ of items such that $\Span(D) \le \eps\Span(I)$
and then \emph{modifies} each item in $I-D$ to get a set $\Itild$ of items.
We say that the output of $\round(I, \eps)$ is $(\Itild, D)$, where
items in $\Itild$ are called \emph{rounded items}.
Intuitively, after rounding, the items in $\Itild$ are of $O(1)$ types,
which makes packing easier.
Also, since $\Span(D)$ is small, $D \cap S$ can be packed
into a small number of bins using $\simplePack$.

We impose some restrictions on $\round$, which we denote as conditions C1 and C2,
that we describe in \cref{sec:rna:round}.
Previous versions of R\&A only allowed modifications where
each item's dimensions were rounded up.
We do not have this restriction; we also allow rounding down some dimensions.
We also allow $\round$ to output a $O(\poly(n))$-sized list of guesses of $(\Itild, D)$.

\item \textbf{Pack rounded items}:
Let $\Stild$ be the rounded items corresponding to $S \setminus D$.
Pack $\Stild$ into bins using any bin packing algorithm that satisfies `condition C3',
which we describe in \cref{sec:rna:complex-pack}.
Let us name this algorithm $\complexPack$.

\item \textbf{Unrounding}:
Given a bin packing of $\Stild$, let $\unround$ be a subroutine that
computes a bin packing of $S \setminus D$.
$\unround$ is trivial in previous versions of R\&A, because they only increase dimensions
of items during rounding. In our applications, we may round down items,
so $\unround$ can be non-trivial.
$\unround$ can be any algorithm that satisfies `condition C4',
which we describe in \cref{sec:rna:unround}.

\end{enumerate}

We can think of the R\&A framework as a \emph{meta-algorithm}, i.e.,
we give it the algorithms $\round$, $\complexPack$ and $\unround$ as inputs
and it outputs the algorithm $\rnaPack$.
The R\&A framework requires that $\round$, $\complexPack$ and $\unround$
satisfy four conditions C1, C2, C3, C4, which we describe in
\cref{sec:rna:round,sec:rna:complex-pack,sec:rna:unround}.
Prospective users of the R\&A framework need to design these three subroutines
and prove that they satisfy these four conditions.

Intuitively, $\rnaPack$ first packs some items into $T$ bins by randomized rounding of $\xhat$.
We can prove that $\Pr[i \in S] \le 1/\beta$,
so $S$ contains a small fraction of the items in $I$
(see \cref{lem:pr-elem-in-residual}).
We will then try to prove that if the rest of the algorithm ($\round+\complexPack+\unround$)
packs $I$ into $m$ bins, then it will pack $S$ into roughly $m/\beta$ bins.
This notion was referred to in \cite{bansal2009new} as \emph{subset-obliviousness}.
We will use subset-obliviousness to bound the AAR of $\rnaPack$.

\begin{lemma}
\label{lem:pr-elem-in-residual}
$\forall i \in I$, $\Pr(i \in S)
\le \exp\left(-\frac{T}{\norm{\widehat{x}}_1}\right) \le \frac{1}{\beta}$.
\end{lemma}
\begin{proof}
Let $C_1, C_2, \ldots, C_T$ be the \config{}s chosen during randomized rounding
(line \ref{alg-line:rna-pack:rround-choose} in \cref{algo:rna-pack}).
Let $\mathcal{C}_i$ be the \config{}s that contain the element $i$.
\begin{align*}
\Pr(i \in S) &= \Pr\left(\bigwedge_{t=1}^T (C_t \not\in \mathcal{C}_i) \right)
= \prod_{t=1}^T \Pr(C_t \not\in \mathcal{C}_i)  \tag{all $C_t$ are independent}
\\ &= \prod_{t=1}^T \left(1 - \sum_{C \in \mathcal{C}_i} \Pr(C_t = C)\right)
= \left(1 - \sum_{C \in \mathcal{C}_i} \frac{\widehat{x}_C}{\norm{\widehat{x}}_1} \right)^T
\\ &\le \left(1 - \frac{1}{\norm{\widehat{x}}_1} \right)^T
\tag{constraint in \config{} LP for item $i$}
\\ &\le \exp\left( - \frac{T}{\norm{\widehat{x}_1}} \right)
\le \frac{1}{\beta}
\qedhere
\end{align*}
\end{proof}

\Cref{sec:rna:simple-pack} shows how to break $\simplePack$
into $\round$, $\complexPack$ and $\unround$
and use it with R\&A.

\subsection{Fractional Structured Packing}
\label{sec:rna:frac-pack}

Let $(\Itild, D)$ be an output of $\round(I)$ and let $\Xtild$
be an arbitrary subset of $\Itild$.
Our analysis of $\rnaPack$ is based around a concept called
\emph{fractional structured packing} of $\Xtild$.
Note that the notion of fractional structured packing only appears in the
analysis of $\rnaPack$. It is not needed to describe any algorithm.

First, we discuss the notion of structured packing.
Several types of packings are used in bin packing algorithms, such as
\emph{container-based} \cite{jansen2016new},
\emph{shelf-based} (e.g. output of Caprara's HDH algorithm \cite{caprara2008}),
\emph{guillotine-based} \cite{gilmore1965multistage},
\emph{corridor-based} \cite{l-packing}, etc.
A type of packing is called to be a {\em structured packing} if it satisfies \emph{downward closure}, i.e.,
\textsl{a structured packing remains structured even after removing some items from the packed bins}.
For example, Jansen and Pr\"adel \cite{jansen2016new} showed that given any packing of a
2-D GBP instance into $m$ bins, we can slice some of the items and repack them
into $(1.5+\eps)m + O(1)$ bins such that the resulting packing is \emph{container-based}.
\emph{Container-based} roughly means that in each bin,
items are packed into rectangular regions called containers,
and containers' heights and widths belong to a fixed set of $O(1)$ values.
Hence, \emph{container-based} is an example of a structured packing as
a container-based packing of bins remains a \emph{container-based} packing,
even if we remove some items from the packed bins.
In fact, all the commonly used packing types are indeed structured packings.
Also note that the set of all possible packings is trivially a structured packing.
Our R\&A framework gives algorithm designers the freedom to use or define
the structured packing in any way they want, as long as they satisfy the downward closure property.
Typically, the choice of which definition of structured packing to use will depend
on the ease of proving Conditions C2 and C3 for that definition.
This helped us to go beyond the R\&A framework of Bansal and Khan \cite{bansal2014binpacking},
which only considered container-based packings of $\Itild$.

Intuitively, a fractional structured packing is one where we slice
each item of $\Xtild$ into pieces and then find a structured packing of the pieces.
Define $\fsopt(\Xtild)$ as the number of bins used in the optimal
fractional structured packing of $\Xtild$.
To analyze the AAR of $\rnaPack$, we will bound the number of bins used to pack
the residual instance $S$ in terms of $\fsopt(\Stild)$,
and then we will bound $\fsopt(\Stild)$ in terms of $\opt(I)$.

To define fractional structured packing,
we first define what it means to slice an item.
From a geometric perspective, slicing an item perpendicular to the $k\Th$ dimension
means cutting the item into 2 parts using a hyperplane perpendicular to the $k\Th$ axis.
The vector dimensions get split proportionately across the slices.
E.g., for $d_g=2$, if $k=1$ for item $i$, then we slice $i$ using a vertical cut,
and if $k=2$, we slice $i$ using a horizontal cut.

\begin{definition}[Slicing an item]
Let $i$ be a \geomvecdim{$d_g$}{$d_v$} item.
Slicing $i$ perpendicular to geometric dimension $k$
with proportionality $\alpha$ (where $0 < \alpha < 1$)
is the operation of replacing $i$ by two items $i_1$ and $i_2$ such that:
(i) $\forall j \neq k, \ell_j(i) = \ell_j(i_1) = \ell_j(i_2)$,
(ii) $\ell_k(i_1) = \alpha \ell_k(i)$ and $\ell_k(i_2) = (1-\alpha)\ell_k(i)$,
(iii) $\forall j \in [d_v], v_j(i_1) = \alpha v_j(i) \wedge v_j(i_2) = (1-\alpha) v_j(i)$.
\end{definition}

\begin{definition}[Fractional packing]
Let $\Itild$ be \geomvecdim{$d_g$}{$d_v$} items, where for each item $i \in \Itild$,
we are given a set $X(i)$ of axes perpendicular to which we can repeatedly slice $i$
($X(i)$ can be empty, which would mean that the item cannot be sliced).
If we slice items as per their given axes and then pack the slices into bins,
then the resulting packing is called a \emph{fractional} bin packing.
\end{definition}
See \cref{fig:frac-pack} for an example of a fractional packing.

\begin{figure}[htb]
\centering
\begin{tikzpicture}[
myarrow/.style = {->,>={Stealth},semithick},
mybrace/.style = {decoration={amplitude=3pt,brace,mirror,raise=1pt},semithick,decorate},
every node/.style = {scale=0.8},
scale=0.8
]
\draw (0,0) rectangle +(3,3);
\draw[fill={black!25}] (0,0) rectangle +(1.2,3);
\draw[fill={black!10}] (3,0) rectangle +(-1.5,1.2);
\draw[fill={black!10}] (3,1.2) rectangle +(-1.5,1.2);
\draw[mybrace] (0,0) -- node[below=1pt] {0.4} +(1.2,0);
\draw[mybrace] (1.5,0) -- node[below=1pt] {0.5} +(1.5,0);
\draw[mybrace] (3,0) -- node[right=2pt] {0.4} +(0,1.2);
\draw[mybrace] (3,1.2) -- node[right=2pt] {0.4} +(0,1.2);

\draw[fill={black!25}] (-8,0) rectangle +(1.2,3);
\path (-8,0) -- node[below=0pt] {0.4} +(1.2,0);
\path (-8,0) -- node[left=0pt] {1} +(0,3);
\node at (-6.2,1.5) {+};
\draw[fill={black!10}] (-5,1) rectangle +(3,1.2);
\path (-5,1) -- node[left=0pt] {0.4} +(0,1.2);
\draw[dashed] (-3.5,2.5) -- (-3.5,0.5);
\node[rotate=90,transform shape] at (-3.5,0.3) {\large\ding{34}};
\draw[mybrace] (-5,1) -- node[below=1pt] {0.5} +(1.5,0);
\draw[mybrace] (-3.5,1) -- node[below=1pt] {0.5} +(1.5,0);
\draw[myarrow] (-1.5,1.5) -- (-0.5,1.5);
\end{tikzpicture}

\caption{Example of a fractional packing of two items into a bin.}
\label{fig:frac-pack}
\end{figure}
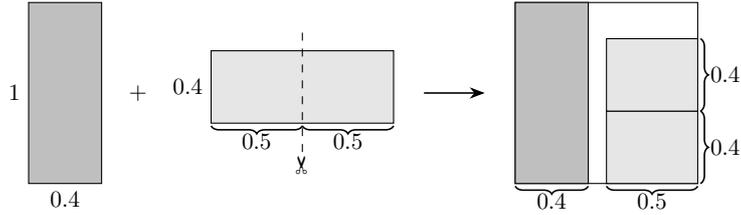

\subsection{Properties of \texorpdfstring{$\round$}{round}}
\label{sec:rna:round}

\begin{definition}[Density vector]
The density vector of a \geomvecdim{$d_g$}{$d_v$} item $i$ is the vector
$v_{\Span} \defeq \left[v_0(i)/{\Span(i)}, {v_1(i)}/{\Span(i)}, \ldots, {v_{d_v}(i)}/{\Span(i)}\right]$.
Recall that $v_0(i) \defeq \vol(i)$
and note that $\|v_{\Span}\|_{\infty} = 1$.
\end{definition}

The subroutine $\round(I)$ returns a set of pairs of the form $(\Itild, D)$.
\textbf{Condition C1} is defined as the following constraints over each pair $(\Itild, D)$:
\begin{itemize}
\item {\bf C1.1.} {\em Small discard:} $D \subseteq I$ and $\Span(D) \le \eps\Span(I)$.
\item {\bf C1.2.} {\em Bijection from $I-D$ to $\Itild$:} Each item in $\Itild$ is obtained by modifying an item in $I-D$.
    Let $\pi$ be the bijection from $I-D$ to $\Itild$ that captures this relation.
\item  {\bf C1.3.} {\em Homogeneity properties:} $\round$ partitions items in $\Itild$ into a constant number of classes:
    $\Ktild_1, \Ktild_2, \ldots, \Ktild_q$. These classes should satisfy the following properties,
    which we call \emph{homogeneity} properties:
    \begin{itemize}
    \item All items in a class have the same density vector.
    \item For each class $\Ktild_j$, we decide the set $X$ of axes perpendicular to which we can
        slice items in $\Ktild_j$. If items in a class $\Ktild_j$ are not allowed to be
        sliced perpendicular to dimension $k$, then all items in that class have
        the same length along dimension $k$.
        (For example, if $d_g=2$ and vertical cuts are forbidden, then all items
        have the same width.)
    \end{itemize}
\item  {\bf C1.4.} {\em Bounded expansion:} Let $C$ be any \config{} of $I$ and $\Ktild$ be any one of the classes of $\Itild$.
    Let $\Ctild \defeq \{\pi(i): i \in C - D\}$. Then we need to prove that
    $\Span(\Ktild \cap \Ctild) \le c_{\max}$ for some constant $c_{\max}$.
    Let us call this result `bounded expansion lemma'.
\end{itemize}

Intuitively, the homogeneity properties allow us to replace (a slice of) an item
in a fractional packing by slices of other items of the same class.
Thus, while trying to get a fractional packing, we can focus on the item classes,
which are constant in number, instead of focusing on the $n$ items.
Intuitively, the bounded expansion lemma ensures that we do not round up items too much.

\textbf{Condition C2} (also called \emph{structural theorem}):
For some constant $\rho > 0$ and for some $(\Itild, D) \in \round(I)$,
$\fsopt(\Itild) \le \rho\opt(I) + O(1)$.

Intuitively, the structural theorem says that allowing slicing as per $\round$
and imposing a structure on the packing
does not increase the minimum number of bins by too much.
The AAR of $\rnaPack$ increases with $\rho$, so we want $\rho$ to be small.

\subsection{\texorpdfstring{$\complexPack$}{complex-pack}}
\label{sec:rna:complex-pack}

\textbf{Condition C3}:
For some constant $\alpha > 0$ and for any $(\Itild, D) \in \round(I)$
and any $\Xtild \subseteq \Itild$, $\complexPack(\Xtild)$ packs $\Xtild$ into at most
$\alpha\fsopt(\Xtild) + O(1)$ bins.

Condition C3 says that we can find a packing that is
close to the optimal fractional structured packing.
The AAR of $\rnaPack$ increases with $\alpha$, so we want $\alpha$ to be small.

\subsection{\texorpdfstring{$\unround$}{unround}}
\label{sec:rna:unround}

\textbf{Condition C4}:
For some constant $\gamma > 0$, if $\complexPack(\Stild)$ outputs a packing
of $\Stild$ into $m$ bins, then $\unround$ modifies that packing
to get a packing of $S-D$ into $\gamma m + O(1)$ bins.

Intuitively, condition C4 says that unrounding does not increase the number of bins by too much.
The AAR of $\rnaPack$ increases with $\gamma$, so a small $\gamma$ is desirable.
If $\round$ only increases the dimensions of items,
then unrounding is trivial and $\gamma = 1$.

\subsection{AAR of R\&A}
\label{sec:rna:aar}

Recall that $\simplePack$ is a $2b(d_v+1)$-approximation algorithm for \geomvec{$d_g$}{$d_v$} BP
(see \cref{sec:span-pack}), where $b \defeq 3$ when $d_g=2$,
$b \defeq 9$ when $d_g = 3$,
$b \defeq 4^{d_g}+2^{d_g}$ when $d_g > 3$,
and $b \defeq 1$ when $d_g = 1$.
Our key ingredient in the analysis of R\&A is the following lemma.
We give a sketch of the proof here and defer the full proof to \cref{sec:rna-extra}.

\begin{restatable}{lemma}{rthmFoptConc}
\label{lem:fopt-conc}
Let $\Stild$ be as computed by $\rnaPack(I, \beta, \eps)$.
Then with high probability, we get
\[ \fsopt(\Stild) \le \fsopt(\Itild)/\beta + 2b\mu\epsLP\opt(I) + O(1/\epsLP^2). \]
\end{restatable}
\begin{proof}[Proof sketch]
Our proof of \cref{lem:fopt-conc} is inspired by the analysis in \cite{Khan16}.
We prove it by analyzing the \emph{fractional structured} \config{} LP of $\Itild$.

\begin{definition}
Let $(\Itild, D) \in \round(I)$. Suppose $\round$ partitioned $\Itild$ into classes
$\Ktild_1, \Ktild_2, \ldots \Ktild_q$.
Let $\mathcal{C}_f$ be the set of all structured \config{}s of items in $\Itild$
that allow items to be sliced as per $\round$.
For any $\Stild \subseteq \Itild$, the fractional structured \config{} LP of $\Stild$,
denoted as $\fscLP(\Stild)$, is
\[ \begin{array}{*3{>{\displaystyle}l}}
\min_{x \in \mathbb{R}^{|\mathcal{C}_f|}}
    & \multicolumn{2}{>{\displaystyle}l}{\sum_{C \in \mathcal{C}_f} x_C}
\\ \textrm{where} & \sum_{C \in \mathcal{C}_f} \Span(C \cap \Ktild_j)x_C
    \;\ge\; \Span(\Stild \cap \Ktild_j) & \forall j \in [q]
\\ & x_C \ge 0 & \forall C \in \mathcal{C}_f
\end{array} \]
The integer version of this program is denoted as $\fscIP(\Stild)$.
The optimal objective values of $\fscLP(\Stild)$ and $\fscIP(\Stild)$
are denoted as $\fscLP^*(\Stild)$ and $\fscIP^*(\Stild)$, respectively.
\end{definition}

Intuitively, $\fscIP$ is the same as the structured fractional bin packing problem
because of the downward closure property, so $\fscIP^*(\Stild) = \fsopt(\Stild)$%
\footnote{There is a caveat in the definition of $\fscIP^*(\Stild)$,
because of which this isn't completely true.
See \cref{lem:fcip-eq-fopt} in \cref{sec:rna-extra} for an explanation and a fix.}.
By the homogeneity property (C1.3), the number of constraints in this LP is a constant $q$.
So by rank lemma\footnote{Rank Lemma: the number of non-zero variables in an
extreme-point solution to an LP is at most the number of constraints.
See Lemma 2.1.4 in \cite{iterative-methods}.},
we can show that $|\fsopt(\Stild) - \fscLP^*(\Stild)| \in O(1)$
(see \cref{lem:fcip-eq-fopt,lem:fclp-rank-lemma} in \cref{sec:rna-extra}).
Now to prove \cref{lem:fopt-conc}, roughly, we need to show that
$\fscLP^*(\Stild) \lessapprox \fscLP^*(\Itild)/\beta$.

The RHS in the $j\Th$ constraint of $\fscLP(\Stild)$ is a random variable
$\Span(\Stild \cap \Ktild_j)$.
The RHS in the $j\Th$ constraint of $\fscLP(\Itild)$ is $\Span(\Ktild_j)$.
Note that $\E(\Span(\Stild \cap \Ktild_j)) \le \Span(\Ktild)/\beta$
by \cref{lem:pr-elem-in-residual}.
In fact, we can harness the randomness of $\Stild$,
the bounded expansion property (C1.4),
and a concentration inequality \cite{mcdiarmid1989method},
to show that $\Span(\Stild \cap \Ktild) \lessapprox \Span(\Ktild)/\beta$.
Therefore, if $x^*$ is an optimal solution to $\fscLP(\Itild)$,
then $x^*/\beta$ is \emph{roughly} a solution to $\fscLP(\Stild)$,
which implies $\fscLP^*(\Stild) \lessapprox \fscLP^*(\Itild)/\beta$.
\end{proof}

\begin{theorem}
\label{thm:rna-pack}
With high probability, the number of bins used by $\rnaPack(I, \beta, \eps)$ is at most
\[ \left( (\ln\beta)\mu + \frac{\gamma\alpha\rho}{\beta}
    + 2b(d_v+1+\gamma\alpha\mu)\eps \right) \opt(I) + O(1/\eps^2). \]
\end{theorem}
\begin{proof}
Let $J_{\LP}$ be the set of bins packed in the {\em randomized rounding of
configuration LP} step
(see line \ref{alg-line:rna-pack:rround-pack} in \cref{algo:rna-pack}
in \cref{sec:rna-extra}),
$J_D$ be the set of bins used to pack the discarded items $D \cap S$,
$J$ be the set of bins used to pack the rest of the items $S \setminus D$,
 and $\Jtild$ be the set of bins used by $\complexPack$ to pack items in $\Stild$.

Then $|J_{\LP}| \le T = \smallceil{(\ln\beta)\smallnorm{\xhat}_1} \le (\ln\beta)\mu\opt(I) + O(1)$.

Now, we have $|J_D| \le b\ceil{2\Span(D)} \le 2b\eps\Span(I) + b \le 2b(d_v+1)\eps\opt(I) + b$.
The first inequality follows from the property of $\simplePack$, the second follows from
C1.1 (Small Discard) and the last follows from \cref{lem:span-lb-opt}. Finally,
\begin{align*}
|J| &\le \gamma|\Jtild| + O(1)  \tag{property of $\unround$ (C4)}
\\ &\le \gamma\alpha\fsopt(\Stild) + O(1)  \tag{property of $\complexPack$ (C3)}
\\ &\le \gamma\alpha\left( {\fsopt(\Itild)}/{\beta} + 2b\mu\eps\opt(I) \right) + O(1/\eps^2)
    \tag{by \cref{lem:fopt-conc}}
\\ &\le \gamma\alpha\left(\rho/\beta + 2b\mu\eps\right)\opt(I) + O(1/\eps^2)
\end{align*}
Here, the last inequality follows from the structural theorem (C2),
which says that $\exists (\Itild, D) \in \round(I)$ such that
$\fsopt(\Itild) \le \rho\opt(I) + O(1)$.
Hence, the total number of bins is at most
\begin{align*}
& |J_{\LP}| + |J_D| + |J|
\le \left( (\ln\beta)\mu + \frac{\gamma\alpha\rho}{\beta}
    + 2b(d_v+1+\gamma\alpha\mu)\eps \right)\opt(I) + O(1/\eps^2).
\qedhere \end{align*}
\end{proof}

The AAR of $\rnaPack(I)$ is roughly $\mu\ln\beta + \gamma\alpha\rho/\beta$.
This is minimized for $\beta = \gamma\alpha\rho/\mu$ and the minimum value is
$\mu\left(1 + \ln\left(\alpha\gamma\rho/\mu\right)\right)$.
As we require $\beta \ge 1$, we get this AAR only when $\gamma\alpha\rho \ge \mu$.
If $\mu \ge \gamma\alpha\rho$, the optimal $\beta$ is 1 and the AAR is roughly $\gamma\alpha\rho$.

\subsection{Example: \texorpdfstring{$\simplePack$}{simple-pack}}
\label{sec:rna:simple-pack}

We will show how to use $\simplePack$ with the R\&A framework.
Recall that $\simplePack$ is a $2b(d_v+1)$-approximation algorithm for \geomvec{$d_g$}{$d_v$} BP
(see \cref{sec:span-pack}), where $b \defeq 3$ when $d_g=2$,
$b \defeq 9$ when $d_g = 3$,
$b \defeq 4^{d_g}+2^{d_g}$ when $d_g > 3$,
and $b \defeq 1$ when $d_g = 1$.
Using the R\&A framework on $\simplePack$ will
improve its AAR from $2b(d_v+1)$ to $b(1+\ln(2(d_v+1))) + O(\eps)$.
To do this, we need to show how to implement $\round$, $\complexPack$ and $\unround$.

\begin{enumerate}
\item $\cLPsolve(I)$:
Using the KS algorithm of \cref{sec:simple-ks} and the LP algorithm of \cite{eku-pst},
we get a $b(1+\eps)$-approximate solution to $\cLP(I)$. Therefore, $\mu = b(1+\eps)$.

\item $\round(I)$:
returns just one pair: $(\Itild, \{\})$, where $\Itild \defeq \{\pi(i): i \in I\}$
and $\pi(i)$ is an item having height ($d_g\Th$ geometric dimension) equal to $\Span(i)$,
all other geometric dimensions equal to 1, and all vector dimensions equal to $\Span(i)$.
There is just one class in $\Itild$ and all items are allowed to be sliced
perpendicular to the height, so the homogeneity properties are satisfied.
Also, $c_{\max} = d_v+1$ by \cref{lem:span-lb-opt}.

\item \textbf{Structural theorem}:
We take structured packing to be the set of all possible packings.
Then $\fsopt(\Itild) = \ceil{\Span(I)} \le (d_v+1)\opt(I)$. Therefore, $\rho = d_v+1$.

\item $\complexPack(\Stild)$:
We can treat $\Stild$ as the 1-D bin packing instance $\{\Span(i): i \in S\}$
and pack it using Next-Fit. Therefore,
$|\complexPack(\Stild)| \le 2\Span(S) + 1 \le 2\fsopt(\Stild) + 1$.
Therefore, $\alpha = 2$.

\item $\unround(\Jtild)$:
For each bin in $\Jtild$, we can pack the corresponding unrounded items into $b$ bins
(using Steinberg's algorithm or $\hdhIV$). Therefore, $\gamma = b$.
\end{enumerate}

Therefore, we get an AAR of
$\mu(1+\ln(\gamma\alpha\rho/\mu)) + O(\eps) \approx b(1+\ln(2(d_v+1))) + O(\eps)$.

For $d_g = 2$, we can slightly improve the AAR by using the
$(2+\eps)$-approximation algorithm of \cite{gvks} for \geomvec{2}{$d_v$} KS.
This gives us an AAR of $2(1 + \ln(3(d_v+1))) + O(\eps)$.
This is better than the AAR of $\betterSimplePack$ for $d_v \ge 3$.

The above example is presented only to illustrate an easy use of the R\&A framework.
It doesn't exploit the full power of the R\&A framework.
The algorithm $\cbPack$, which we outline in \cref{sec:improve},
uses more sophisticated subroutines $\round$, $\complexPack$ and $\unround$,
and uses a more intricate definition of fractional structured packing
to get an even better AAR of $2(1+\ln(\frac{d+4}{2})) + \eps$
(improves to $2(1+\ln(19/12))+\eps$ $\approx 2.919+\eps$ for $d=1$).

\section{Improved Approximation Algorithms}
\label{sec:improve}

In this section, we give an overview of the $\cbPack$ algorithm for \geomvec{2}{$d$} BP,
which is inspired from the
1.5 asymptotic approximation algorithm for 2-D GBP \cite{pradel-thesis}.
$\cbPack$ is based on the following two-step procedure, as common in many packing algorithms.

In the first step ({\em structural step}), we start with an optimal solution
and using some transformation show the existence of
a good {\em structured} solution. %
Our structural result roughly says that if items $I$
can be packed into $m$ bins, then we can round $I$ to get a new instance $\Itild$
such that $\fsopt(\Itild) \le \rho m + O(1)$ for some constant $\rho$,
where $\fsopt(\Itild)$ is the number of bins in the optimal \emph{structured fractional}
packing of $I$. Roughly, the notion of \emph{structured} packing that we use here,
which we call \emph{compartmental packing},
imposes the following additional constraints over the
\emph{container-based} packing of \cite{pradel-thesis}:
\begin{itemize}
\item An item $i$ is said to be dense iff $v_{\max}(i)/a(i)$ is above a certain threshold.
If a bin contains dense items, then we reserve a sufficiently-large rectangular region
exclusively for dense items, and dense items can only be packed into this region.
\item For a constant $\eps$, for every $j \in [d]$,
the sum of the $j\Th$ weights of items in each bin is at most $1-\eps$.
\end{itemize}
The proof of our structural result differs significantly from that of
\cite{pradel-thesis} because the presence of vector dimensions inhibit a
straightforward application of their techniques.

In the second step ({\em algorithmic step}), we show how to find a near-optimal structured solution.
$\cbPack$ first rounds $I$ to $\Itild$ and then
uses a mixture of brute-force and LP to pack the items into containers,
similar to \cite{pradel-thesis} or \cite{kenyon1996strip}.
This gives us the optimal structured fractional packing of $\Itild$.
Then we convert that packing to a non-fractional packing of $I$ with
only a tiny increase in the number of bins.
Intuitively, rounding is helpful here because it reduces the number of different types of items
by partitioning the items into a constant number of classes
such that items in each class are \emph{similar}.
This simplicity introduced by rounding enables us to
find the optimal structured fractional packing efficiently.

Then we show that $\cbPack$ fits into the R\&A framework and gives an AAR of roughly $2(1+\ln(\rho/2))$.
In particular, \cref{sec:rbbp-2d-extra:rna} shows how components from $\cbPack$
can be used as the R\&A subroutines $\round$, $\complexPack$ and $\unround$.
To (approximately) solve the \config{} LP, we use the linear programming algorithm
from \cite{eku-pst} and the $(2+\eps)$-approximation algorithm for
\geomvec{2}{$d$} KS from \cite{gvks}.

We now give a brief overview of some key ideas used in our structural result.
Due to space limitations, the details of the structural result and the algorithm
can be found in \cref{sec:rbbp-2d-extra}.
Like \cite{pradel-thesis},
we start with a packing of input $I$ into $m$ bins,
and transform it into a structured fractional packing of $\Itild$
into $\rho m + O(1)$ bins.
We do this in three steps:
\begin{enumerate}
\item %
We round up one geometric dimension of each item and pack the items into
roughly $\rho m + O(1)$ bins. We call these bins \emph{quarter-structured}
(see \cref{sec:rbbp-2d-extra:one-side,sec:rbbp-2d-extra:slack}).
\item %
We round the remaining dimensions of items and partition them into classes such that
they satisfy the \emph{homogeneity properties} (see \cref{sec:rna:round}).
We allow slicing and repack the items into almost the same number of bins.
We call the resulting bin packing \emph{semi-structured}
(see \cref{sec:rbbp-2d-extra:round-weights,sec:rbbp-2d-extra:other-side}).
\item %
Finally, we transform the packing into a \emph{\compartmentalHyp} packing
(see \cref{sec:rbbp-2d-extra:compart}).
Compartmental packings have nice properties which help us to find them efficiently.

\end{enumerate}

In steps 1, 2 and 3, \cite{pradel-thesis} uses the Next-Fit-Decreasing-Height (NFDH)
algorithm \cite{coffman1980performance} to pack items of $O(\eps m)$ area into $O(\eps m)$ bins.
This does not work when vector dimensions are present; an item of low area can have large weights.
In step 2, \cite{pradel-thesis} uses \emph{linear grouping}, i.e.,
each item is moved in place of a geometrically larger item so that it can be rounded up.
Vector dimensions make such cross-bin movement difficult,
since that can violate bins' weight capacities.
\cite{pradel-thesis} uses cross-bin movement in step 1 too.

Our first crucial observation is that most difficulties associated with vector dimensions
disappear if items' \emph{density} is upper-bounded by a constant. Here density of item $i$
is defined as $v_{\max}(i)/a(i)$.
Specifically, if items of bounded density (we call them non-dense items) have
small area, then we can use $\simplePack$
to pack them into a small number of bins.
Linear grouping can also be made to work for such items with some more effort.
Therefore, our strategy is to segregate items as \emph{dense} and \emph{non-dense}.
We reserve a thin rectangular region in bins for dense items,
and the rest is reserved for non-dense items.

Furthermore, dense items in a bin must have low total area, due to their high density.
If we reserve enough space for them in the bin, we can always pack them in their
reserved region using NFDH (see \cref{lem:nfdh-strip} in \cref{sec:next-fit}).
Such a guarantee means that we can essentially ignore their geometric dimensions
and simply treat them as vectors.

In step 2, we want to round up vector dimensions with only a marginal increase in
the number of bins. To do this, we require each quarter-structured bin to be $\eps$-slacked.
$\eps$-slackness roughly means that for a set $J$ of items in a bin,
$\forall j \in [d], v_j(J) \le 1-\eps$ (see \cref{sec:rbbp-2d-extra:slack}
for a formal description).
$\eps$-slackness also helps us in designing the packing algorithm,
because we can then use techniques from resource-augmented vector bin packing.
Also, during the rounding step, we round down the weight of some dense items,
and $\eps$-slackness allows us to \emph{unround} with no increase in the number of bins.

The observations above guide our definition of \emph{quarter-structured}.
Roughly, a packing is quarter-structured %
if items having large width have their width and $x$-coordinate
rounded to a multiple of %
$\epsLarge^2/4$ and %
each bin is $\eps$-slacked. We reserve a thin rectangular region
of width $\epsLarge/2$ for packing dense items (only if the bin contains dense items).

\begin{figure}[t]
\centering
\definecolor{better-red}{rgb}{0.97,0.73,0.7}
\definecolor{better-green}{rgb}{0.7,0.95,0.7}
\definecolor{better-blue}{rgb}{0.65,0.7,0.99}
\begin{tikzpicture}[
mybrace/.style = {decoration={amplitude=3pt,brace,mirror,raise=1pt},semithick,decorate},
greenfill/.style = {fill=better-green},
bluefill/.style = {fill=better-blue},
redfill/.style = {fill=better-red},
scale=0.8
]
\draw[greenfill] (0,0) rectangle +(2, 1);
\draw[greenfill] (0,1) rectangle +(0.8,0.8);
\draw[greenfill] (0,2) rectangle +(1,1);
\draw[greenfill] (0.8,1.5) rectangle +(1,0.5);
\draw[bluefill] (3,3) rectangle +(-0.15,-1.2);
\draw[bluefill] (2.85,3) rectangle +(-0.15,-1.5);
\draw[bluefill] (3,0) rectangle +(-0.5,0.6);
\draw[redfill] (0.8,1) rectangle (3,1.5);
\draw[redfill] (1.5,2) rectangle +(1.2,1);

\draw (0, 0) rectangle (3, 3);
\draw[mybrace] (2.4,0) -- node[below=2pt] {{\small$\eps$}} (3,0);
\draw[dashed,semithick] (2.4,-0.3) -- +(0,3.6);
\end{tikzpicture}

\caption{Example of classifying items as blue, red and green based on an $\eps$-strip.}
\label{fig:eps-strip}
\end{figure}
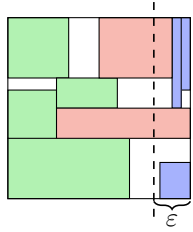

In step 1, \cite{pradel-thesis} uses a standard cutting-strip argument:
They create a strip of width $\epsLarge$ next to an edge of the bin
(see \cref{fig:eps-strip}).
Items lying completely inside the strip (called blue items),
have small area and are packed separately using NFDH.
Items intersecting the boundary of the strip (called red items), are removed.
This creates an empty space of width $\epsLarge$ in the bin.
Using this empty space, items lying outside the strip (called green items),
can then have their width and $x$-coordinate rounded to a multiple of $\epsLarge^2/2$.
Their key idea is how to pair up most bins so that red items from two bins
can be rounded and packed together into a new bin.
This is roughly why they get an AAR of $1.5+\eps$.

We use the cutting-strip argument too, but with some differences.
We cannot freely mix red items from different bins if they have large weight,
and we cannot simply pack blue items into a small number of bins.
We also need bins to be slacked. So, we get a larger AAR of $d+4+\eps$.
For $d=1$, however, we allow mixing items using more sophisticated techniques,
which improves the AAR to $19/6+\eps$.
Also, we round green items to a multiple of $\epsLarge^2/4$ instead of $\epsLarge^2/2$,
which leaves an empty strip of width $\epsLarge/2$ in the bin
even after rounding, and we reserve this space for dense items.
This gives us a quarter-structured packing.

Based on the broad ideas above, we make more changes to the quarter-structured packing
to get a compartmental packing.

Finally, in the algorithmic step, we use a combination of brute-force and LP
to pack different types of items into different type of compartments and find a
{\em good} compartmental packing efficiently.

\acknowledgements{We thank Nikhil Bansal and Thomas Rothvoss for their helpful comments.}

\bibliographystyle{plainurl}
\bibliography{bibdb}
\appendix
\section{Some Simple Packing Algorithms}
\label{sec:next-fit}

\begin{claim}[Next-Fit]
\label{claim:next-fit}
Let $I \defeq \{i_1, i_2, \ldots, i_n\}$ be a 1D bin packing problem instance.
The Next-Fit algorithm finds a packing of $I$ into at most $\ceil{2\sum_{j=1}^n i_j}$ bins.
\end{claim}

\begin{lemma}[NFDH for strip packing \cite{coffman1980performance}]
\label{lem:nfdh-strip}
A set $I$ of rectangles can be packed into a strip of height at most
$2a(I) + max_{i \in I} h(i)$ using the Next-Fit Decreasing Height (NFDH) algorithm.
\end{lemma}

\begin{lemma}[NFDH for small items \cite{coffman1980performance}]
\label{lem:nfdh-small}
Let $I$ be a set of rectangles where each rectangle has width at most $\delta_W$
and height at most $\delta_H$. Let there be a rectangular bin of width $W$ and height $H$.
If $a(I) \le (W - \delta_W)(H - \delta_H)$, then NFDH can pack $I$ into the bin.
\end{lemma}

\section{Omitted Proofs}
\label{sec:omitted}

\begin{lemma}[Steinberg's algorithm \cite{steinberg1997strip}]
\label{thm:steinberg}
Let $I$ be a set of rectangles. Let $w_{\max} \defeq \max_{i \in I} w(i)$
and $h_{\max} \defeq \max_{i \in I} h(i)$. Consider a bin of width $W$ and height $H$,
where $w_{\max} \le W$ and $h_{\max} \le H$.
Then there is a $O(n\log^2 n/\log\log n)$-time algorithm to pack $I$ into the bin if
\[ 2a(I) \le WH - \Big(\big(\max(2w_{\max} - W, 0))\cdot (\max(2h_{\max} - H, 0)\big)\Big) \]
\end{lemma}

\begin{omittedproof}{corr:steinberg}
Pack $I$ into a bin of width 2 and height 1 using Steinberg's algorithm
(see \cref{thm:steinberg}).
Then cut the bin vertically into 2 unit-sized squares.
The items which lie completely inside the left half can be packed into a unit-sized bin.
The items which lie completely inside the right half can be packed into a unit-sized bin.
The items which lie on the cutting line are stacked one-over-the-other,
so we can pack them into a unit-sized bin.
\end{omittedproof}

\section{The \texorpdfstring{$\hdhIV$}{HDH4} Algorithm}
\label{sec:hdhIV}

\begin{lemma}
\label{lem:hdh4}
Let $I$ be a set of $d_g$-dimensional cuboids.
Then $I$ can be packed into at most $4^{d_g} + 2^{d_g}\vol(I)$ bins
using a variant of $\hdhIV$.
\end{lemma}
\begin{proof}
Define $f_4: [0, 1] \mapsto [0, 1]$ and $\type_4: [0, 1] \mapsto [4]$ as
\begin{align*}
f_4(x) &= \begin{cases}
\frac{1}{q} & x \in \left(\left.\frac{1}{q+1}, \frac{1}{q}\right]\right. \textrm{ for } q \in [3]
\\ 2x & x \le \frac{1}{4} \end{cases}
& \type_4(x) &= \begin{cases}
q & x \in \left(\left.\frac{1}{q+1}, \frac{1}{q}\right]\right. \textrm{ for } q \in [3]
\\ 4 & x \le \frac{1}{4} \end{cases}
\end{align*}
Note that $x \le f_4(x) \le 2x$.
Let $i$ be a $d_g$-dimensional cuboid.
Define $f_4(i)$ as the cuboid of length $f_4(\ell_j(i))$ in the $j\Th$ dimension.
Therefore, $\vol(i) \le \vol(f_4(i)) \le 2^{d_g}\vol(i)$.
For a set $I$ of cuboids, $f_4(I) \defeq \{f_4(i): i \in I\}$.
Define $\type(i)$ to be a $d$-length vector whose $j\Th$ component is $\type(\ell_j(i))$.

For a set $I$ of cuboids, define $I^{(k)}$ to be the set of items
obtained by ignoring all dimensions other than the first $k$.
\cite{caprara2008} (implicitly) defines a recursive algorithm $\hdhIVunit^{[t]}(I, d)$.
For a sequence $I$ of $d$D cuboids all having the same $\type_4$ $t$, where
$\vol(f_4((I-\{\last(I)\})^{(d)})) < 1$,
$\hdhIVunit^{[t]}(I, d)$ returns a packing of $I^{(d)}$ into a $d$D bin.
Here $\last(I)$ is the last item in sequence $I$.

Now I'll define a slight variation of the $\hdhIV$ algorithm in \cite{caprara2008}.
This variant also works when rotation of items is allowed.
First, partition the items by $\type_4$. The number of partitions is at most $Q = 4^{d_g}$.
Let $I^{[q]}$ be the partition containing items of $\type_4$ $q$.
Order the items in $I^{[q]}$ arbitrarily.
Then repeatedly pick the smallest prefix $J$ of $I^{[q]}$ such that
either $J = I^{[q]}$ or $\vol(f_4(J)) \ge 1$, and pack $J$ into a bin using
$\hdhIVunit^{[q]}(J, d_g)$.

Suppose $\hdhIV(I^{[q]})$ produces $m^{[q]}$ bins. Let $B_j^{[q]}$ be the $j\Th$ of these bins.
Given the way we choose prefixes, $\vol(f_4(B_j^{[q]})) \ge 1$ for $j \in [m^{[q]}-1]$,
i.e. at most 1 bin is partially-filled.
\[ \vol(f_4(I^{[q]})) = \sum_{j=1}^{m^{[q]}} \vol(f_4(B_j^{[q]})) \ge (m^{[q]}-1) \]
Total number of bins used is
\[ \sum_{q=1}^Q m^{[q]} \le \sum_{q=1}^Q (1 + \vol(f_4(I^{[q]})))
\le Q + \vol(f_4(I)) = 4^{d_g} + 2^{d_g} \vol(I) \qedhere\]
\end{proof}

\section{Algorithm for Knapsack}
\label{sec:simple-ks}

Let $I$ be a set of \geomvecdim{2}{$d$} items. Let $p(i)$ be the profit of item $i$.
We want to pack a max-profit subset of $I$ into a bin.
Let $\Ihat$ be a set of $(d+1)$-dimensional vectors obtained by replacing the
geometric dimensions of each item $i$ by a single vector dimension $a(i)$.
Let $\Acal$ be a $(d+1)$-D VKS algorithm having approximation ratio $\alpha \ge 1$.
$\Acal$ gives us a packing of items $\Jhat \subseteq \Ihat$ into a bin.
Let $J$ be the corresponding items in $I$.
Then $\vol(J) \le 1$ and $\forall k \in [d], v_k(J) \le 1$.
Use Steinberg's algorithm to compute a packing of $J$ into 3 bins: $J_1, J_2, J_3$.
\WLoG, assume $p(J_1) \ge p(J_2) \ge p(J_3)$. Then output the packing $J_1$
as the answer to the \geomvec{2}{$d$} KS problem.
Since any feasible solution to the \geomvec{2}{$d$} KS instance $(I, p)$ also gives a
feasible solution to the VKS instance $(\Ihat, p)$, we get $\opt(\Ihat, p) \ge \opt(I, p)$.
Since $\Acal$ is $\alpha$-approximate, we get $p(J) \ge \opt(\Ihat, p)/\alpha$. Hence,
\[ p(J_1) \ge \frac{p(J)}{3}
\ge \frac{\opt(\Ihat, p)}{3\alpha}
\ge \frac{\opt(I, p)}{3\alpha} \]
Therefore, we get a $3\alpha$-approximation algorithm for \geomvec{2}{$d$} KS.
Using the PTAS for $(d+1)$-D VKS by Frieze and Clarke \cite{FriezeClarke84},
we get $\alpha = 1+\eps$.

\section{Details of the R\&A Framework}
\label{sec:rna-extra}

\begin{lemma}
\label{lem:fcip-eq-fopt}
$\fsopt(\Stild) \le \fscIP^*(\Stild) \le \fsopt(\Stild) + q$.
\end{lemma}
\begin{proof}
Due to the downward closure property, changing inequality constraints to equality constraints
doesn't affect the optimum values of the above LP and IP.
Therefore, $\fscIP(\Stild)$ is equivalent to the fractional structured bin packing problem.

A problem with the above definition of $\fscLP(\Itild)$ is that
the number of variables can be infinite if certain classes allow slicing.
We circumvent this problem by \emph{discretizing} the \config{}s:
Let $\delta$ be the smallest dimension of any item, i.e.
$\delta \defeq \min\left(\min_{j=1}^{d_g} \ell_j(i), \min_{j=1}^{d_v} v_j(i)\right)$.

In any optimal integral solution to $\fscLP(\Itild)$ that uses $m$ bins,
we can slice out some items from each class in each bin so that
the $\Span$ of each class in each bin is a multiple of $\delta^{d_g}/n$.
In each class, the total size of sliced out items across all bins is at most $\delta^{d_g}$.
Therefore, for each class, slices of that class can fit into a single item of that class.
If each such single item is packed in a separate bin,
the total number of bins used is at most $m + q$.

Therefore, we only need to consider \config{}s where either the $\Span$ of each class
is a multiple of $\delta^{d_g}/n$ or there is a single item in the \config.
This gives us a finite number of \config{}s and completes the proof.
\end{proof}

\begin{lemma}
\label{lem:fclp-rank-lemma}
$\fscLP^*(\Stild) \le \fscIP^*(\Stild) \le \fscLP^*(\Stild) + q$.
\end{lemma}
\begin{proof}
By rank lemma, the number of positive-valued variables in an extreme-point solution
to a linear program is at most the number of constraints (other than the
variable non-negativity constraints).

Thus, an optimal extreme-point solution to $\fscLP(\Stild)$ has at most $q$
positive-valued variables. Rounding up those variables to the nearest integer
will give us an integral solution and increase the objective value by at most $q$.
Hence, $\fscIP^*(\Stild) \le \fscLP^*(\Stild) + q$.
\end{proof}

Let $\cLP(I)$ denote the \config{} LP of items $I$ and let $\cLP^*(I)$ denote the
optimal objective value of $\cLP(I)$.
Recall that $\simplePack$ is a $2b(d_v+1)$-approximation algorithm for \geomvec{$d_g$}{$d_v$} BP
(see \cref{sec:span-pack}), where $b \defeq 3$ when $d_g=2$,
$b \defeq 9$ when $d_g = 3$,
$b \defeq 4^{d_g}+2^{d_g}$ when $d_g > 3$,
and $b \defeq 1$ when $d_g = 1$.

\begin{lemma}
\label{lem:clp-vs-span}
For a set $I$ of \geomvecdim{$d_g$}{$d_v$} items,
$\cLP^*(I) \in \Theta(\Span(I)) + O(1)$.
\end{lemma}
\begin{proof}
Let $A$ be the \config{} matrix of $I$.
Let $x^*$ be the optimal solution to $\cLP(I)$.
In $\cLP(I)$, the constraint for item $i$ gives us
$\sum_{C \in \mathcal{C}} A[i, C] x^*_C \ge 1$.
Multiplying each constraint by $\Span(i)$ and adding these constraints together, we get
\begin{align*}
\Span(I) &\le \sum_{C \in \mathcal{C}} \sum_{i \in I} \Span(i) A[i, C] x^*_C
= \sum_{C \in \mathcal{C}} \Span(C) x^*_C
\\ &\le (d_v+1) \sum_{C \in \mathcal{C}} x^*_C = (d_v+1)\cLP^*(I).
\end{align*}
Therefore, $\cLP^*(I) \ge \Span(I)/(d_v+1)$.
We also have
\[ \cLP^*(I) \le \opt(I) \le |\simplePack(I)| \le 2b\Span(I)+b. \]
Therefore, $\cLP^*(I) \in \Theta(\Span(I)) + O(1)$.
\end{proof}

\begin{lemma}[Independent Bounded Difference Inequality \cite{mcdiarmid1989method}]
\label{lem:ind-bdi}
Let $X \defeq [X_1, X_2, \ldots, X_n]$ be random variables with $X_j \in A_j$.
Let $\phi: \prod_{i=1}^n A_j \mapsto \mathbb{R}$ be a function such that
$\abs{\phi(x) - \phi(y)} \le c_j$
whenever vectors $x$ and $y$ differ only in the $j\Th$ coordinate.
Then for any $t \ge 0$,
\[ \Pr[ \phi(X) - \E(\phi(X)) \ge t ] \le \exp\left(-\frac{2t^2}{\sum_{j=1}^n c_j^2}\right). \]
\end{lemma}

\begin{lemma}
\label{lem:fclp-conc}
Let $\Stild$ be as computed by $\rnaPack(I, \beta, \eps)$. Let $\epsLP \in (0, 1)$ be a constant.
When $\Span(I)$ is large compared to $1/\epsLP^2$, we get that with high probability
\[ \fscLP^*(\Stild) \le \frac{\fscLP^*(\Itild)}{\beta} + 2b\mu\epsLP\opt(I) + O(1). \]
\end{lemma}
\begin{proof}
Let $y \in \mathcal{C}^T$ be the \config{}s chosen during randomized rounding.
When viewed as a vector of length $T$, all coordinates of $y$ are independent.
Define $\uncovered(y) \defeq I - \bigcup_{t=1}^T y_t$.

Let $\Ktild_1, \Ktild_2, \ldots, \Ktild_q$ be the classes of $\Itild$.
Let $\pi$ be the bijection from $I-D$ to $\Itild$.
For a set $X \subseteq I$, define $\Itild[X] \defeq \{\pi(i): i \in X-D\}$.
For $j \in [q]$, define $\phi_j \in \mathcal{C}^T \mapsto \mathbb{R}_{\ge 0}$ as
\[ \phi_j(y) \defeq \Span\left(\Ktild_j \cap \Itild[\uncovered(y)]\right). \]
For any set $X \subseteq I$, define $g_j(X) \defeq \Span(\Ktild_j \cap \Itild[X])$.
Then $\phi_j(y) = g_j(\uncovered(y))$ and $g_j$ is a non-negative additive function.

Let $y^{(1)}, y^{(2)} \in \mathcal{C}^T$ such that
$y^{(1)}$ and $y^{(2)}$ differ only in coordinate $t$.
Let $y^{(1)}_t = C_1$ and $y^{(2)}_t = C_2$.
Let $S_1 = \uncovered(y^{(1)})$ and $S_2 = \uncovered(y^{(2)})$.

It is easy to see (using Venn diagrams) that
$S_1 - S_2 \subseteq C_2 - C_1$ and $S_2 - S_1 \subseteq C_1 - C_2$.
\begin{align*}
\abs{\phi_j(y^{(1)}) - \phi_j(y^{(2)})}
&= \abs{g_j(S_1) - g_j(S_2)}
\\ &= \abs{g_j(S_1 - S_2) - g_j(S_2 - S_1)}  \tag{additivity of $g_j$}
\\ &\le \max(g_j(S_1 - S_2), g_j(S_2 - S_1))
\\ &\le \max(g_j(C_2), g_j(C_1))
\\ &\le \max_{C \in \mathcal{C}} \Span(\Ktild_j \cap \Itild[C])
\le c_{\max}.  \tag{by bounded expansion lemma}
\end{align*}
\begin{align*}
\E(\phi_j(y)) &= \E(g_j(S))
\\ &= \sum_{i \in \Itild} g_j(\{i\}) \Pr(i \in S)
    \tag{linearity of expectation and additivity of $g_j$}
\\ &\le \sum_{i \in \Itild} g_j(\{i\}) (1/\beta)
    \tag{by \cref{lem:pr-elem-in-residual}}
\\ &= \frac{g_j(\Itild)}{\beta} = \frac{\Span(\Ktild_j)}{\beta}.
\end{align*}
$\forall j \in [q]$, define $Q_j$ as the smallest prefix of $\Stild \cap \Ktild_j$
such that either $Q_j = \Stild \cap \Ktild_j$ or $\Span(Q_j) \ge \epsLP\norm{\widehat{x}}_1/q$.
Define $Q \defeq \bigcup_{j=1}^q Q_j$. Therefore,
\[ \Span(Q) \le \epsLP \norm{\widehat{x}}_1 + q \le \epsLP \mu \opt(I) + O(1). \]
\begin{align*}
\fscLP^*(\Stild) &\le \fscLP^*(\Stild - Q) + \fscLP^*(Q)
\\ &\le \fscLP^*(\Stild - Q) + b(2\Span(Q) + 1)  \tag{by \cref{sec:span-pack}}
\\ &\le \fscLP^*(\Stild - Q) + 2b\mu\epsLP\opt(I) + O(1).
\end{align*}
Now we will try to prove that with high probability,
$\fscLP^*(\Stild - Q) \le \fscLP^*(\Itild)/\beta$.

If $Q_j = \Stild \cap \Ktild_j$, then $\Span(\Ktild_j \cap (\Stild - Q)) = 0$.
Otherwise,
\begin{align*}
& \Pr\left[ \Span(\Ktild_j \cap (\Stild - Q)) \ge \frac{\Span(\Ktild_j)}{\beta} \right]
= \Pr\left[ \Span(\Ktild_j \cap \Stild) - \frac{\Span(\Ktild_j)}{\beta} \ge \Span(Q_j) \right]
\\ &\le \Pr\left[ \phi_j(y) - \E(\phi_j(y)) \ge \frac{\epsLP}{q} \norm{\widehat{x}}_1 \right]
\le \exp\left( -\frac{2}{Tc_{\max}^2} \left(\frac{\epsLP}{q}\norm{\widehat{x}}_1\right)^2\right)
\tag{\cref{lem:ind-bdi}}
\\ &\le \exp\left( -\frac{2\epsLP^2}{\ln(\beta)c_{\max}^2q^2} \norm{\widehat{x}}_1\right).
\end{align*}
Therefore, by union bound, we get
\[ \Pr\left[ \bigvee_{j=1}^q \left(\Span(\Ktild_j \cap (\Stild - Q))
    \ge \frac{\Span(\Ktild_j)}{\beta}\right) \right]
\le q\exp\left( -\frac{2\epsLP^2}{\ln(\beta)c_{\max}^2q^2} \norm{\widehat{x}}_1\right). \]
Since $\cLP^*(I) \le \norm{\widehat{x}}_1 \le \mu \cLP^*(I) + O(1)$,
and $\cLP^*(I) \in \Theta(\Span(I)) + O(1)$ (by \cref{lem:clp-vs-span}),
we get $\norm{\widehat{x}}_1 \in \Theta(\Span(I)) + O(1)$.
When $\Span(I)$ is very large compared to $1/\epsLP^2$, we get
that with high probability, $\forall j \in [q]$,
\[ \Span(\Ktild_j \cap (\Stild - Q)) \le \frac{\Span(\Ktild_j)}{\beta}. \]
Let $x^*$ be the optimal solution to $\fscLP(\Itild)$.
Then with high probability, $x^*/\beta$ is a feasible solution to $\fscLP(\Stild - Q)$.
Therefore,
\begin{align*}
\fscLP^*(\Stild)
&\le \fscLP^*(\Stild - Q) + 2b\mu\epsLP \opt(I) + O(1)
\\ &\le \fscLP^*(\Itild)/\beta + 2b\mu\epsLP \opt(I) + O(1).
\qedhere \end{align*}
\end{proof}

\rthmFoptConc*
\begin{proof}
When $\Span(I)$ is very large compared to $1/\epsLP^2$, we get
\begin{align*}
\fsopt(\Stild) &\le \fscIP^*(\Stild) + O(1)  \tag{by \cref{lem:fcip-eq-fopt}}
\\ &\le \fscLP^*(\Stild) + O(1)  \tag{by \cref{lem:fclp-rank-lemma}}
\\ &\le \fscLP^*(\Itild)/\beta + 2b\mu\epsLP \opt(I) + O(1)  \tag{by \cref{lem:fclp-conc}}
\\ &\le \fsopt(\Itild)/\beta + 2b\mu\epsLP \opt(I) + O(1).  \tag{by \cref{lem:fcip-eq-fopt}}
\end{align*}
Otherwise, if $\Span(I) \in O(1/\epsLP^2)$, we get
\begin{align*}
\fsopt(\Stild) &\le \rho\opt(I) + O(1)  \tag{by structural theorem}
\\ &\le \rho|\simplePack(I)| + O(1)
\\ &\le \Theta(\Span(I)) + O(1) \le O(1/\epsLP^2).  \tag{by \cref{sec:span-pack}}
\end{align*}
\end{proof}

\subsection{Error in Previous R\&A Framework}
\label{sec:rna-extra:bugfix}

Here we describe a minor error in the R\&A framework of \cite{Khan16},
and how it can be fixed.

We define $(\Itild, D)$ as an output of $\round(I)$ and for the residual instance $S$,
we define $\Stild$ as the corresponding rounded items of $S - D$.
Our proof of \cref{lem:fopt-conc} relies on the fact that for any subset
of rounded items, the $\Span$ reduces by a factor of at least $\beta$
if we restrict our attention to the residual instance.
Formally, this means that for any $\Xtild \subseteq \Itild$, we have
\[ \E(\Span(\Xtild \cap \Stild))
= \sum_{i \in \Xtild} \Span(i)\Pr(i \in \Stild)
\le \Span(\Xtild)/\beta. \]
The equality follows from linearity of expectation and the fact that
$\Span(i)$ is deterministic, i.e., it doesn't depend on the randomness
used in the randomized rounding of the \config{} LP.
This is because $\round$ is not given any information about what $S$ is.
The inequality follows from \cref{lem:pr-elem-in-residual},
which says that $\Pr(i \in S) \le 1/\beta$.

The R\&A framework of \cite{Khan16} used similar techniques in their analysis.
In their algorithm, however, they round items differently.
Specifically, they define a subroutine $\round$ and define $\Itild \defeq \round(I)$
and $\Stild \defeq \round(S)$.
They, too, claim that for any subset of rounded items,
the $\Span$ reduces by a factor of at least $\beta$
if we restrict our attention to the residual instance.
While their claim is correct for input-agnostic rounding
(where items are rounded up to some constant size collection
values chosen independent of the problem instance),
the claim is incorrect for input-sensitive rounding
(where the values are chosen based on the specific problem instance).
So the claim is incorrect if $\round$ is not deterministic, as
then an item can be rounded differently depending on
 different residual instances.

In fact, they use their R\&A framework with the algorithm of
Jansen and Pr\"adel \cite{jansen2016new}, which uses
\emph{linear grouping} (along with some other techniques) for rounding.
Linear grouping rounds items in an \emph{input-sensitive} way, i.e.,
the rounding of each item depends on the sizes of items in $S$
which is a random subset.

\section{Algorithm for \geomvec{2}{\texorpdfstring{$d$}{d}} Bin Packing}
\label{sec:rbbp-2d-extra}

Here we will see a rounding-based algorithm for \geomvec{2}{$d$} bin packing,
called $\cbPack$ (named after `compartment-based packing').
$\cbPack$ will be parametrized by a parameter $\eps$,
where $\eps^{-1} \in 2\mathbb{Z}$ and $\eps \le 1/8$.

\begin{claim}
\label{claim:wmax-s-lim}
If a set $J$ of items fits into a bin, then $v_{\max}(J) \le d$ and $\Span(J) \le d+1$.
\end{claim}

We will first remove some items from the bin, so that the remaining items
can be classified into useful categories.

\subsection{Classifying Items}
\label{sec:rbbp-2d-extra:classify}

\begin{definition}
For constants $\epsSmall < \epsLarge$, a bin packing instance $I$
is called $(\epsSmall, \epsLarge)$-non-medium iff $\forall i \in I$,
$(w(i) \not\in (\epsSmall, \epsLarge])
\wedge (h(i) \not\in (\epsSmall, \epsLarge])
\wedge (\forall j \in [d], v_j(i) \not\in (\epsSmall, \epsLarge])$.
\end{definition}

An $(\epsSmall, \epsLarge)$-non-medium instance has useful properties.
Therefore, we want to remove some items from the input instance $I$
so that it becomes $(\epsSmall, \epsLarge)$-non-medium and the removed items
can be packed into a small number of bins.

\begin{definition}
\label{defn:remmed}
Let $\delta_0, \eps \in (0, 1]$ be constants and let
$f: (0, 1] \mapsto (0, 1]$ be a function such that $\forall x \in (0, 1], f(x) < x$.
Let $T \defeq \ceil{(d+2)/\eps}$.
For $t \in [T]$, define $\delta_t \defeq f(\delta_{t-1})$ and define
\[ J_t \defeq \left\{i \in I: w(i) \in (\delta_t,\delta_{t-1}] \vee h(i) \in (\delta_t,\delta_{t-1}]
    \vee \left(\bigvee_{j=1}^d v_j(i) \in (\delta_t,\delta_{t-1}]\right)\right\}. \]
Define $\remMed(I, \eps, f, \delta_0)$ as the tuple $(J_r, \delta_r, \delta_{r-1})$,
where $r \defeq \argmin_{t=1}^T \Span(J_t)$.
\end{definition}

\begin{lemma}
\label{thm:med-span}
Let $(\Imed, \epsSmall, \epsLarge) \defeq \remMed(I, \eps, f, \delta_0)$.
Then $\Span(\Imed) \le \eps\Span(I)$.
\end{lemma}
\begin{proof}
Each item belongs to at most $d+2$ sets $J_t$. Therefore,
\[ \Span(\Imed) = \min_{t=1}^T \Span(J_t)
\le \frac{1}{T} \sum_{t=1}^T \Span(J_t)
\le \frac{d+2}{\ceil{(d+2)/\eps}} \Span(I)
\le \eps \Span(I). \qedhere \]
\end{proof}
By \cref{defn:remmed}, $I - \Imed$ is $(\epsSmall, \epsLarge)$-non-medium.
By \cref{thm:med-span,thm:span-pack}, $\Span(\Imed)$ can be packed into
at most $6(d+1)\eps\opt(I) + 3$ bins using the $\simplePack$ algorithm.

We will choose $f$ to be independent of $I$, so $\epsLarge$ and $\epsSmall$ are constants.
Also note that $\epsSmall \defeq f(\epsLarge)$ and $\epsLarge \le \delta_0$.
We choose ${\delta_0 \defeq \min\left(1/(4d+1), 2\eps/3\right) }$,
so $\delta_0^{-1} \in \mathbb{Z}$.
We will choose $f$ \hyperref[eqn:remmed-f]{later}. For now, we will impose these conditions:
$f(x) \le \eps x^2/2$, and $(x^{-1} \in \mathbb{Z} \implies f(x)^{-1} \in \mathbb{Z})$.
The second condition implies that $\epsLarge^{-1}, \epsSmall^{-1} \in \mathbb{Z}$.

Henceforth, assume all \geomvec{2}{$d$} bin packing instances
to be $(\epsSmall, \epsLarge)$-non-medium.

\begin{definition}
We can classify item $i$ by its geometric dimensions as follows:
\begin{tightemize}
\item Big item: $w(i) > \epsLarge$ and $h(i) > \epsLarge$.
\item Wide item: $w(i) > \epsLarge$ and $h(i) \le \epsSmall$.
\item Tall item: $w(i) \le \epsSmall$ and $h(i) > \epsLarge$.
\item Small item: $w(i) \le \epsSmall$ and $h(i) \le \epsSmall$.
\end{tightemize}
\end{definition}
When rotation of items is allowed, assume \wLoG{} that there are no tall items in $I$.

\begin{definition}[Dense items]
Item $i$ is dense iff either $a(i) = 0$ or $v_{\max}(i)/a(i) > 1/\epsLarge^2$.
\end{definition}
Note that big items cannot be dense.
\begin{definition}[Heavy and light items]
\label{defn:heavy-light}
A dense item $i$ is said to be heavy in vector dimension $j$ iff $v_j(i) \ge \epsLarge$.
Otherwise $i$ is said to be light in dimension $j$.
If $i$ is heavy in some dimension, then $i$ is said to be heavy, otherwise $i$ is light.
\end{definition}

\subsection{Rounding One Side}
\label{sec:rbbp-2d-extra:one-side}

In this subsection, we will show how a packing of $I$ into bins can be modified
to get a more structured packing where one of the geometric dimensions is rounded up.

\begin{definition}
In a bin, assume a coordinate system where $(0, 0)$ is at the bottom left
and $(1, 1)$ is on the top right. We define the following regions, called \emph{strips}:
\begin{tightemize}
\item $S^{(T)} \defeq [0, 1] \times [1-\epsLarge, 1]$
    and $S^{(T')} \defeq [0, 1] \times [1-\epsLarge/2, 1]$.
\item $S^{(B)} \defeq [0, 1] \times [0, \epsLarge]$.
\item $S^{(L)} \defeq [0, \epsLarge] \times [0, 1]$.
\item $S^{(R)} \defeq [1-\epsLarge, 1] \times [0, 1]$
    and $S^{(R')} \defeq [1-\epsLarge/2, 1] \times [0, 1]$.
\end{tightemize}
We say that an item intersects a strip iff a non-zero volume of that item lies inside the strip.
\end{definition}

\begin{property}
\label{prop:hrnd}
A bin is said to satisfy \cref{prop:hrnd} iff both of these conditions hold:
\begin{enumerate}[a]
\item \label{item:hrnd:geom} The $x$-coordinate and width of all non-dense wide and big items
    is a multiple of $\epsLarge^2/4$.
\item \label{item:hrnd:dense} If the bin contains dense items, then dense items
are packed inside $S^{(R')}$ and no non-dense item intersects $S^{(R')}$.
\end{enumerate}
\end{property}

\begin{property}
\label{prop:vrnd}
A bin is said to satisfy \cref{prop:vrnd} iff both of these conditions hold:
\begin{enumerate}[a]
\item \label{item:vrnd:geom} The $y$-coordinate and height of all non-dense tall and big items
    is a multiple of $\epsLarge^2/4$.
\item \label{item:vrnd:dense} If the bin contains dense items, then dense items
are packed inside $S^{(T')}$ and no non-dense item intersects $S^{(T')}$.
\end{enumerate}
\end{property}

Equivalently, we can say that a bin satisfies \cref{prop:vrnd} iff
its mirror image about the line $y = x$ satisfies \cref{prop:hrnd}.

The main result of this subsection is the following:

\begin{lemma}
\label{lem:rnd1}
Given a packing of items into a bin, we can
round up the width of some wide and big non-dense items to a multiple of $\epsLarge^2/4$
or round up the height of some tall and big non-dense items to a multiple of $\epsLarge^2/4$
and get a packing into 2 bins and 2 boxes where:
\begin{itemize}
\item Each bin satisfies either \cref{prop:hrnd} or \cref{prop:vrnd}.
\item $v_1$ of each box is at most $1/2$.
\item One of the boxes has only dense items.
Its larger dimension is 1 and its smaller dimension is
$\deltaDense \defeq 2d\epsLarge^2 + \epsSmall$.
\item One of the boxes has only non-dense items.
Its larger dimension is 1 and its smaller dimension is $\epsLarge + \epsSmall$.
\item One of the boxes is horizontal, i.e., has width 1 and only contains wide and small items.
The other box is vertical, i.e., has height 1 and only contains tall and small items.
\end{itemize}
\end{lemma}

Before proving \cref{lem:rnd1}, we first prove a few ancillary results.

For $X \in \{T, B\}$, the items lying completely inside $S^{(X)}$ are either small or wide.
Let $C^{(X)}$ be the set of small and wide items that intersect $S^{(X)}$.
For $X \in \{L, R\}$, the items lying completely inside $S^{(X)}$ are either small or tall.
Let $C^{(X)}$ be the set of small and tall items that intersect $S^{(X)}$.
Since $2\epsLarge + \epsSmall \le 1$, we get
$C^{(T)} \cap C^{(B)} = C^{(L)} \cap C^{(R)} = \{\}$.

\WLoG, assume that $v_1(C^{(T)}) \le 1/2$ because $v_1(C^{(B)} \cup C^{(T)}) \le 1$
and if $v_1(C^{(T)}) > v_1(C^{(B)})$, then we can mirror-invert the bin along a horizontal axis.
Similarly assume that $v_1(C^{(R)}) \le 1/2$.

\begin{observation}
If a bin only contains tall and small items,
it trivially satisfies \cref{prop:hrnd}\ref{item:hrnd:geom}.
If a bin only contains wide and small items,
it trivially satisfies \cref{prop:vrnd}\ref{item:vrnd:geom}.
\end{observation}

\begin{lemma}
\label{lem:right-shift}
Suppose we're given a packing of items into a bin such that no item intersects $S^{(R)}$.
Then we can increase the widths of all wide and big items to a multiple of $\epsLarge^2/4$
and repack the items so that they satisfy \cref{prop:hrnd}\ref{item:hrnd:geom}
and no item intersects $S^{(R')}$.
\end{lemma}
\begin{proof}
Let $y_b(i)$ and $y_t(i)$ be the $y$-coordinates of the bottom and top edge
respectively of item $i$. If an item $j$ intersects the strip
$[0, 1] \times [y_b(i), y_t(i)]$ and lies to the right of $i$
(i.e. the left edge of $j$ is to the right of the right edge of $i$),
we say that $i \precImm j$
(see \cref{fig:prec-imm}).
Let $\preceq$ denote the reflexive and transitive closure of the relation $\precImm$.
It is easy to see that $\preceq$ is a partial ordering of $I$.
Define $i \prec j$ as $i \preceq j \wedge i \neq j$.

\begin{figure}[!ht]
\centering
\begin{tikzpicture}
\newcommand*{\yAlo}{0.5}
\newcommand*{\yAhi}{1.5}
\draw (0,0) rectangle (4,3.5);
\draw[fill={black!5}] (0.5,\yAlo) rectangle (1.3,\yAhi) node[pos=0.5] {$A$}
    (1.5,0.7) rectangle (2.3,2.8) node[pos=0.5] {$B$}
    (2.5,2.3) rectangle (3.5,3) node[pos=0.5] {$C$}
    (2.5,0.8) rectangle (3.5,1.7) node[pos=0.5] {$D$};
\draw[gray,opacity=0.3] (0,\yAlo) -- (4,\yAlo)
    (0,\yAhi) -- (4,\yAhi);
\end{tikzpicture}

\caption{Items $A$, $B$, $C$ and $D$ in a bin.
Here $A \precImm D$ but $A \not\precImm C$.
Also, $A \precImm B \precImm C$, so $A \preceq C$.}
\label{fig:prec-imm}
\end{figure}
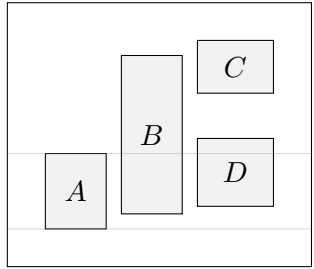

Define $p_w(i)$ to be 1 if it is wide or big and to be 0 if it is neither wide
nor big. Also, define $n_w(i) \defeq p_w(i) + \max_{j \prec i} n_w(j)$
(if there is no $j \prec i$, define $\max_{j \prec i} n_w(i) \defeq 0$).
Intuitively, $n_w(i)$ denotes the length of the largest chain of wide items preceding $i$.
The $x$-coordinate of the right edge of item $i$ is more than $\epsLarge n_w(i)$.
Therefore, $n_w(i) < 1/\epsLarge - 1$.

\begin{transformation}
\label{trn:right-shift}
Move each item $i$ to the right by $(n_w(i) - p_w(i))\epsLarge^2/2$.
Additionally, if $i$ is wide or big, move it further to the right so that the $x$-coordinate
of its left edge is a multiple of $\epsLarge^2/4$, and increase its width so that it
is a multiple of $\epsLarge^2/4$.
\end{transformation}

On applying \cref{trn:right-shift} to item $i$, the $x$-coordinate of its right
edge increases by less than $n_w(i)\epsLarge^2/2$.
Since $n_w(i) < 1/\epsLarge - 1$, the increase is less than $\epsLarge/2$.
Therefore, $i$ will not intersect $S^{(R')}$ after this transformation.
Also, after applying this transformation to all items, the bin satisfies
\cref{prop:hrnd}\ref{item:hrnd:geom}.

We will now prove that after applying \cref{trn:right-shift} to all items, no items overlap.
If $i$ and $j$ are not relatively ordered by $\preceq$,
they cannot overlap because we only moved items rightwards.
Now assume \wLoG{} that $i \prec j$.
The $x$-coordinate of the right edge of $i$ increases by less than $n_w(i)\epsLarge^2/2$.
The $x$-coordinate of the left edge of $j$ increases by at least $(n_w(j) - p_w(j))\epsLarge^2/2$.
Since $n_w(i) \le \max_{i' \prec j} n_w(i') = n_w(j) - p_w(j)$,
$i$ and $j$ don't overlap.
\end{proof}

\begin{lemma}
\label{lem:dense-pack}
Let $R$ be a set of wide and small items that are dense and have total weight at most 1.
They can be packed in polynomial time into a box of width 1 and height
$\deltaDense \defeq 2d\epsLarge^2 + \epsSmall$.
\end{lemma}
\begin{proof}
$a(R) \le \epsLarge^2 v_{\max}(R) \le d\epsLarge^2$.
\\ So by \cref{lem:nfdh-strip}, height used by NFDH is less than
$2a(R) + \epsSmall \le 2d\epsLarge^2 + \epsSmall$.
\end{proof}
We can get an analogous result for tall and small dense items.

\begin{proof}[Proof of \cref{lem:rnd1}]
Suppose the bin contains items $J$. Then we can use \cref{lem:dense-pack} to move
dense wide items to box $D_W$ and move dense tall and small items to box $D_H$.
We will later repack one of $D_W$ and $D_H$ into a bin.

$v_1(D_W \cup D_H) \le 1$. This gives us 2 cases:
$v_1(D_W) \le 1/2$ or $v_1(D_H) \le 1/2$.
The first case is the same as the second one with the coordinate axes swapped,
so assume \wLoG{} that $v_1(D_W) \le 1/2$.

Move $C^{(R)}$ to a box of height 1 and width $\epsLarge + \epsSmall \le 1/2$.
$C^{(R)}$ only has tall and small non-dense items. Also, $v_1(C^{(R)}) \le 1/2$.

Let $I^{(R)}$ be the set of big and wide items that intersect $S^{(R)}$.
Move $I^{(R)}$ to a separate bin.
The items in $I^{(R)}$ are stacked on top of each other.
Therefore, we can round their widths to a multiple of $\epsLarge^2/4$.
$I^{(R)}$ doesn't have dense items.
Therefore, this new bin satisfies the desired properties.

Since we removed $C^{(R)}$ and $I^{(R)}$ from the bin, $S^{(R)}$ is empty.
By \cref{lem:right-shift}, we can round the $x$-coordinate and width of big and wide items
in the bin to a multiple of $\epsLarge^2/4$ and then repack the items in the bin so that
the bin satisfies \cref{prop:hrnd}\ref{item:hrnd:geom} and $S^{(R')}$ is empty.
Observe that
\[ \deltaDense = 2d\epsLarge^2 + \epsSmall
\le 2d\epsLarge^2 + \frac{\eps\epsLarge^2}{2}
\le \frac{\epsLarge}{2} (4d+1)\epsLarge \le \frac{\epsLarge}{2}. \]
Since $\deltaDense \le \epsLarge/2$, pack $D_H$ into $S^{(R')}$.
Now this bin also satisfies \cref{prop:hrnd}\ref{item:hrnd:dense}.
In total, we used 2 bins and 2 boxes ($C^{(R)}$ and $D_W$).
The dense box is horizontal and the non-dense box is vertical.
Refer to \cref{fig:split-bin} for an example.

\begin{figure}[!ht]
\centering
\begin{tikzpicture}[scale=1.25, transform shape,
dense/.style = {draw=black,fill={black!65}},
ndense/.style = {draw=black,fill={black!8}},
bigbrace/.style = {decoration={brace,amplitude=10pt},line width=1pt,decorate},
smallbrace/.style = {decoration={brace,mirror,raise=1pt},semithick,decorate},
bin1-dense-items/.pic = {
    \path[ dense] (0,2.95) rectangle +(0.1,-0.1);
    \path[ dense] (2,4) rectangle +(0.15,-0.1);
    \path[ dense] (0,0) rectangle +(0.1,1);
    \path[ dense] (0.15,1.5) rectangle +(0.1,1);
    \path[ dense] (0.4,0) rectangle +(0.1,1);
    \path[ dense] (1.3,2) rectangle +(0.1,0.15);
    \path[ dense] (1.6,2) rectangle +(0.15,0.1);
    \path[ dense] (0.15,2.5) rectangle (3.1,2.55);
    \path[ dense] (1.4,2.85) rectangle (3.1,2.8);
    \path[ dense] (1.5,3.8) rectangle (3.9,3.75);
    \path[ dense] (1.65,3.65) rectangle (3.9,3.6);
    \path[ dense] (3.55,4) rectangle +(0.15,-0.15);
    \path[ dense] (3.8,2) rectangle +(0.1,1);
},
bin2-dense-items/.pic = {
    \path[ dense] (4,0) rectangle +(-0.15,0.1);
    \path[ dense] (4,0.1) rectangle +(-0.15,0.15);
    \path[ dense] (4,0.25) rectangle +(-0.15,0.1);
    \path[ dense] (4,0.35) rectangle +(-0.1,0.1);
    \path[ dense] (4,0.45) rectangle +(-0.1,0.15);
    \path[ dense] (4,0.6) rectangle +(-0.1,1);
    \path[ dense] (4,1.6) rectangle +(-0.1,1);
    \path[ dense] (4,2.6) rectangle +(-0.1,1);
    \path[ dense] (3.85,0) rectangle +(-0.1,1);
},
dense-items-box/.pic = {
    \path[ dense] (0,0) rectangle +(2.95,0.05);
    \path[ dense] (0,0.05) rectangle +(2.4,0.05);
    \path[ dense] (0,0.1) rectangle +(2.25,0.05);
    \path[ dense] (2.25,0.1) rectangle +(1.7,0.05);
},
bin1-items/.pic = {
    \path[ndense] (0.45,1) rectangle +(0.55,1);
    \path[ndense] (0,4) rectangle +(1.1,-1);
    \path[ndense] (1.1,4) rectangle +(0.15,-1.2);
    \path[ndense] (1.25,4) rectangle +(0.15,-1.3);
    \path[ndense] (1.4,4) rectangle +(0.1,-1);
    \path[ndense] (1.5,3.1) rectangle (3.4,3.2);
    \path[ndense] (1.5,3.2) rectangle (3.4,3.35);
    \path[ndense] (1.4,2.95) rectangle (2.7,2.85);
    \path[ndense] (0.2,2.55) rectangle (3.1,2.65);
    \path[ndense] (0.3,2.4) rectangle (3,2.5);
    \path[ndense] (0.55,2.25) rectangle (3,2.4);
    \path[ndense] (3.25,2) rectangle +(0.15,1);
    \path[ndense] (3.1,2) rectangle +(0.15,1);
    \path[ndense] (1,2) rectangle (4,0);
    \path[ndense] (1.5,3.75) rectangle (3.9,3.65);
    \path[ndense] (1.75,3.6) rectangle (3.9,3.5);
},
bin2-items/.pic = {
    \path[ndense] (0.5,1) rectangle +(0.75,1);
    \path[ndense] (0,4) rectangle +(1.25,-1);
    \path[ndense] (1.25,4) rectangle +(0.15,-1.2);
    \path[ndense] (1.4,4) rectangle +(0.15,-1.3);
    \path[ndense] (1.55,4) rectangle +(0.1,-1);
    \path[ndense] (1.75,3.1) rectangle (3.75,3.2);
    \path[ndense] (1.75,3.2) rectangle (3.75,3.35);
    \path[ndense] (1.75,2.95) rectangle (3.25,2.85);
    \path[ndense] (0.25,2.55) rectangle (3.25,2.65);
    \path[ndense] (0.5,2.4) rectangle (3.25,2.5);
    \path[ndense] (0.75,2.25) rectangle (3.25,2.4);
    \path[ndense] (3.25,2) rectangle +(0.15,1);
    \path[ndense] (3.4,2) rectangle +(0.15,1);
},
bin1-bin2-items/.pic = {
    \path[ndense] (0.2,2.9) rectangle +(0.1,-0.1);
    \path[ndense] (0.4,2.9) rectangle +(0.15,-0.15);
    \path[ndense] (0,1) rectangle +(0.15,1);
    \path[ndense] (0,2) rectangle +(0.1,0.7);
    \path[ndense] (0.15,0) rectangle +(0.15,1.5);
    \path[ndense] (0.3,0) rectangle +(0.1,1);
    \path[ndense] (0.3,1) rectangle +(0.15,1.2);
    \path[ndense] (0.5,0) rectangle +(0.15,0.8);
    \path[ndense] (0.65,0) rectangle +(0.15,0.7);
    \path[ndense] (1,2) rectangle +(0.1,0.15);
},
bin3-items/.pic = {
    \path[ndense] (1,2) rectangle (4,0);
    \path[ndense] (1.5,3.75) rectangle (4,3.65);
    \path[ndense] (1.75,3.6) rectangle (4,3.5);
},
strip-items/.pic = {
    \path[ndense] (-0.1,3.5) rectangle +(-0.1,-0.1);
    \path[ndense] (-0.6,2) rectangle +(0.15,1.3);
    \path[ndense] (-0.6,4) rectangle +(0.15,-0.15);
    \path[ndense] (-0.35,2) rectangle +(0.15,1.2);
    \path[ndense] (-0.45,2) rectangle +(0.1,1.1);
    \path[ndense] (-0.1,2) rectangle +(0.1,2);
},
bin/.pic = {
    \draw[{black!40},xstep=0.25,ystep=4,line width=0.3pt,densely dotted] (0,0) grid (4,4);
    \draw[{black!60}] (3.5,-0.1) -- (3.5,4.1);
    \draw (0,0) rectangle (4,4);
}]

\node at (0, 0) {=};
\coordinate (bin1) at (-4.5,-2);
\coordinate (bin1end) at (-0.5,-2);
\coordinate (bin2) at (1,0.7);
\coordinate (box-dense) at (1,-0.125);
\coordinate (bin3) at (1,-4.7);
\coordinate (box-ndense) at (5.5,-2);
\coordinate (box-ndense-end) at (6.15,-2);
\draw[bigbrace] (0.75,-5) -- (0.75,5);
\draw[bigbrace] (6.4,5) -- (6.4,-5);
\draw[smallbrace] (-1,-2) -- node[below=1pt] {\small $\epsLarge$} (-0.5,-2);
\draw[smallbrace] (box-ndense) -- node[below=1pt] {\small $\epsLarge + \epsSmall$} (box-ndense-end);

\pic at (bin1) {bin1-items};
\pic at (bin1) {bin1-bin2-items};
\pic at (bin1) {bin1-dense-items};
\pic at (bin2) {bin2-items};
\pic at (bin2) {bin1-bin2-items};
\pic at (bin2) {bin2-dense-items};
\pic at (bin3) {bin3-items};
\pic at (bin1end) {strip-items};
\pic at (box-ndense-end) {strip-items};
\pic at (box-dense) {dense-items-box};

\pic at (bin1) {bin};
\pic at (bin2) {bin};
\pic at (bin3) {bin};
\draw (box-dense) rectangle +(4,0.25);
\draw (box-ndense) rectangle +(0.65,4);
\end{tikzpicture}

\caption{A bin is split into 2 bins and 2 boxes.
Then the widths and $x$-coordinates of big and wide non-dense items
are rounded up to a multiple of $\epsLarge^2/4$.
Dense items are shaded dark and non-dense items are shaded light.}
\label{fig:split-bin}
\end{figure}
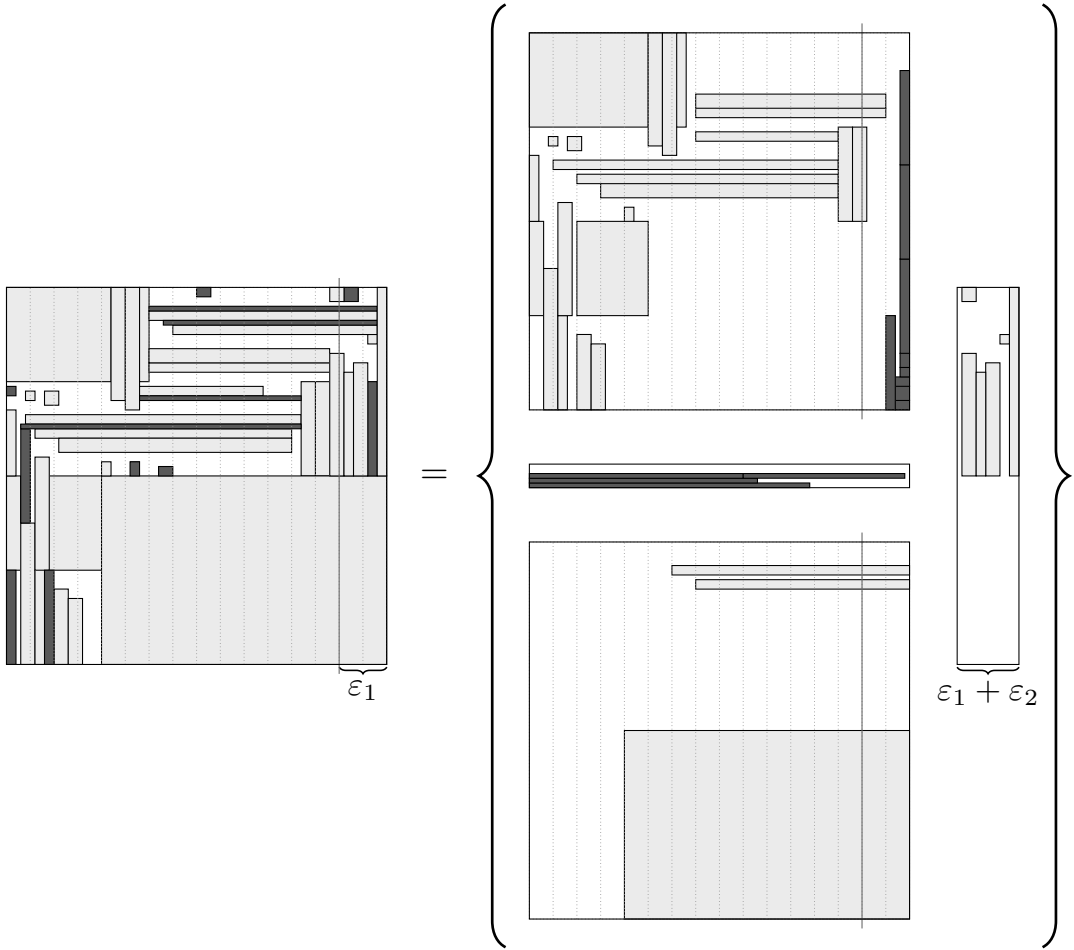
\end{proof}

\begin{lemma}
\label{lem:rnd1-rot}
When item rotation is allowed, given a packing of items into a bin, we can
round up the width of some wide and big items to a multiple of $\epsLarge^2/4$
and round up the height of some tall and big items to a multiple of $\epsLarge^2/4$
and get a packing into 2 bins and 1 box where:
\begin{tightemize}
\item Each bin satisfies \cref{prop:hrnd}.
\item $v_1$ of the box is at most $1/2$.
\item The box has height 1 and width $\epsLarge + \epsSmall$.
\item The box only contains tall and small non-dense items.
\end{tightemize}
\end{lemma}
\begin{proof}[Proof sketch]
Suppose the bin contains items $J$. Move dense items to a vertical box $D_H$
using \cref{lem:dense-pack}.
Move $C^{(R)}$ to a box of height 1 and width $\epsLarge + \epsSmall$.
Move $I^{(R)}$ to a new bin and round item widths to a multiple of $\epsLarge^2/4$.
Now $S^{(R')}$ is empty, so pack $D_H$ into $S^{(R')}$.
The rest of the proof is similar to that of \cref{lem:rnd1}.
\end{proof}

\subsection{Getting Slack in Weight of Bins}
\label{sec:rbbp-2d-extra:slack}

For a bin $J$, if $\forall j \in [d], v_j(J) \le 1 - \eps$,
then we can round up weights of items.
Hence, we would like to have bins with (roughly) this property.

\begin{definition}
A bin $J$ is said to be $\eps$-slacked iff at least one of these conditions holds:
\begin{itemize}
\item $\forall j \in [d], v_j(J) \le 1-\eps$.
\item $|J| = 1$.
\item $|J| = 2$ and $J$ only contains dense items and
$\forall i \in J, v_{\max}(i) \le 1/2$.
\end{itemize}
A packing of items into multiple bins is said to be $\eps$-slacked
iff all bins in the packing are $\eps$-slacked.
\end{definition}

We say that packing of items in a bin is \emph{quarter-structured} iff
the bin is $\eps$-slacked and satisfies either \cref{prop:hrnd} or \cref{prop:vrnd}.
We would like to round up the width or height of each item in $I$ and repack the items
into bins such that each bin is quarter-structured.

\begin{lemma}
\label{lem:split-slack}
Let $D \subseteq [d]$ and we have a parameter $\delta \le 1/4$.
Let $I$ be a set of items where $\forall j \in D, v_j(I) \le V_j$.
Then we can partition $I$ into at most $\abs{D}+1$ disjoint subsets
such that each subset $I'$ satisfies one of these properties:
\begin{itemize}
\item $|I'| = 1$.
\item $\forall j \in D, v_j(I') \le (1-\delta)V_j$.
\end{itemize}
\end{lemma}
\begin{proof}
Let $I_L \defeq \{i \in I: \exists j \in D, v_j(i) > (1-2\delta)V_j\}$.
Move each item in $I_L$ to a new box.
Let $D' \defeq \{ j \in D: v_j(I_L) > (1-2\delta)V_j\}$.
Then $|D'| \ge |I_L|$ and $\forall j \in D', v_j(I-I_L) < 2\delta V_j \le (1-\delta)V_j$.

Order the items in $I-I_L$ arbitrarily.
For each $j \in D-D'$, find the smallest prefix $P_j$ such that $v_j(P_j) \ge \delta V_j$.
Let $i_j$ be the last item in $P_j$. Then $v_j(P_j - i_j) < \delta V_j$.
Since we removed items from $I_L$, $v_j(i_j) \le (1-2\delta)V_j$.
Therefore, $v_j(P_j) \le (1-\delta)V_j$.

Now order these prefixes in non-decreasing order of cardinality.
Let them be $P_1'$, $P_2'$, $\ldots$, $P_{|D-D'|}'$.
The sets $P_1'$, $P_2'-P_1'$, $P_3'-P_2'$, $\ldots$ form a partition of $P_{|D-D'|}$.
Put each such set in a new box, if the set is not empty.
The items which remain in the original box are $Q \defeq I-I_L-P_{|D-D'|}'$.
$\forall j \in D-D', Q \subseteq I-I_L-P_j$.
Since $v_j(P_j) \ge \delta V_j$, we get that $\forall j \in D-D', v_j(Q) \le (1-\delta)V_j$.

Therefore, total number of boxes needed is at most
$1 + |I_L| + |D-D'| \le 1 + |D'| + |D-D'| \le |D| + 1$.
\end{proof}

\begin{lemma}
\label{lem:ws-d-pack}
Given a packing of items $I$ into $m$ bins, we can
round up the width of some non-dense wide and big items in $I$ to a multiple of $\epsLarge^2/4$
and round up the height of some non-dense tall and big items in $I$ to a multiple of $\epsLarge^2/4$
and repack the items into $(d+4)m$ $\eps$-slacked bins such that
each bin satisfies either \cref{prop:hrnd} or \cref{prop:vrnd}.
\end{lemma}
\begin{proof}
Let $B_1, B_2, \ldots, B_m$ be a packing of items $I$ into $m$ bins.
For each bin $B_k$, we can use \cref{lem:rnd1} to round up some items in $B_k$
and split $B_k$ into bins $J_k$ and $K_k$ and boxes $W_k$ and $H_k$.
\WLoG, assume $W_k$ is a horizontal box. Put each box in a new bin.
Then $W_k$ satisfies \cref{prop:vrnd} and $H_k$ satisfies \cref{prop:hrnd}.

Let $D_k \defeq \{j \in [d]: v_j(J_k) > (1-\eps)\}$,
$E_k \defeq \{j \in [d]: v_j(K_k) > (1-\eps)\}$,
$F_k \defeq \{j \in [d]: v_j(W_k) > (1-\eps)\}$
and $G_k \defeq \{j \in [d]: v_j(H_k) > (1-\eps)\}$.
$D_k$, $E_k$, $F_k$ and $G_k$ are pairwise disjoint and they are subsets of $[d]$.
Now use \cref{lem:split-slack} with parameter $\delta=\eps$ on bin $J_k$
with dimensions $D_k$. This splits $J_k$ into $|D_k| + 1$ $\eps$-slacked bins.
Similarly, by splitting $K_k$, $W_k$ and $H_k$, we get
$|E_k|+1$, $|F_k|+1$ and $|G_k|+1$ $\eps$-slacked bins respectively.

The total number of bins from $B_k$ is $|D_k| + |E_k| + |F_k| + |G_k| + 4 \le d+4$.
Therefore, we get a total of $(d+4)m$ bins.

$J_k$, $K_k$, $W_k$ and $H_k$ satisfy the desired properties except possibly
$\eps$-slackness. When we split a bin into multiple bins, the position of items
relative to the bin isn't changed. Therefore, the split bins continue to satisfy these properties.
\end{proof}

\begin{lemma}
\label{lem:ws-d-pack-rot}
Given a packing of items $I$, if item rotations are allowed, we can
round up the width of some non-dense wide and big items in $I$ to a multiple of $\epsLarge^2/4$
and round up the height of some non-dense tall and big items in $I$ to a multiple of $\epsLarge^2/4$
and repack $I$ into $(d+3)m$ $\eps$-slacked bins such that
each bin satisfies \cref{prop:hrnd}.
\end{lemma}
\begin{proof}[Proof sketch]
Use \cref{lem:rnd1-rot} on each bin in the packing to get 3 bins.
The rest of the proof is similar to \cref{lem:ws-d-pack}.
\end{proof}

We will now try to improve upon \cref{lem:ws-d-pack,lem:ws-d-pack-rot}
for the special case $d=1$.

\begin{lemma}
\label{lem:light-cont-pack}
Let there be $m$ boxes of width 1 and height at most $\eps$.
Let the weight of each box be at most $k\eps$ in each dimension, where $k \in \{1, 2\}$.
Then we can pack these boxes into $1 + m\cdot k\eps/(1-k\eps)$ bins such that
the resulting bins are $k\eps$-slacked.
\end{lemma}
\begin{proof}
We can pack $1/k\eps-1$ boxes in 1 bin with the resulting bin being $k\eps$-slacked.
This is because the sum of weights is at most $1-k\eps$ in each dimension
and the total height of the boxes is at most $1/k-\eps \le 1$.
The total number of bins used is at most
$\ceil{m/((1/k\eps)-1)}$ which in turn is at most $1 + m\cdot k\eps/(1-k\eps)$.
\end{proof}
The above lemma can also be used for vertical boxes, i.e. height 1 and width at most $\eps$.

\begin{lemma}
\label{lem:non-dense-cont-pack}
Let $d=1$. Let there be $m$ boxes of width 1 and height at most $\eps$
containing only non-big non-dense items. Let the weight of each box be at most $1/2$.
Then we can pack these boxes into
$2 + m\left(1/2 + \eps/(1-\eps)\right)$
bins such that the resulting bins are $\eps$-slacked.
\end{lemma}
\begin{proof}
Let $i$ be an item in the box.
Boxes only have non-big items, so $a(i) \le \epsSmall$.
Boxes only have non-dense items,
so $v_1(i) \le a(i)/\epsLarge^2 \le \eps/2$.

From each box, choose the smallest prefix $S$ for which $v_1(S) \ge \eps/2$.
Then $v_1(S) \le \eps$. Remove $S$ and put it in a new box of the same dimensions.

This gives us $m$ boxes of weight at most $(1-\eps)/2$.
We can pair them up and pack them into $\ceil{m/2} \le m/2 + 1$ bins.
Those bins will be $\eps$-slacked.

We also get $m$ new boxes of weight at most $\eps$.
We can pack $1/\eps-1$ such boxes into a bin.
This gives us at most $1 + m\cdot\eps/(1-\eps)$ bins.
These bins are $\eps$-slacked.

Total number of bins used is
 $\left(\frac{m}{2} + 1\right) + \left(1 + m\cdot\eps/(1-\eps)\right)
= 2 + m\left(1/2 + \eps/(1-\eps)\right)$.
\end{proof}

\begin{lemma}
\label{lem:dense-cont-pack}
Let $d=1$. Let there be $m$ boxes of width 1 and height $\deltaDense$.
Suppose the boxes only have dense items,
and each box has total weight at least $\eps$ and at most $1/2$.
Then we can pack the items of these boxes into at most $3 + 2m/3$ bins such that
the resulting bins are $\eps$-slacked.
\end{lemma}
\begin{proof}
Let there be $t$ boxes that have an item of weight at least $1/2-2\eps$.
No item has weight more than half.
Therefore, we can pair up these high-weight items into at most $t/2 + 1$ $\eps$-slacked bins.
In each of these $t$ boxes, the remaining items have weight at most $2\eps$.
Since $\eps \ge \deltaDense$, by \cref{lem:light-cont-pack},
we can pack them into $1 + 2\eps t/(1-2\eps)$ number of $\eps$-slacked bins.

$m-t$ boxes have all items of weight less than $1/2-2\eps$.
Each of these boxes has total weight between $\eps$ and $1/2$.
Each box \emph{can be} split into 2 boxes as follows:
Order the items in a box and find the smallest prefix of weight at least $\eps$.
Since there are no items of weight more than $1/2-2\eps$,
such a prefix has weight between $\eps$ and $1/2-\eps$.

Arrange the $m-t$ boxes into groups of at most 3 boxes each.
Let $C_1$, $C_2$, $C_3$ be these boxes in one such group.
Split $C_1$ and $C_2$ into 2 boxes each by the above method.
Let the resulting boxes be $C_1', C_1'', C_2', C_2''$ respectively.
Assume \wLoG{} that $v_1(C_j') \le v_1(C_j'')$ for $j \in [2]$.
Pack $C_1', C_2', C_1''$ into 1 bin. It has weight at most $1-\eps$, so it is $\eps$-slacked.
Pack $C_2''$ and $C_3$ into 1 bin. It has weight at most $1-\eps$, so it is $\eps$-slacked.
Therefore, we can convert a group of 3 boxes into 2 $\eps$-slacked bins.

Total bins used is at most
$\left(1 + 2\eps t/(1-2\eps)\right) + 2\ceil{(m-t)/3}
\le 3 + 2m/3 - t\left( \frac{2}{3} - \frac{2\eps}{1-2\eps} \right)$
which is at most $3 + 2m/3$, assuming {$\eps \le 1/5$}.
\end{proof}
The above lemma can also be used for vertical boxes,
i.e. height 1 and width at most $\deltaDense$.

\begin{lemma}
\label{lem:ws-1-pack}
Given a packing of items $I$ into $m$ bins, when $d=1$, we can
round up the width of some non-dense wide and big items in $I$ to a multiple of $\epsLarge^2/4$
and round up the height of some non-dense tall and big items in $I$ to a multiple of $\epsLarge^2/4$
and repack $I$ into $\left(3 + \frac{1}{6} + \frac{\eps}{1-\eps}\right)m + 12$
$\eps$-slacked bins such that
each bin satisfies either \cref{prop:hrnd} or \cref{prop:vrnd}.
\end{lemma}

\begin{proof}
Let $B_1, B_2, \ldots, B_m$ be a packing of items $I$ into $m$ bins.
For each bin $B_k$, we can use \cref{lem:rnd1} to round up some items in $B_k$
and split $B_k$ into bins $J_k$ and $K_k$ and boxes $W_k$ and $H_k$,
where $W_k$ is a horizontal box and $H_k$ is a vertical box,
i.e. the width of $W_k$ is 1 and the height of $H_k$ is 1.

\textbf{Classifying bins:}

We will now classify the bins $B_1, B_2, \ldots, B_m$.

\textbf{Type 1}: $v_1(W_k) \le \eps$ and $v_1(H_k) \le \eps$:

Among $J_k$ and $K_k$, at most 1 bin will have weight more than $1-\eps$.
Use \cref{lem:split-slack} to split it into 2 bins.
So for each original bin of type 1, we get at most 3 $\eps$-slacked bins and 2 boxes,
one horizontal and one vertical, each of total weight at most $\eps$.

Both $J_k$ and $K_k$ satisfy either \cref{prop:hrnd} or \cref{prop:vrnd}.
When we split a bin into multiple bins, the position of items
relative to the bin isn't changed. Therefore, the bins continue to satisfy these properties.

\textbf{Type 2}: $v_1(W_k) > \eps$ and $v_1(H_k) \le \eps$:

$v_1(W_k) > \eps$ implies that $v_1(J_k) \le 1-\eps$ and $v_1(K_k) \le 1-\eps$,
so $J_k$ and $K_k$ are already $\eps$-slacked. Pack $W_k$ in a bin.
Since $v_1(W_k) \le 1/2 \le 1-\eps$, $W_k$ is $\eps$-slacked.
$W_k$ satisfies \cref{prop:vrnd}.
So for each original bin of type 2, we get at most 3 $\eps$-slacked bins
and 1 vertical box of weight at most $\eps$.

\textbf{Type 3}: $v_1(W_k) \le \eps$ and $v_1(H_k) > \eps$:

The analysis is similar to type 2.
For each original bin of type 3, we get at most 3 $\eps$-slacked bins
and 1 horizontal box of weight at most $\eps$.

\textbf{Type 4}: $v_1(W_k) > \eps$ and $v_1(H_k) > \eps$:

$v_1(W_k) > \eps$ implies that $v_1(J_k) \le 1-\eps$ and $v_1(K_k) \le 1-\eps$,
so $J_k$ and $K_k$ are already $\eps$-slacked.
So we have at most 2 $\eps$-slacked bins and 2 boxes of weight at most $1/2$.

\textbf{Repacking boxes:}

We will now try to pack the boxes into bins.
Each of these bins packs some dense boxes and some non-dense boxes.
If multiple dense boxes were packed in a bin, we can use \cref{lem:dense-pack}
to repack them into a single dense box and move that box to an edge of the bin.
Bins that only pack horizontal boxes satisfy \cref{prop:vrnd}.
Bins that only pack vertical boxes satisfy \cref{prop:hrnd}.

Among $B_1, B_2, \ldots, B_m$, let there be $m_k$ bins of type $k$.

The number of $\eps$-slacked bins is at most
$3m_1 + 3m_2 + 3m_3 + 2m_4 \le 3m - m_4$.
We also have $m_1 + m_3$ horizontal boxes and $m_1 + m_2$ vertical boxes
of weight at most $\eps$ each. Since $\deltaDense \le \epsLarge/2 \le \eps$
and $\epsLarge + \epsSmall \le 2\eps/3 + 2\eps^3/9 \le \eps$, each box has the smaller
geometric dimension at most $\eps$.
By \cref{lem:light-cont-pack}, the number of bins we need to pack them is at most
$2 + \frac{\eps}{1-\eps}(2m_1 + m_2 + m_3)$.

We have $m_4$ horizontal boxes and $m_4$ vertical boxes
that each have weight between $\eps$ and $1/2$.
$m_4$ of these are dense boxes and $m_4$ are non-dense boxes.

The non-dense boxes don't have big items. Since $\eps \ge \epsLarge + \epsSmall$,
by \cref{lem:non-dense-cont-pack}, the number of bins needed to pack them is at most
$4 + \left(\frac{1}{2} + \frac{\eps}{1-\eps}\right)m_4$.

By \cref{lem:dense-cont-pack}, we can pack the dense boxes into
$6 + 2m_4/3$ bins, where each bin is $\eps$-slacked.

The total number of bins used is at most
\begin{align*}
& (3m - m_4) + \left(2 + \frac{\eps}{1-\eps}(2m_1 + m_2 + m_3)\right)
 + \left(4 + \left(\frac{1}{2} + \frac{\eps}{1-\eps}\right)m_4\right)
+ \left(6 + \frac{2}{3}m_4\right)
\\ &= 12 + \left(3 + \frac{1}{6} + \frac{\eps}{1-\eps}\right)m
    + \frac{\eps}{1-\eps}m_1 + \frac{m_4-m}{6}
\\ &\le 12 + \left(3 + \frac{1}{6} + \frac{\eps}{1-\eps}\right)m
    - \left(\frac{1}{6} - \frac{\eps}{1-\eps}\right)m_1
\tag{$m_4 \le m - m_1$}
\\ &\le 12 + \left(3 + \frac{1}{6} + \frac{\eps}{1-\eps} \right)m.
\tag{$\eps \le 1/8$}
\end{align*}
\end{proof}

\begin{lemma}
\label{lem:ws-1-pack-rot}
Given a packing of items $I$ into $m$ bins, when $d=1$ and item rotations are allowed, we can
round up the width of some non-dense wide and big items in $I$ to a multiple of $\epsLarge^2/4$
and round up the height of some non-dense tall and big items in $I$ to a multiple of $\epsLarge^2/4$
and repack $I$ into $\left(3 + \frac{\eps}{1-\eps}\right)m + 1$
$\eps$-slacked bins such that each bin satisfies \cref{prop:hrnd}.
\end{lemma}

\begin{proof}[Proof sketch]
Using techniques from the proof of \cref{lem:ws-1-pack},
we get at most $3m$ $\eps$-slacked bins and at most $m$ vertical boxes,
where each box has total weight at most $\eps$.
Using \cref{lem:light-cont-pack}, we get the desired results.
\end{proof}

We can summarize \cref{lem:ws-d-pack,lem:ws-d-pack-rot,lem:ws-1-pack,lem:ws-1-pack-rot} as follows:
\begin{theorem}
\label{thm:slack-pack}
Given a packing of $I$ into $m$ bins, we can
round up the width of some non-dense wide and big items in $I$ to a multiple of $\epsLarge^2/4$
and round up the height of some non-dense tall and big items in $I$ to a multiple of $\epsLarge^2/4$
and repack $I$ into at most $am + b$ $\eps$-slacked bins such that
each bin satisfies either \cref{prop:hrnd} or \cref{prop:vrnd}.
Here the values of $a$ and $b$ depend on the value of $d$ and whether item rotations are allowed.
See \cref{table:slack-pack-ab}.
\begin{table}[H]
\centering
\caption{Values of $a$ and $b$ for \cref{thm:slack-pack}.}
\begin{tabular}{|c|c|c|c|}
\hline & $a$ & $b$ & Lemma
\\ \hline rotations forbidden
    & $d+4$
    & $0$
    & \cref{lem:ws-d-pack}
\\ \hline rotations allowed
    & $d+3$
    & $0$
    & \cref{lem:ws-d-pack-rot}
\\ \hline rotations forbidden and $d=1$
    & ${\displaystyle 3 + \frac{1}{6} + \frac{\eps}{1-\eps}}$
    & $12$
    & \cref{lem:ws-1-pack}
\\ \hline rotations allowed and $d=1$
    & ${\displaystyle 3 + \frac{\eps}{1-\eps}}$
    & $1$
    & \cref{lem:ws-1-pack-rot}
\\ \hline
\end{tabular}
\label{table:slack-pack-ab}
\end{table}
\end{theorem}

\subsection{Rounding Weights}
\label{sec:rbbp-2d-extra:round-weights}

\begin{definition}[Weight classes]
\label{defn:weight-classes}
Two items $i_1$ and $i_2$ are said to belong to the same weight class iff
one of these conditions hold:
\begin{tightemize}
\item $i_1$ and $i_2$ are big and $\forall j \in [d], v_j(i_1) = v_j(i_2)$
\item $i_1$ and $i_2$ are non-dense and wide
    and $\forall j \in [d], v_j(i_1)/h(i_1) = v_j(i_2)/h(i_2)$.
\item $i_1$ and $i_2$ are non-dense and tall
    and $\forall j \in [d], v_j(i_1)/w(i_1) = v_j(i_2)/w(i_2)$.
\item $i_1$ and $i_2$ are non-dense and small
    and $\forall j \in [d], v_j(i_1)/a(i_1) = v_j(i_2)/a(i_2)$.
\item $i_1$ and $i_2$ are dense and \hyperref[defn:heavy-light]{light}
    and $\forall j \in [d], v_j(i_1)/v_{\max}(i_1) = v_j(i_2)/v_{\max}(i_2)$.
\item $i_1$ and $i_2$ are dense and \hyperref[defn:heavy-light]{heavy}
    and $\forall j \in [d], v_j(i_1) = v_j(i_2)$.
\end{tightemize}
\end{definition}

Big items that have the same geometric dimensions and the same weight class are identical.
We will round the geometric dimensions of dense items to 0, so heavy items
of the same weight class would be identical.

\begin{definition}[Slicing and fractional packing]
\label{defn:slicing}
Let $I$ be a set of items. $\widehat{I}$ is called a slicing of $I$ iff
$\widehat{I}$ can be obtained from $I$ by one or more of the following operations:
\begin{tightemize}
\item Horizontally slicing a non-dense wide item
(i.e. if a wide item $i$ is sliced into items $i_1$ and $i_2$, then
$w(i) = w(i_1) = w(i_2)$ and $h(i) = h(i_1) + h(i_2)$).
\item Vertically slicing a non-dense tall item
(i.e. if a tall item $i$ is sliced into items $i_1$ and $i_2$, then
$h(i) = h(i_1) = h(i_2)$ and $w(i) = w(i_1) + w(i_2)$).
\item Slicing a non-dense small item in one or both dimensions.
\item Slicing a light dense item of zero area.
\end{tightemize}

A packing of $\widehat{I}$ into bins is called a fractional packing of $I$.
\end{definition}

Wide non-dense items of the same width and of the same weight class
can be used interchangeably while trying to get a fractional packing.
Similarly, tall non-dense items of the same height and of the same weight class are
interchangeable, small non-dense items of the same weight class are interchangeable
and light dense items of zero area and the same weight class are interchangeable.

We will now see how to round the weights of items so that they belong to a constant
number of weight classes.

\subsubsection{Rounding Weights of Non-Dense Items}

\begin{transformation}
\label{trn:wround-non-dense}
Given an item $i$, $\forall j \in [d]$,
\begin{tightemize}
\item If $i$ is big, round up $v_j(i)$ to a positive multiple of $\epsLarge^2\eps/8$.
\item If $i$ is wide and non-dense, round up $v_j(i)$ to a positive multiple of
    $h(i) \epsLarge\eps/8$.
\item If $i$ is tall and non-dense, round up $v_j(i)$ to a positive multiple of
    $w(i) \epsLarge\eps/8$.
\item If $i$ is small and non-dense, round up $v_j(i)$ to a positive multiple of
    $a(i) \eps/8$.
\end{tightemize}
\end{transformation}

\begin{lemma}
\Cref{trn:wround-non-dense} is valid, i.e. for any item $i$,
$\forall j \in [d], v_j(i) \le 1$ after the transformation.
\end{lemma}
\begin{proof}
Since $\epsLarge^{-1}, \eps^{-1} \in \mathbb{Z}$,
the transformed weight of a big item can be at most 1.

Since $\epsSmall \le (1-\eps)\epsLarge^2$,
the weight of a non-dense wide/tall/small item is at most $\epsSmall/\epsLarge^2 \le 1-\eps$
and rounding it up can only increase it by at most $\eps/8$.
\end{proof}

\begin{lemma}
\label{lem:wround-non-dense-slackness}
Let a bin $J$ be $\mu$-slacked, for $\eps/8 \le \mu \le \eps$.
Then after applying \cref{trn:wround-non-dense},
the bin will be $(\mu - \eps/8)$-slacked.
\end{lemma}
\begin{proof}
If the bin contains a single item, it will remain $\mu$-slacked after the transformation.

Suppose the bin contains multiple items, then $\forall j \in [d], v_j(J) \le 1-\mu$.
Let there be $p$ big items.
Let the total height of wide non-dense items be $H$.
Let the total width of tall non-dense items be $W$.
Let the total area of small non-dense items be $A$.

The total area of big items is at least $p\epsLarge^2$,
of wide items is at least $\epsLarge H$
and of tall items is at least $\epsLarge W$.
Since the total area of items in $J$ is at most 1,
\[ \epsLarge^2p + \epsLarge H + \epsLarge W + A \le 1. \]
The total increase in $v_j(J)$ is at most
\[ \frac{\eps}{8}(\epsLarge^2p + \epsLarge H + \epsLarge W + A) \le \frac{\eps}{8}. \]
Therefore, the resulting bin is $(\mu - \eps/8)$-slacked.
\end{proof}

\begin{observation}
\label{obs:non-dense-vinc}
Since we're rounding up weights of non-dense items in \cref{trn:wround-non-dense},
some items may not continue to satisfy the non-denseness property,
i.e. $v_j(i)/a(i)$ may exceed $1/\epsLarge^2$ for some items $i$ and some $j \in [d]$.

However, this will not affect us much. Formally:
\begin{enumerate}[a]
\item Big items will remain non-dense, since $a(i) > \epsLarge^2$ and $v_{\max}(i) \le 1$.
\item Small items will remain non-dense, since $v_{\max}(i)/a(i)$ will be rounded up to
a multiple of $\eps/8$, and $\eps/8$ divides $1/\epsLarge^2$.
\item For wide items, $v_{\max}(i)/a(i)$ may rise to at most $1/\epsLarge^2 + \eps/8$.
Furthermore, $v_{\max}(i)/h(i)$ will be rounded to a multiple of $\epsLarge\eps/8$,
and $\epsLarge\eps/8$ divides $1/\epsLarge^2$, so $v_{\max}(i)/h(i)$ will continue to be
at most $1/\epsLarge^2$.
\item \label{item:non-dense-vinc:tall} For tall items,
$v_{\max}(i)/a(i)$ may rise to at most $1/\epsLarge^2 + \eps/8$.
Furthermore, $v_{\max}(i)/w(i)$ will be rounded to a multiple of $\epsLarge\eps/8$,
and $\epsLarge\eps/8$ divides $1/\epsLarge^2$, so $v_{\max}(i)/w(i)$ will continue to be
at most $1/\epsLarge^2$.
\end{enumerate}
\end{observation}

Even if $v_{\max}(i)/a(i)$ of some non-dense items exceeds $1/\epsLarge^2$
after \cref{trn:wround-non-dense}, we will continue to consider them non-dense items.

\begin{lemma}
\label{lem:num-weight-classes}
Define
$n_{\bwc} \defeq \left(8/(\epsLarge^2\eps)\right)^d,
 n_{\wwc} \defeq \left(8/(\epsLarge^3\eps)\right)^d,
 n_{\swc} \defeq \left(8/(\epsLarge^2\eps)\right)^d$.
After \cref{trn:wround-non-dense}, the number of weight classes
of big items is at most $n_{\bwc}$,
of wide non-dense items is at most $n_{\wwc}$,
of tall non-dense items is at most $n_{\wwc}$
and of small non-dense items is at most $n_{\swc}$.
\end{lemma}

\subsubsection{Rounding Weights of Dense Items}

\begin{transformation}
\label{trn:round-to-0}
For item $i$, if $i$ is non-dense, do nothing. If $i$ is dense, set $w(i)$ and $h(i)$ to 0 and
for each $j \in [d]$, if $v_j(i) \le (\eps/8d) v_{\max}(i)$, then set $v_j(i)$ to 0.
\end{transformation}

Since \cref{trn:round-to-0} rounds down weights,
we need to prove that we can easily undo this transformation.
\begin{lemma}
\label{lem:round-to-0}
Let $J$ be a set of items.
Let $J'$ be the items obtained by applying \cref{trn:round-to-0} to $J$.
Suppose we're given a packing of $J'$ into a bin that satisfies
\cref{prop:hrnd}\ref{item:hrnd:dense} and is $\mu$-slacked, for some $\mu \ge \eps/8$.

Then there is a polynomial-time algorithm to convert the packing of $J'$
into a packing of $J$ that satisfies \cref{prop:hrnd}\ref{item:hrnd:dense},
is $(\mu - \eps/8)$-slacked, and the position of non-dense items in the packing of $J$
is the same as the position of non-dense items in the packing of $J'$.
\end{lemma}
(Analogous lemma holds for \cref{prop:vrnd}\ref{item:vrnd:dense})
\begin{proof}
By \cref{lem:dense-pack}, we can always pack dense items in $J$ in polynomial time
into a box of height 1 and width $\deltaDense$.
Since $\deltaDense \le \epsLarge/2$, this box fits in $S^{(R')}$.
Therefore, \cref{prop:hrnd}\ref{item:hrnd:dense} is satisfied.

Now we will prove that the packing of $J$ is $(\mu - \eps/8)$-slacked.
There are 3 cases to consider:

\textbf{Case 1}: $\forall j \in [d], v_j(J') \le 1-\mu$:\\
For each $j \in [d]$ and each dense item $i$,
reverting the transformation increases $v_j(i)$ by at most $(\eps/8d) v_{\max}(i)$.
So by \cref{claim:wmax-s-lim}, we get
\[ v_j(J) \le v_j(J') + \frac{\eps}{8d} v_{\max}(J')
\le v_j(J') + \frac{\eps}{8} \le 1 - \left(\mu - \frac{\eps}{8}\right). \]
Therefore, $J$ is $(\mu - \eps/8)$-slacked.

\textbf{Case 2}: $|J'| = 1$:\\
Then $|J|=1$, so $J$ is $\mu$-slacked.

\textbf{Case 3}: $|J'| = 2$ and $J'$ only has dense items
and $\forall i \in J', 1/2 - \mu \le v_{\max}(i) \le 1/2$:\\
The 0-value dimensions of $i$ increase to $(\eps/8d) v_{\max}(i) \le v_{\max}(i)$,
so $v_{\max}(i)$ remains the same across this transformation. So $J$ is $\mu$-slacked.
\end{proof}

As the first step in rounding dense items, we apply \cref{trn:round-to-0} to $I$.

Since $\epsLarge \le 1/4d$, we get $\epsSmall \le (\eps/8d)\epsLarge$.
Hence, \cref{trn:round-to-0} forces heavy items to be heavy in all non-zero dimensions.

Now we will use different transformations on heavy and light items.

\begin{transformation}
\label{trn:round-up-heavy}
For a dense item $i$, if $v_j(i) > \epsLarge$, round up $v_j(i)$ to a multiple of $\epsLarge\eps/8$.
\end{transformation}
\begin{lemma}
\label{lem:round-up-heavy}
Let $J$ be a packing of items into a bin that is $\mu$-slacked, for some $\mu \ge \eps/8$.
Let $J'$ be the packing obtained by applying \cref{trn:round-up-heavy}
to dense items in $J$. Then $J'$ is a $(\mu-\eps/8)$-slacked packing.
\end{lemma}
\begin{proof}
\textbf{Case 1}: $\forall j \in [d], v_j(J') \le 1-\mu$:\\
For each $j \in [d]$, there are less than $\epsLarge^{-1}$ items $i$ in $J$
such that $v_j(i) > \epsLarge$. For each such item, $v_j(i)$ increases by less than $\epsLarge\eps/8$.
Therefore, $v_j(J') < v_j(J) + \eps/8$. Therefore, $J$ is $(\mu - \eps/8)$-slacked.

\textbf{Case 2}: $|J| = 1$:\\
$|J'|=1$, so $J'$ is $\mu$-slacked.

\textbf{Case 3}: $|J| = 2$ and $J$ only has dense items
and $\forall i \in J, 1/2 - \mu \le v_{\max}(i) \le 1/2$:\\
$\epsLarge^{-1}\eps^{-1} \in \mathbb{Z}$, so $1/2$ is a multiple of $\epsLarge\eps/8$.
Therefore, for each item $i$, $v_{\max}(i)$ increases to at most $1/2$.
So $J$ is $\mu$-slacked.
\end{proof}

\begin{lemma}
\label{lem:round-up-heavy-n}
The number of distinct heavy items after \cref{trn:round-up-heavy} is at most
\[ n_{\hwc} \defeq \left(\frac{8}{\eps}\left(\frac{1}{\epsLarge}-1\right)\right)^d. \]
\end{lemma}
\begin{proof}
This is because large vector dimensions are rounded to at most
$8/\epsLarge\eps - 8/\eps$ distinct values.
\end{proof}

\begin{transformation}
\label{trn:round-up-light}
For a dense item $i$, if $v_{\max}(i) \le \epsSmall$, then for each $j \in [d]$,
round up $v_j(i)/v_{\max}(i)$ to a power of $1/(1+\eps/8)$ if $v_j(i) > 0$.
\end{transformation}
\begin{lemma}
\label{lem:round-up-light}
Let $J$ be a packing of items into a bin that is $\mu$-slacked, for some
$\eps/8 \le \mu \le \eps$.
Let $J'$ be the packing obtained by applying \cref{trn:round-up-light}
to dense items in $J$. Then $J'$ is a $(\mu-\eps/8)$-slacked packing.
\end{lemma}
\begin{proof}
\textbf{Case 1}: $\forall j \in [d], v_j(J') \le 1-\mu$:\\
For each $j \in [d]$, $v_j(i)$ increases by a factor of at most $1+\eps/8$.
So,
\[ v_j(J') \le v_j(J)\left(1+\frac{\eps}{8}\right) \le v_j(J) + \frac{\eps}{8}. \]
Therefore, $J'$ is $(\mu - \eps/8)$-slacked.

\textbf{Case 2}: $|J| = 1$:\\
$v_{\max}(i) \le 1$ after the transformation, so the packing is valid and $|J'| = 1$.
Therefore, $J'$ is $\mu$-slacked.

\textbf{Case 3}: $|J| = 2$ and $J$ only has dense items
and $\forall i \in J, 1/2 - \mu \le v_{\max}(i) \le 1/2$:\\
Since $\epsSmall \le 1/2 - \eps \le 1/2-\mu$ and
the transformation only applies when $v_{\max}(i) \le \epsSmall$,
$J$ remains the same after the transformation.
\end{proof}

\begin{lemma}
\label{lem:round-up-light-n}
After \cref{trn:round-up-light},
the number of distinct values of the weight vector of light items is at most
\[ \ceil{\frac{\ln(8d/\eps)}{\ln(1+\eps/8)}}^{d-1}
\le \ceil{\frac{8+\eps}{\eps}\ln\left(\frac{8d}{\eps}\right)}^{d-1} \defeq n_{\lwc} \]
\end{lemma}
\begin{proof}
This is because $v_j(i)/v_{\max}(i)$ is lower-bounded by $\eps/8d$
because of \cref{trn:round-to-0}.
\end{proof}

\begin{transformation}[Weight-rounding]
\label{trn:wround}
For a set $I$ of items, weight-rounding is the process of applying
\cref{trn:wround-non-dense,trn:round-to-0,trn:round-up-heavy,trn:round-up-light} to all items.
A set $I$ of items is said to be \emph{weight-rounded} iff $I$ is invariant under
\cref{trn:wround-non-dense,trn:round-to-0,trn:round-up-heavy,trn:round-up-light}.
\end{transformation}

\subsection{Rounding the Other Side}
\label{sec:rbbp-2d-extra:other-side}

So far, for each item, we seen how to round their weights and how to round one geometric dimension.
In this subsection, we will see how to use linear grouping to round the other geometric dimension.
We will show that after this operation, the items will belong to a constant number of
homogeneous classes (see Condition C1.3 in \cref{sec:rna:round}).

\subsubsection{Rounding Geometric Dimensions of Non-Dense Items}

\begin{transformation}[Linear grouping]
\label{trn:lingroup}
Suppose we are given a packing of items $I$ into $m$ bins,
where $I$ is invariant under \cref{trn:wround-non-dense} and
each bin satisfies \cref{prop:hrnd}\ref{item:hrnd:geom}.

Partition the big items in $I$ by their width and weight class.
Partition the tall non-dense items in $I$ by their weight class.
The number of partitions is constant by \cref{prop:hrnd}\ref{item:hrnd:geom}
and \cref{lem:num-weight-classes}.
Let $\deltaLG \defeq \eps\epsLarge/(d+1)$ (so $\deltaLG^{-1} \in \mathbb{Z}$).

For each partition $S$ of big items, do the following:
\begin{enumerate}
\item Order the items in $S$ in non-increasing order of height.
\item Let $k \defeq \floor{\deltaLG|S|}+1$.
    Let $S_1$ be the first $k$ items, $S_2$ be the next $k$ items, and so on, till $S_T$,
    where $T = \ceil{|S|/k} \le 1/\deltaLG$.
    For $t \in [T]$, $S_t$ is called the $t\Th$ linear group of $S$.
    The first item in $S_t$ is called the \emph{leader} of $S_t$, denoted as $\leader(S_t)$.
\item Increase the height of each item in $S_1$ to $h(\leader(S_1))$.
    Unpack the items $S_1 - \{\leader(S_1)\}$.
\item For each $t \in [T] - \{1\}$ and each $j \in [|S_t|] - \{1\}$,
    let $i$ be the $j\Th$ item in $S_t$ and let $i'$ be the $j\Th$ item in $S_{t-1}$.
    Note that $h(i') \ge h(\leader(S_t)) \ge h(i)$.
    Increase $h(i)$ to $h(\leader(S_t))$ and pack $i$ where $i'$ was packed.
    Since $i$ has the same width and weights as $i'$,
    the geometric constraints are not violated
    and the total weights of the bin do not increase.
    The number of distinct heights in $S$ now becomes at most $T \le 1/\deltaLG$.
\end{enumerate}

For each partition $S$ of tall items, do the following:
\begin{enumerate}
\item Order the items in $S$ in non-increasing order of height and arrange them side-by-side
    on the real line, starting from the origin.
\item Let the total width of $S$ be $W$.
    Let $S_t$ be the items in the interval $[(t-1)\deltaLG W, t \deltaLG W]$.
    Slice the items if they lie on the boundaries of the interval.
    $S_t$ is called the $t\Th$ linear group of $S$.
    The first item in $S_t$ is called the leader of $S_t$, denoted as $\leader(S_t)$.
\item Increase the height of each item in $S_1$ to $h(\leader(S_1))$. Then unpack $S_1$.
\item For each $t \in [1/\deltaLG] - \{1\}$,
    move the items in $S_t$ to the space occupied by $S_{t-1}$
    (items in $S_t$ may need to be sliced for this)
    and increase the height of each item $i \in S_t$ to $h(\leader(S_t))$.
    This doesn't violate geometric constraints since $S_t$ and $S_{t-1}$
    have the same total width and this doesn't increase the total weights of the bin
    because all items in $S$ have the same weight class.
    The number of distinct heights in $S$ now becomes at most $1/\deltaLG$.
\end{enumerate}

We can extend the definition of this transformation to bins satisfying
\cref{prop:vrnd}\ref{item:vrnd:geom} by swapping vertical and horizontal directions.
(Partition big items by height and weight class and partition wide items by weight class.
Then round up the width of items in each partition using the above techniques.)
\end{transformation}

\begin{lemma}
\label{lem:bot-slack-pack}
Let $J$ be a set of big and tall items and let $\mu \in [0, 1)$ be a constant.
Define $\Span'(i) \defeq \max\left(w(i), \min\left(\frac{v_{\max}(i)}{1-\mu}, 1\right)\right)$.
Then $\sum_{i \in J} \Span'(i) \le 1$ implies $J$ can be packed into
a $\mu$-slacked bin where all items touch the bottom of the bin.
\end{lemma}
\begin{proof}
$\sum_{i \in J} w(i) \le \sum_{i \in J} \Span'(i) \le 1$.

$\forall i \in J, \Span'(i) > 0$.
So if $\Span'(i) = 1$ for some $i \in J$, then $|J| = 1$.
So $J$ can be packed into a bin, and the bin is $\mu$-slacked since $|J| = 1$.

Now let $\Span'(i) < 1$ for all $i \in J$. So $v_{\max}(i) < 1-\mu$ and $\forall j \in [d]$,
\[ v_j(J) = \sum_{i \in J} v_j(i) \le (1-\mu) \sum_{i \in J} \frac{v_{\max}(i)}{1-\mu}
\le (1-\mu) \sum_{i \in J} \Span'(i) \le 1-\mu. \]
Therefore, $J$ can be packed into a $\mu$-slacked bin.
\end{proof}

\begin{lemma}
\label{thm:lg-repack}
Suppose we are given a packing of items $I$ into $m$ bins,
where $I$ is invariant under \cref{trn:wround-non-dense} and
each bin satisfies \cref{prop:hrnd}\ref{item:hrnd:geom}.
Let $U$ be the items unpacked by linear grouping (\cref{trn:lingroup}).
Then $U$ can be repacked into $\frac{2\eps}{1-\eps}m + 1$ number of
$\eps$-slacked bins that satisfy \cref{prop:hrnd}.
\end{lemma}
\begin{proof}
Define $\Span'(i) \defeq \max\left(w(i), \min\left(\frac{v_{\max}(i)}{1-\eps}, 1\right)\right)$.
Let $K$ be the set of big and tall non-dense items in $I$.
For any $J \subseteq K$, define $\Span'(J) \defeq \sum_{i \in J} \Span'(i)$.

Interpret each item $i \in U$ as a 1D item of size $\Span'(i)$.
Order the items such that big items in $U$ appear before tall non-dense items in $U$.
Pack them on the bottom of new bins using the Next-Fit algorithm.
By \cref{lem:bot-slack-pack}, they will require at most $2\Span'(U) + 1$ $\eps$-slacked bins.
These bins satisfy \cref{prop:hrnd}\ref{item:hrnd:geom}
since the width of all big items in $U$ is a multiple of $\epsLarge^2/4$,
and they satisfy \cref{prop:hrnd}\ref{item:hrnd:dense} since $U$
only contains non-dense items.

In \cref{trn:lingroup}, we partitioned all big items in $I$ by width and weight class.
Let $S \subseteq I$ be one such partition.
Given the way we created the groups $S_1, S_2, \ldots$, we get
$|S \cap U| \le \floor{\deltaLG|S|}$.
Since all items in $S$ have the same width and weights,
$\Span'(i)$ is the same for each $i \in S$. Therefore,
\[ \Span'(S \cap U) = \Span'(i)|S \cap U|
\le \Span'(i)\floor{\deltaLG |S|}
\le \deltaLG \Span'(S). \]
In \cref{trn:lingroup}, we partitioned all tall non-dense items in $I$ by weight class.
Let $S \subseteq I$ be one such partition.
Given the way we created the groups $S_1, S_2, \ldots$, we get $w(S \cap U) = \deltaLG w(S)$.
All items in $S$ have the same weights-to-width ratio,
which is at most $1/\epsLarge^2$ by \cref{obs:non-dense-vinc}\ref{item:non-dense-vinc:tall}.
Since $\epsSmall \le \epsLarge^2(1-\eps)$,
we get $v_j(i) \le 1-\eps$ for all $i \in S$,
so $\Span'(i)/w(i)$ is the same for each $i \in S$.
Let that common ratio be $\alpha$. Then,
\[ \Span'(S \cap U) = \alpha w(S \cap U)
\le \alpha \deltaLG w(S) = \deltaLG \Span'(S). \]
Summing over all partitions $S$, we get
\begin{equation}
\label{eqn:unpacked-vs-packed}
\Span'(U) = \sum_S \Span'(U \cap S) \le \sum_S \deltaLG \Span'(S)
    \le \deltaLG \Span'(K).
\end{equation}

For $i \in K$, we get
\begin{align*}
\frac{\Span'(i)}{\Span(i)}
\le \frac{\max\left(w(i), \frac{v_{\max}(i)}{1 - \mu}\right)}{
\max\left(w(i)h(i), v_{\max}(i)\right)}
\le \frac{\frac{1}{1-\mu}\max\left(w(i), v_{\max}(i)\right)}{
\max\left(w(i)\epsLarge, v_{\max}(i)\right)}
\le \frac{1}{(1-\mu)\epsLarge}.
\end{align*}
The last inequality follows because for big and tall items, $h(i) \ge \epsLarge$.

The number of bins used to pack $U$ is
\begin{align*}
2\Span'(U) + 1
&\le 2\deltaLG \Span'(K) + 1
\le \frac{2\deltaLG}{\epsLarge(1-\eps)} \Span(K) + 1
\\ &\le \frac{2(d+1)\deltaLG}{\epsLarge(1-\eps)}m + 1
= \frac{2\eps}{1-\eps}m + 1.
\end{align*}
The first inequality follows from \eqref{eqn:unpacked-vs-packed} and the third
inequality follows from \cref{claim:wmax-s-lim}.
\end{proof}

\begin{lemma}
\label{lem:rnd2}
Suppose we are given a packing of items $I$ into $m$ bins,
where $I$ is weight-rounded, each bin is $\mu$-slacked for some $\mu \le \eps$,
and each bin satisfies either \cref{prop:hrnd} or \cref{prop:vrnd}.
Then after applying linear grouping (\cref{trn:lingroup}) to this packing of $I$,
we get a packing of items $\Ihat$ into $m'$ bins, where all of the following hold:
\begin{itemize}
\item $\Ihat$ is a rounding-up of $I$ and contains a constant number of homogeneous classes
    (see Condition C1.3 in \cref{sec:rna:round}).
\item Each bin in the packing of $\Ihat$ is $\mu$-slacked
    and satisfies either \cref{prop:hrnd} or \cref{prop:vrnd}.
\item ${\displaystyle m' \le \left(1 + \frac{2\eps}{1-\eps}\right)m + 2}$.
\end{itemize}
\end{lemma}
\begin{proof}[Proof sketch]
Follows from the definition of linear grouping and \cref{thm:lg-repack}.
Note that we apply linear grouping separately to bins satisfying \cref{prop:hrnd}
and bins satisfying \cref{prop:vrnd}.
\end{proof}

\subsubsection{Coarse and Fine Partitioning}

Our approach so far has been to start from an optimal packing of items and show how to
modify it to obtain an approximately-optimal structured packing of a rounded instance.
However, the rounding algorithm must round items without knowing the optimal packing.
To design such an algorithm, we first need to introduce additional concepts:
\emph{coarse partitioning} and \emph{fine partitioning}.

At a high level, our rounding algorithm first partitions the items by weight classes
to get a \emph{coarse partitioning}. It then further partitions the coarse partitions
to get a \emph{fine partitioning}. It then rounds up the geometric dimensions of items
in each fine partition to make that partition homogeneous.

We will first formally define coarse and fine partitioning.
We will then restate \cref{thm:slack-pack,lem:rnd2} using the language of fine partitioning.
Then in \cref{sec:round}, we will see an algorithm for computing a fine partitioning of $I$.

\begin{tightemize}
\item Let $B(I)$ be the set of big items in $I$.
\item Let $W(I)$ be the set of wide non-dense items in $I$.
\item Let $H(I)$ be the set of tall non-dense items in $I$.
\item Let $S(I)$ be the set of small non-dense items in $I$.
\item Let $D^{l,1}(I)$ be the set of light dense items in $I$ that are either tall or small.
\item Let $D^{l,2}(I)$ be the set of light dense wide items in $I$.
\item Let $D^{h,1}(I)$ be the set of heavy dense items in $I$ that are either tall or small.
\item Let $D^{h,2}(I)$ be the set of heavy dense wide items in $I$.
\end{tightemize}
When the set of items $I$ is clear from context, we will use
$B$, $W$, $H$, $S$, $D^{l,1}$, $D^{l,2}$, $D^{h,1}$, $D^{h,2}$ to refer to these sets.

\begin{definition}[Coarse partitioning]
Let $I$ be a weight-rounded instance. Partition items $I$ by their weight classes.
Then for each partition containing dense items, split that partition into 2 partitions:
one containing only tall and small items and the other containing only wide items.
The resulting partitioning is called a coarse partitioning of $I$.

We number the coarse partitions in $B$ arbitrarily from $1$ onwards.
There will be at most $n_{\bwc}$ such partitions by \cref{lem:num-weight-classes}.
Denote the $p\Th$ coarse partition by $B_p$.

Similarly, denote the $p\Th$ coarse partition
\begin{tightemize}
\item in $W$ by $W_p$, where $p \in [n_{\wwc}]$.
\item in $H$ by $H_p$, where $p \in [n_{\wwc}]$.
\item in $S$ by $S_p$, where $p \in [n_{\swc}]$.
\item in $D^{l,1}$ by $D^{l,1}_p$, where $p \in [n_{\lwc}]$.
\item in $D^{l,2}$ by $D^{l,2}_p$, where $p \in [n_{\lwc}]$.
\item in $D^{h,1}$ by $D^{h,1}_p$, where $p \in [n_{\hwc}]$.
\item in $D^{h,2}$ by $D^{h,2}_p$, where $p \in [n_{\hwc}]$.
\end{tightemize}
\end{definition}

\begin{observation}
\label{obs:coarse-part}
There is a unique coarse partitioning of $I$.
Furthermore, the unique coarse partitioning can be found in $O(|I|)$ time.
\end{observation}

In \cref{thm:slack-pack,lem:rnd2}, widths of wide and big items are rounded.
The rounding is different for \cref{thm:slack-pack,lem:rnd2}:
In \cref{thm:slack-pack}, we round the widths of some items to multiples of $\epsLarge^2/4$
so that the bin satisfies \cref{prop:hrnd}\ref{item:hrnd:geom},
and in \cref{lem:rnd2}, we round the widths of items in bins satisfying
\cref{prop:vrnd}\ref{item:vrnd:geom} using linear grouping.
To get a rounding algorithm, we have to guess whether the bin of a wide or big item
will satisfy \cref{prop:hrnd} or \cref{prop:vrnd}.
We will capture these guesses in the fine partitioning.

\begin{definition}[Fine partitioning]
\label{defn:fine-part}
Let $Q \defeq \mathbb{Z} \cap \left[\frac{4}{\epsLarge}+1, \frac{4}{\epsLarge^2}\right]$,
$R \defeq \left\{1, 2, \ldots, \frac{1}{\deltaLG}\right\}$,
$Q_q \defeq \left\{x \in \mathbb{R}: (q-1)\frac{\epsLarge^2}{4} < x
    \le q\frac{\epsLarge^2}{4} \right\}$.

Given a coarse partitioning of a set $I$ of items,
let $(B_p^w, B_p^h)$ be a partitioning of $B_p$,
$(W_p^w, W_p^h)$ be a partitioning of $W_p$
and $(H_p^w, H_q^h)$ be a partitioning of $H_p$.

\begin{itemize}
\item $B_p^w$ is partitioned into sets $\{ B_{p,q,r}^w: q \in Q, r \in R \}$
where $i \in B^w_{p,q,r} \implies w(i) \in Q_q$.
\item $B_p^h$ is partitioned into sets $\{ B_{p,q,r}^h: q \in Q, r \in R \}$
where $i \in B^h_{p,q,r} \implies h(i) \in Q_q$.
\item $W_p^w$ is partitioned into sets $\{ W_{p,q}^w: q \in Q \}$
where $i \in W^w_{p,q} \implies w(i) \le q\epsLarge^2/4$.
\item $W_p^h$ is partitioned into sets $\{ W_{p,r}^h: r \in R \}$.
\item $H_p^w$ is partitioned into sets $\{ H_{p,r}^w: r \in R \}$.
\item $H_p^h$ is partitioned into sets $\{ H_{p,q}^h: q \in Q \}$
where $i \in H^h_{p,q} \implies h(i) \le q\epsLarge^2/4$.
\end{itemize}

A fine partitioning of $I$ is any partitioning of $I$ into sets of the form
$B^w_{p,q,r}$, $B^h_{p,q,r}$, $W^w_{p,q}$, $W^h_{p,r}$, $H^w_{p,r}$, $H^h_{p,q}$,
$S_p$, $D^{l,1}_p$, $D^{l,2}_p$, $D^{h,1}_p$, $D^{h,2}_p$.
\end{definition}

Note that for a given set $I$ of items, there can be multiple fine partitionings.

Given a fine partitioning, we use the `$*$' character in superscript or subscript to denote
the union of some partitions. For example, $B^w_{p,*,r} \defeq \bigcup_q B^w_{p,q,r}$
and $W^w_{*,*} \defeq \bigcup_{p,q} W^w_{p,q}$, and $D^{*,1}_p \defeq D^{l,1}_p \cup D^{h,1}_p$.

When item rotations are allowed, the fine partitioning includes information
on which items to rotate, and we can assume \wLoG{} that
$H(I) = D^{*,2} = B^h_{*,*,*} = W^h_{*,*} = H^h_{*,*} = \{\}$.

\begin{transformation}
\label{trn:round-geom}
Given a fine partitioning of $I$, execute the following operations:
\begin{itemize}
\item $\forall i \in B^w_{*,q,*} \cup W^w_{*,q}$, increase $w(i)$ to $q\epsLarge^2/4$.
\item $\forall i \in B^h_{*,q,*} \cup H^h_{*,q}$, increase $h(i)$ to $q\epsLarge^2/4$.
\item $\forall i \in B^w_{p,q,r}$, increase $h(i)$ to $\max_{i \in B^w_{p,q,r}} h(i)$.
\item $\forall i \in B^h_{p,q,r}$, increase $w(i)$ to $\max_{i \in B^h_{p,q,r}} w(i)$.
\item $\forall i \in W^h_{p,r}$, increase $w(i)$ to $\max_{i \in W^h_{p,r}} w(i)$.
\item $\forall i \in H^w_{p,r}$, increase $h(i)$ to $\max_{i \in H^w_{p,r}} h(i)$.
\end{itemize}
\end{transformation}

The number of fine partitions is constant and after applying \cref{trn:round-geom},
each partition is homogeneous.

\begin{definition}[Semi-structured packing of fine partitioning]
\label{defn:semi-struct}
Suppose we are given a fine partitioning of items $I$.
A packing of items $J \subseteq I$ into a bin is said to be `division-1 semi-structured'
with respect to the fine partitioning iff
$J$ doesn't contain items from $B^h_{*,*,*}$, $W^h_{*,*}$, $H^h_{*,*}$ and $D^{*,2}$
and $J$ satisfies \cref{prop:hrnd}.

A packing of items $J \subseteq I$ into a bin is said to be `division-2 semi-structured'
with respect to the fine partitioning iff
$J$ doesn't contain items from $B^w_{*,*,*}$, $W^w_{*,*}$, $H^w_{*,*}$ and $D^{*,1}$
and $J$ satisfies \cref{prop:vrnd}.

Packing of items into bins is called semi-structured iff
each bin is either division-1 semi-structured or division-2 semi-structured.
\end{definition}

\begin{definition}[Balanced fine partitioning]
\label{defn:bal-fine-part}
A fine partitioning is said to be \emph{balanced} iff it satisfies all of the following conditions:
\begin{itemize}
\item $\forall p, \forall r, h(W^h_{p,r}) = \deltaLG h(W^h_{p,*})$
\item $\forall p, \forall r, w(H^w_{p,r}) = \deltaLG w(H^w_{p,*})$
\item $\forall p, \forall q$, the sets $\{B^w_{p,q,r}: \forall r\}$ can be obtained from
    $B^w_{p,q,*}$ by ordering the items in $B^w_{p,q,*}$ in non-increasing order of
    height (breaking ties arbitrarily) and putting the first $k$ items in $B^w_{p,q,1}$,
    the next $k$ items in $B^w_{p,q,2}$, and so on,
    where $k \defeq \floor{\deltaLG|B^w_{p,q,*}|} + 1$.
\item $\forall p, \forall q$, the sets $\{B^h_{p,q,r}: \forall r\}$ can be obtained from
    $B^h_{p,q,*}$ by ordering the items in $B^h_{p,q,*}$ in non-increasing order of
    width (breaking ties arbitrarily) and putting the first $k$ items in $B^h_{p,q,1}$,
    the next $k$ items in $B^h_{p,q,2}$, and so on,
    where $k \defeq \floor{\deltaLG|B^h_{p,q,*}|} + 1$.
\end{itemize}
\end{definition}

We now restate \cref{thm:slack-pack,lem:rnd2} in terms of fine partitioning.

\begin{lemma}
\label{lem:rnd2.5}
Let $I$ be a set of items and $\widehat{I}$ be the items obtained by
\weightRoundingHyp{} $I$.
Then there exists a balanced fine partitioning of a slicing of $\widehat{I}$
such that after applying \cref{trn:round-geom} to $\widehat{I}$,
there is a semi-structured $(5\eps/8)$-slacked fractional packing of $\widehat{I}$ into
$\left(1 + \frac{2\eps}{1-\eps}\right)(a \opt(I) + b) + 2$ bins.
Here $a$ and $b$ are as defined in \cref{table:slack-pack-ab} in \cref{thm:slack-pack}.
\end{lemma}
\begin{proof}
By \cref{thm:slack-pack}, we can round up the width of some big and wide
non-dense items in $I$ to the nearest multiple of $\epsLarge^2/4$ and round up the height
of some big and tall non-dense items in $I$ to the nearest multiple of $\epsLarge^2/4$
and then pack $I$ into $a\opt(I) + b$ $\eps$-slacked bins
such that each bin satisfies either \cref{prop:hrnd} or \cref{prop:vrnd}.
Let $\mathcal{B}$ be such a bin packing of $I$.
By \cref{lem:wround-non-dense-slackness,lem:round-up-heavy,lem:round-up-light},
$\mathcal{B}$ gives us a $(5\eps/8)$-slacked packing of $\widehat{I}$.

Call the bins in $\mathcal{B}$ that satisfy \cref{prop:hrnd} division-1 bins.
Call the rest of the bins division-2 bins. The items whose width needs to be rounded up
to a multiple of $\epsLarge^2/4$ are the big and wide items in division-1 bins
and the items whose height needs to be rounded up to a multiple of $\epsLarge^2/4$
are the big and tall items in division-2 bins. No other items need to have their width or height
rounded up in the packing $\mathcal{B}$ produced by \cref{thm:slack-pack}.

Let $\widehat{I}^w$ and $\widehat{I}^h$ be the items of $\widehat{I}$
in division-1 bins and division-2 bins respectively.

We can compute the coarse partitioning of $\widehat{I}$. Define
\begin{align*}
B_p^w &\defeq B_p \cap \widehat{I}^w
& W_p^w &\defeq W_p \cap \widehat{I}^w
& H_p^w &\defeq H_p \cap \widehat{I}^w
\\ B_p^h &\defeq B_p \cap \widehat{I}^h
& W_p^h &\defeq W_p \cap \widehat{I}^h
& H_p^h &\defeq H_p \cap \widehat{I}^h
\end{align*}
Define
\begin{tightemize}
\item $B^w_{p,q} \defeq \{i \in B^w_p: (q-1)\epsLarge^2/4 < w(i) \le q\epsLarge^2/4 \}$.
\item $B^h_{p,q} \defeq \{i \in B^h_p: (q-1)\epsLarge^2/4 < h(i) \le q\epsLarge^2/4 \}$.
\item $W^w_{p,q} \defeq \{i \in W^w_p: (q-1)\epsLarge^2/4 < w(i) \le q\epsLarge^2/4 \}$.
\item $H^h_{p,q} \defeq \{i \in H^h_p: (q-1)\epsLarge^2/4 < h(i) \le q\epsLarge^2/4 \}$.
\end{tightemize}
Define
\begin{tightemize}
\item $B^w_{p,q,r}$ as the $r\Th$ linear group of $B^w_{p,q}$ (see \cref{trn:lingroup}).
\item $B^h_{p,q,r}$ as the $r\Th$ linear group of $B^h_{p,q}$.
\item $W^h_{p,r}$ as the $r\Th$ linear group of $W^h_p$.
\item $H^w_{p,r}$ as the $r\Th$ linear group of $H^w_p$.
\end{tightemize}
This is how we get a fine partitioning of a slicing of $\widehat{I}$.

As per \cref{lem:rnd2}, on applying \cref{trn:round-geom} to $\widehat{I}$,
the resulting instance can be sliced and packed into
$\left(1 + \frac{2\eps}{1-\eps}\right)(a \opt(I) + b) + 2$ number of
$(5\eps/8)$-slacked bins.
\end{proof}

\subsection{Rounding Algorithm}
\label{sec:round}

Let $I$ be a set of weight-rounded items.
To use \cref{lem:rnd2.5} to get an approximately-optimal packing of items $I$,
we would like to iterate over all balanced fine partitionings of slicings of $I$.
However, even if we don't consider slicings, doing that will take exponential time,
since for each big, wide and tall item, we need to decide whether to designate it
as a division-1 item or a division-2 item.

We can get around this problem by iterating over a polynomial-sized set
$\mathcal{S}_{\Pi}$ of fine partitionings such that
each balanced fine partitioning $\Pcal$ of a slicing of $I$ is `close to'
a fine partitioning $\Pcalhat$ in $\mathcal{S}_{\Pi}$.
We will now define what we mean by `close to'.

\begin{definition}[Predecessor of a set of items]
Let $I_1$ and $I_2$ be sets of items. Interpret each item $i \in I_1$
as a bin whose geometric dimensions are the same as that of $i$
and whose weight capacities are the same as the weights of $i$.
$I_2$ is said to be a predecessor of $I_1$ (denoted as $I_2 \preceq I_1$)
iff $I_2$ can be sliced and packed into $I_1$.
\end{definition}

We will design a polynomial-time algorithm $\iterFineParts$ that takes as input
a weight-rounded set $I$ of items and outputs a set $\mathcal{S}_{\Pi}$ of pairs such that
for each balanced fine partitioning $\Pcal$ of a slicing of $I$,
there exists a pair
$(D, \Pcalhat) \in \mathcal{S}_{\Pi}$ such that all of these conditions hold:
\begin{itemize}
\item $\Pcalhat$ is a fine partitioning of $I-D$.
\item After applying \cref{trn:round-geom}, each partition in $\Pcalhat$
    is a predecessor of the corresponding partition in $\Pcal$.
\item $D$ is a set of non-dense items (called discarded items) such that $\Span(D)$ is small
    compared to $\Span(I)$.
\end{itemize}

\subsubsection{Big Items}

Let $\partBig$ be an algorithm that takes a coarse partition $B_p$ of big items as input,
and outputs multiple fine partitionings of $B_p$.
We can use $\partBig$ as a subroutine in $\iterFineParts$.

To design $\partBig$, we will guess the cardinality of sets $B^w_{p,q,r}$ and $B^h_{p,q,r}$.
We will then guess the maximum height in $B^w_{p,q,r}$ and the maximum width in $B^h_{p,q,r}$.
Then for each guess, we will solve a max-flow problem to check if items in $B_p$
can be assigned to these sets. We skip the details here since this approach
is similar to that of Pr\"adel (see Section 3.3.1 in \cite{pradel-thesis}).

Formally $\partBig(B_p)$ outputs a set of pairs of the form $(\{\}, \Pcalhat)$,
where $\Pcalhat$ is supposed to be a fine partitioning of $B_p$.

\begin{claim}
\label{claim:part-big-time}
$\partBig(B_p)$ generates
$O\left(n^{2|Q|(1/\deltaLG+1) - 1}\right) = O\left(n^{8(d+1)/\eps\epsLarge^3}\right)$
values, where $n \defeq |B_p|$, and the running time per value is $O(n^2/\eps\epsLarge)$.
\end{claim}

\begin{claim}
\label{lem:part-big}
Let $\Pcal \defeq \{B^w_{p,q,r}: \forall q, \forall r\} \cup \{B^h_{p,q,r}: \forall q, \forall r\}$
be a balanced fine partitioning of $B_p$.
Then there is an output $(\{\}, \Pcalhat)$ of $\partBig(B_p)$ where
$\Pcalhat \defeq \{\widehat{B}^w_{p,q,r}: \forall q, \forall r\}
\cup \{\widehat{B}^h_{p,q,r}: \forall q, \forall r\}$ such that
$\Pcalhat$ is a fine partitioning of $B_p$ and
after applying \cref{trn:round-geom},
\[ \forall q, \forall r, \widehat{B}^w_{p,q,r} \preceq B^w_{p,q,r}
\textrm{ and } \widehat{B}^h_{p,q,r} \preceq B^h_{p,q,r}. \]
\end{claim}

$\partBig$ for the rotational case is similar to the non-rotational case:
When rotations are allowed, assume \wLoG{} that $B^h_{*,*,*} = \{\}$.
We will guess the cardinality and maximum height in sets $B^w_{p,q,r}$.
Then for each guess, we will solve a max-flow problem to check if
items in $B_p$ can be assigned to these sets, possibly after rotating some items.

\subsubsection{Wide and Tall Items}

Let $\partWide$ be an algorithm that takes a coarse partition $W_p$
of wide items as input, and outputs multiple fine partitionings of $W_p$.
We can use $\partWide$ as a subroutine in $\iterFineParts$.

Let $\Pcal \defeq \{W_{p,q}^w: \forall q\} \cup \{W_{p,r}^h: \forall r\}$ be
a balanced fine partitioning of a slicing of $W_p$.
We want $\partWide$ to find a fine partitioning
$\Pcalhat \defeq \{\widehat{W}_{p,q}^w: \forall q\} \cup \{\widehat{W}_{p,r}^h: \forall r\}$
of a large subset of $W_p$ such that
after applying \cref{trn:round-geom} to $\Pcal$ and $\Pcalhat$,
every fine partition in $\Pcalhat$ is a predecessor of
the corresponding fine partition in $\Pcal$.

For any $J \subseteq W_p$, define
$h(J) \defeq \sum_{i \in J} h(i)$ and $w(J) \defeq \max_{i \in J} w(i)$.
We will create a rectangular box for each fine partition
and then try to pack a large subset of $W_p$ into these boxes.
Let $Q \defeq \mathbb{Z} \cap \left[\frac{4}{\epsLarge}+1, \frac{4}{\epsLarge^2}\right]$.
For each $q \in Q$, let $s^w_q$ be a box
of width $q\epsLarge^2/4$ and height $h(W^w_{p,q})$.
For each $r \in [1/\deltaLG]$, let $s^h_r$ be a box
of width $w(W^h_{p,r})$ and height $h(W^h_{p,r})$.
Since we don't know $W^h_{p,r}$ and $W^w_{p,q}$,
we will guess the value of $w(W^h_{p,r})$ and
we will guess very close lower bounds on $h(W^w_{p,q})$ and $h(W^h_{p,r})$.
We will then try to pack most of the items from $W_p$ into these boxes.

Let $\widehat{W}^w_{p,q}$ be the items packed into $s^w_q$ and
let $\widehat{W}^h_{p,r}$ be the items packed into $s^h_r$.
Then $\Pcalhat \defeq \{\widehat{W}_{p,q}^w: \forall q\} \cup \{\widehat{W}_{p,r}^h: \forall r\}$
is a fine partitioning of a large subset of $W_p$.
After applying \cref{trn:round-geom} to $\Pcal$ and $\Pcalhat$,
each item in $W^w_{p,q}$ and $\widehat{W}^w_{p,q}$ has width $q\epsLarge^2/4$.
Since $h(\widehat{W}^w_{p,q}) \le h(s^w_q) \le h(W^w_{p,q})$,
we get $\widehat{W}^w_{p,q} \preceq W^w_{p,q}$ after \cref{trn:round-geom}.
We can similarly prove that $\widehat{W}^h_{p,r} \preceq W^h_{p,r}$.
Therefore, $\Pcalhat$ is a suitable fine partitioning.

The details on how to approximately guess the size of boxes
and how to pack a large subset of items into boxes can be deduced from
section 3.3.1 of \cite{pradel-thesis}.

Formally, $\partWide(W_p)$ outputs a set of pairs of the form $(D, \Pcalhat)$,
where items in $D$ are called \emph{discarded} items and $\Pcalhat$ is supposed to be
a fine partitioning of $W_p - D$.

\begin{claim}
\label{lem:part-wide-discard}
For every output $(D, \Pcalhat)$ of $\partWide(W_p)$,
\[ h(D) \le (3\epsSmall)\left( \frac{d+1}{\eps\epsLarge} + \frac{4}{\epsLarge^2}
    - \frac{4}{\epsLarge} \right). \]
\end{claim}

\begin{claim}
\label{lem:part-wide}
Let $\Pcal \defeq \{W_{p,q}^w: \forall q\} \cup \{W_{p,r}^h: \forall r\}$ be
a balanced fine partitioning of a slicing of $W_p$.
Then for some output $(D, \Pcalhat)$ of $\partWide(W_p)$,
$\Pcalhat$ is a fine partitioning of $W_p - D$
and after applying \cref{trn:round-geom} to $\Pcal$ and $\Pcalhat$,
\[ (\forall q, \widehat{W}^w_{p,q} \preceq W^w_{p,q})
\textrm{ and } (\forall r, \widehat{W}^h_{p,r} \preceq W^h_{p,r}). \]
\end{claim}

\begin{claim}
\label{claim:part-wide-time}
Let there be $n$ items in $W_p$. Let
$n_q \defeq 4/\epsLarge^2 - 4/\epsLarge$.
Then $\partWide(W_p)$ outputs at most $\deltaLG n^{n_q + 1 + 1/\deltaLG}$
distinct values. The running time per value is $O(n\log n)$.
\end{claim}

$\partWide$ can analogously be used for sub-partitioning coarse partitions of tall items.

When item rotations are allowed, $\partWide(W_p)$ gives us $H^w_{p,r}$ instead of $W^h_{p,r}$.

\subsubsection{Rounding Algorithm}

Based on $\partBig$ and $\partWide$, we define an algorithm $\iterFineParts$ (\cref{algo:ifp})
that returns a set of fine partitionings.
We then use $\iterFineParts$ to design the algorithm $\round$ (\cref{algo:round}).

\begin{algorithm}[!ht]
\caption{$\iterFineParts(I)$: $I$ is a set of weight-rounded items.
Returns a set of pairs of the form $(D, \Pcalhat)$,
where $D$ is a subset of items to discard
and $\Pcalhat$ is a \finePartHyp{} of $I-D$.}
\label{algo:ifp}
\begin{algorithmic}[1]
\State $\texttt{outputs} = \{\}$
\State $\{B_p: \forall p\} \cup \{W_p: \forall p\} \cup \{H_p: \forall p\}
    \cup \textrm{(small and dense partitions)}
    = \hyperref[obs:coarse-part]{\operatorname{coarse-partition}}(I)$
\State $\texttt{iters} = {\displaystyle \cartProd_{p=1}^{\nbwcHyp} \partBigHyp(B_p)
        \times \cartProd_{p=1}^{\nwwcHyp} \partWideHyp(W_p)
        \times \cartProd_{p=1}^{\nwwcHyp} \partWideHyp(H_p)}$
\For{$\mathcal{L} \in \texttt{iters}$}
    \Comment{$\mathcal{L}$ is a list of pairs}
    \State $\Pcalhat = \textrm{(small and dense partitions)}$
    \State $D = \{\}$
    \For{$(D_j, \Pcalhat_j) \in \mathcal{L}$}
        \State $D = D \cup D_j$
        \State Include partitions of $\Pcalhat_j$ into $\Pcalhat$.
    \EndFor
    \State $\outputAdd((D, \Pcalhat))$
\EndFor
\State \Return \texttt{outputs}
\end{algorithmic}
\end{algorithm}

\begin{algorithm}[!ht]
\caption{$\round(I, \eps)$: Returns a set of pairs of the form $(\widetilde{I}, D')$,
where $D'$ is a subset of items to discard and $\widetilde{I}$ is a rounding of $I-D'$.}
\label{algo:round}
\begin{algorithmic}[1]
\State $\texttt{outputs} = \{\}$
\State $\delta_0 \defeq {\displaystyle \min\left(\frac{1}{4d+1}, \frac{2\eps}{3}\right)}$
\State $(\Imed, \epsSmall, \epsLarge) = \remMedHyp(I, \eps, f, \delta_0)$
\Comment{$f$ will be described \hyperref[eqn:remmed-f]{later}}
\LineComment{Assume $\eps$ and $\epsLarge$ are passed as parameters to
    all subroutines and transformations.}
\State Let $\widehat{I}$ be the weight-rounding (\cref{trn:wround}) of $I - \Imed$.
\For{$(D, \Pcalhat) \in \iterFinePartsHyp(\widehat{I})$}
    \State Let $\widetilde{I}$ be the instance obtained by applying \cref{trn:round-geom}
        to $\widehat{I} - D$ based on the fine partitioning $\Pcalhat$.
    \State $\outputAdd((\widetilde{I}, D \cup \Imed))$.
\EndFor
\State \Return \texttt{outputs}
\end{algorithmic}
\end{algorithm}

Assume that in an output $(\widetilde{I}, D)$ of $\round(I, \eps)$,
the knowledge of the associated fine partitioning is implicitly present in $\widetilde{I}$.

\begin{lemma}[Polynomial-time]
\label{lem:round-time}
The total number of outputs of $\round(I)$ is $O(n^\gamma)$,
where there are $n$ items in $I$ and
\[ \gamma \defeq \nbwcHyp\frac{8(d+1)}{\eps\epsLarge^3}
    + 2\nwwcHyp\left(\frac{4}{\epsLarge^2} + \frac{d+1}{\eps\epsLarge}\right). \]
The time for each output is at most $O(n^2/\eps\epsLarge)$.
\end{lemma}
\begin{proof}
Follows from \cref{claim:part-big-time,claim:part-wide-time}.
\end{proof}

\begin{lemma}[Low discard]
\label{lem:round-discard}
Let $(\widetilde{I}, D')$ be an output of $\round(I, \eps)$. Then
\begin{align*}
\Span(D') &\le \eps \Span(I) + \frac{6\epsSmall\nwwcHyp}{\epsLarge^2}
    \left(\frac{d+1}{\eps\epsLarge} + \frac{4}{\epsLarge^2} - \frac{4}{\epsLarge}\right)
\\ &\le \eps \Span(I) + 6(d+5)\frac{\epsSmall\nwwcHyp}{\epsLarge^4}.
\end{align*}
\end{lemma}
\begin{proof}
Let $D' = D \cup \Imed$.
By \cref{thm:med-span}, $\Span(\Imed) \le \eps \Span(I)$.
Let $D_1$ be the wide items in $D$ and $D_2$ be the tall items in $D$.
Since all items in $D$ come from $\partWide$, $D$ only contains non-dense items
and $D = D_1 \cup D_2$.
Let
\[ \lambda \defeq n_{\wwc}(3\epsSmall)
\left(\frac{d+1}{\eps\epsLarge} + \frac{4}{\epsLarge^2} - \frac{4}{\epsLarge}\right). \]
By \cref{lem:part-wide-discard},
$h(D_1) \le \lambda$ and $w(D_2) \le \lambda$.
\[ \Span(D_1) = \sum_{i \in D_1} \max(a(i), v_{\max}(i))
\le \sum_{i \in D_1} \max\left(h(i), \frac{h(i)}{\epsLarge^2}\right)
\le \sum_{i \in D_1} \frac{h(i)}{\epsLarge^2} \le \frac{\lambda}{\epsLarge^2}. \]
Similarly, $\Span(D_2) \le \lambda/\epsLarge^2$.
\end{proof}

\begin{lemma}[Homogeneity]
\label{lem:round-homo}
Let $(\widetilde{I}, D)$ be an output of $\round(I, \eps)$.
Then the number of types of items in $\widetilde{I}$ is at most a constant:
\[ \frac{8(d+1)\nbwcHyp}{\eps\epsLarge^3}
+ 2\nwwcHyp\left(\frac{d+1}{\eps\epsLarge} + \frac{4}{\epsLarge^2}\right)
+ \nswcHyp + 2(\nlwcHyp + \nhwcHyp). \]
\end{lemma}
\begin{proof}
After applying \cref{trn:round-geom} to $\widehat{I} - D$,
all non-dense items in each fine partition have:
\begin{tightemize}
\item the same weight class.
\item the same width, if the fine partition contains only wide or big items.
\item the same height, if the fine partition contains only tall or big items.
\end{tightemize}
From the definition of fine partitioning,
the number of different partitions of non-dense items is at most
\[ \frac{8n_{\bwc}}{\epsLarge^2\deltaLG}
+ 2n_{\wwc}\left(\frac{1}{\deltaLG} + \frac{4}{\epsLarge^2}\right)
+ n_{\swc}. \]
In each division, there are $n_{\hwc}$ distinct heavy dense items and $n_{\lwc}$ weight classes
in light dense items by \cref{lem:round-up-heavy-n,lem:round-up-light-n}.
\end{proof}

\begin{table}[!ht]
\centering
\caption{Upper bound on the number of different types of items}
${\displaystyle \begin{array}[t]{*3{|>{\displaystyle}c}|}
\hline & \textrm{no. of types}
\\ \hline \textrm{Division-1 big items } (B^w_{*,*,*})
    & 4(d+1)\nbwcHyp/\eps\epsLarge^3
\\ \hline \textrm{Division-2 big items } (B^h_{*,*,*})
    & 4(d+1)\nbwcHyp/\eps\epsLarge^3
\\ \hline \textrm{Division-1 wide non-dense items } (W^w_{*,*})
    & 4\nwwcHyp/\epsLarge^2
\\ \hline \textrm{Division-2 wide non-dense items } (W^h_{*,*})
    & (d+1)\nwwcHyp/\eps\epsLarge
\\ \hline \textrm{Division-1 tall non-dense items } (H^w_{*,*})
    & (d+1)\nwwcHyp/\eps\epsLarge
\\ \hline \textrm{Division-2 tall non-dense items } (H^h_{*,*})
    & 4\nwwcHyp/\epsLarge^2
\\ \hline \textrm{Small non-dense items } (S_{*})
    & \nswcHyp
\\ \hline \textrm{Division-1 heavy dense items } (D^{h,1})
    & \nhwcHyp
\\ \hline \textrm{Division-2 heavy dense items } (D^{h,2})
    & \nhwcHyp
\\ \hline \textrm{Division-1 light dense items } (D^{l,1})
    & \nlwcHyp
\\ \hline \textrm{Division-2 light dense items } (D^{l,2})
    & \nlwcHyp
\\ \hline
\end{array} }$
\label{table:item-types}
\end{table}

\begin{lemma}
\label{lem:ifp}
Let $\Pcal$ be a balanced fine partitioning of a slicing of $I$.
Then there is some output $(D, \Pcalhat)$ of $\iterFineParts(I)$ (\cref{algo:ifp}) such that
$\Pcalhat$ is a fine partitioning of $I-D$ and
after \cref{trn:round-geom}, each partition in $\Pcalhat$
is a predecessor of the corresponding partition in $\Pcal$.
\end{lemma}
\begin{proof}
Follows from \cref{lem:part-big,lem:part-wide}
\end{proof}

\begin{theorem}
\label{thm:round}
There is an output $(\widetilde{I}, D \cup \Imed)$
of $\round(I, \eps)$ such that $\widetilde{I}$ can be fractionally packed into at most
$\left(1 + \frac{2\eps}{1-\eps}\right)(a\opt(I) + b) + 2$
semi-structured $(5\eps/8)$-slacked bins.
Here values of $a$ and $b$ are as per \cref{table:slack-pack-ab}.
\end{theorem}
\begin{proof}
Let $m \defeq \left(1 + \frac{2\eps}{1-\eps}\right)(a\opt(I) + b) + 2$.
\Cref{lem:rnd2.5} guarantees the existence of a balanced fine partitioning $\Pcal$
of a slicing of $\widehat{I}$ such that after applying \cref{trn:round-geom},
there is a $(5\eps/8)$-slacked fractional packing of $\widehat{I}$
into $m$ bins that is semi-structured relative to $\Pcal$.

By \cref{lem:ifp}, we get that there is a fine partitioning $\Pcalhat$
of $\widehat{I} - D$ such that after applying \cref{trn:round-geom},
each partition of $\Pcalhat$ is a predecessor
of the corresponding partition of $\Pcal$.
Therefore, we can pack each item in $\widetilde{I}$ into the place of some
items in the packing of $\widehat{I} - D$.
Therefore, there is a $(5\eps/8)$-slacked fractional packing of $\widetilde{I}$
into at most $m$ bins that is semi-structured relative to $\Pcalhat$.
\end{proof}

\subsection{Existence of Compartmental Packing}
\label{sec:rbbp-2d-extra:compart}

\begin{definition}[Compartmental packing]
\label{defn:compartmental}
Consider a \semiStrucHyp{} packing of items $I$ into $m$ bins,
of which $m_1$ bins are division-1 bins and $m - m_1$ are division-2 bins.

A \emph{compartment} is defined to be a rectangular region in a bin such that every item
either lies completely inside the region or completely outside the region.
Furthermore, a compartment doesn't contain big items,
and a compartment doesn't contain both wide and tall items.

In a division-1 bin, a compartment is called a \emph{dense compartment} iff
it is the region $S^{(R')}$ and it contains a dense item
(recall that $S^{(R')} \defeq [1-\epsLarge/2, 1] \times [0, 1]$).
In a division-1 bin, a compartment is called a \emph{sparse compartment}
iff it satisfies all of these properties:
\begin{tightemize}
\item The compartment doesn't contain dense items.
\item The compartment contains at least 1 wide item or 1 tall item.
    If it contains wide items, it is called a wide compartment,
    and if it contains tall items, it is called a tall compartment.
\item The $x$-coordinate of the left edge of the compartment
    is a multiple of $\epsLarge^2/4$.
\item The compartment's width is a multiple of $\epsLarge^2/4$,
    and if the compartment is tall, its width is exactly $\epsLarge^2/4$.
\item \label{item:compartment:height}\empty
    The compartment's height is rounded, i.e., if the compartment is wide,
    its height is a multiple of a constant $\epsCont \defeq \eps\epsLarge^5/12$
    (note that $\epsCont^{-1} \in \mathbb{Z}$),
    and if the compartment is tall, its height is a sum of the heights of
    at most $1/\epsLarge-1$ of the items inside the compartment.
\end{tightemize}

A division-1 bin is said to be \emph{compartmental}
iff we can create non-overlapping dense and sparse compartments in the bin such that
all wide items, tall items and dense items are packed into compartments.

We can analogously define compartmental packing for division-2 bins
by swapping the coordinate axes.
A semi-structured bin packing of items is called compartmental
if each bin in the packing is compartmental.
\end{definition}

\begin{lemma}
\label{lem:empty-to-rects}
Let there be a rectangular bin $B \defeq [0, 1]^2$.
Let there be a set $I$ of rectangles packed inside the bin.
Then there is a polynomial-time algorithm which can decompose the empty space in the bin
($B - I$) into at most $3|I|+1$ rectangles by either making horizontal cuts only
or making vertical cuts only.
\end{lemma}
\begin{proof}
Extend the top and bottom edges of each rectangle leftwards and rightwards
till they hit another rectangle or the bin boundary.
This partitions the empty space into rectangles $R$.

For each rectangle $i \in I$, the top edge of $i$ is the bottom edge of a rectangle in $R$,
the bottom edge of $i$ is the bottom edge of two rectangles in $R$.
Apart from possibly the rectangle in $R$ whose bottom edge is at the bottom of the bin,
the bottom edge of every rectangle in $R$ is either the bottom or top edge of a rectangle in $I$.
Therefore, $|R| \le 3|I| + 1$.

If instead of the top and bottom edges, we extended left and right edges,
we would get the same result, except that the cuts would be vertical.
\end{proof}

\begin{lemma}
\label{lem:ss-to-round-cont}
If there exists a semi-structured $\mu$-slacked packing of items $I$ into $m$ bins,
then there exists a compartmental $\mu$-slacked fractional packing of $I$ into
$\left( 1 + 2\eps/(1-\mu) \right) m + 2$ bins.
\end{lemma}
\begin{proof}
Consider a division-1 bin. By \cref{prop:hrnd}\ref{item:hrnd:dense},
we get that all dense items (if any) lie in dense compartments.
Next, using the method of Section 3.2.3 in \cite{pradel-thesis}
(Rounding the Other Side $>$ Containers for the wide and long rectangles),
we can slice items without moving them and create non-overlapping compartments in the bin
such that all tall and wide items are packed into compartments.
Their method works by first constructing tall compartments
and then using the algorithm of \cref{lem:empty-to-rects} to partition
the space outside tall compartments and big items into wide compartments.
The resulting packing is compartmental, except that compartments' heights are not rounded.
We will now show how to round the heights of compartments.

Let there be $n_t$ tall compartments in the bin and $\nBig$ big items in the bin. Define
\begin{align*}
\nTCont &\defeq \frac{4}{\epsLarge^2}\left(\frac{1}{\epsLarge}-1\right)
    \le \frac{4}{\epsLarge^3}
&\nWCont &\defeq \frac{12}{\epsLarge^2}\left(\frac{1}{\epsLarge} - 1\right) + 1
    \le \frac{12}{\epsLarge^3}
\end{align*}
In the bin, there are $4/\epsLarge^2$ slots of width $\epsLarge^2/4$ and height 1.
Consider one such slot. Let there be $k$ big items that intersect that slot ($k$ can be 0).
The height of each big item and each sparse tall compartment is more than $\epsLarge$.
Therefore, the number of tall sparse compartments in that slot is at most $1/\epsLarge - 1 - k$.
Each big item spans at least $4/\epsLarge+1$ slots,
and reduces by 1 the number of tall compartments in the slots it spans.
Hence, the number of tall sparse compartments is at most
\[ n_t \le \frac{4}{\epsLarge^2}\left(\frac{1}{\epsLarge}-1\right)
    - \nBig\left(\frac{4}{\epsLarge} + 1\right)
\le \frac{4}{\epsLarge^2}\left(\frac{1}{\epsLarge}-1\right) - \nBig
= \nTCont - \nBig. \]
By \cref{lem:empty-to-rects}, the number of wide compartments is at most
$3(n_t + \nBig) + 1 \le \nWCont$.

Since small items can be sliced in both dimensions, we can treat them like a liquid.
For each tall sparse compartment $C$ in the bin,
let the tall items in $C$ sink down in this liquid.
Then shift down the top edge of $C$ to the top edge of the
topmost tall item in $C$ (so some small items will no longer be inside $C$).
Then the height of $C$ will be the sum of the heights of at most $1/\epsLarge-1$
tall items inside $C$ (since tall items have height $> \epsLarge$ and $C$ has height at most 1).

For each wide compartment $C$ in the bin, unpack a horizontal slice of height
$h(C) \bmod \epsCont$ from $C$ (this may require slicing items)
and move down the top edge of $C$ by $h(C) \bmod \epsCont$.
This rounds down $h(C)$ to a multiple of $\epsCont$.

Apply the above transformation to all division-1 bins
and an analogous transformation to all division-2 bins.
This gives us a $\mu$-slacked compartmental packing into $m$ bins.
However, we unpacked some items from wide containers in division-1 bins
and tall containers in division-2 bins. We need to repack these items.

Let there be $m_1$ division-1 bins. We removed a slice of height less than $\epsCont$
from each wide compartment in division-1 bins.
Let $S$ be the set of all such slices from division-1 bins.
There are at most $\nWCont$ wide compartments, so $h(S) \le \epsCont\nWCont m_1$.
For each slice $i \in S$, define
\[ \Span'(i) \defeq \max\left(h(i), \min\left(\frac{v_{\max}(i)}{1-\mu}, 1\right)\right). \]
For each slice $i \in S$,
\[ \Span'(i) \le \max\left(h(i), \frac{v_{\max}(i)}{1-\mu}\right)
\le \max\left(h(i), \frac{h(i)}{\epsLarge^2(1-\mu)}\right)
\le \frac{h(i)}{\epsLarge^2(1-\mu)}. \]
Therefore,
\[ \sum_{i \in S} \Span'(i) \le \frac{h(S)}{\epsLarge^2(1-\mu)}
\le \frac{\epsCont\nWCont}{\epsLarge^2(1-\mu)} m_1. \]
Interpret each slice $i \in S$ as a 1D item of size $\Span'(i)$ and pack the slices
one-over-the-other touching the left edge of bins using Next-Fit.
By \cref{lem:bot-slack-pack}, we can pack them all into
$1 + 2\sum_{i \in S}\Span'(i)$ bins that are $\mu$-slacked.
Since $\epsCont = \eps\epsLarge^5/12 \le \eps\epsLarge^2/\nWCont$,
the number of bins used to pack $S$ is at most
\[ \frac{2\epsCont\nWCont}{(1-\mu)\epsLarge^2} m_1 + 1
\le \frac{2\eps}{1-\mu} m_1 + 1. \]
These bins are division-1 compartmental;
they have just one wide compartment of width 1 and height 1.

Similarly, we can pack unpacked items from division-2 bins
into at most $\frac{2\eps}{1-\mu} (m-m_1) + 1$ division-2 compartmental $\mu$-slacked bins.
\end{proof}

\begin{theorem}
\label{thm:struct-theorem}
There is an output $(\widetilde{I}, D)$ of $\roundHyp(I, \eps)$ such that $\widetilde{I}$
has a \compartmentalHyp{} $(5\eps/8)$-slacked fractional packing into at most
$\left( 1 + 2\eps/(1-\eps) \right)^2 (a\opt(I) + b) + 4/(1-\eps)$
bins. Here $a$ and $b$ are as per \cref{table:slack-pack-ab}.
\end{theorem}
\begin{proof}
By \cref{thm:round}, there is an output $(\widetilde{I}, D)$ of $\round(I, \eps)$
such that $\widetilde{I}$ can be packed into
$m \defeq \left( 1 + 2\eps/(1-\eps) \right)(a\opt(I) + b) + 2$
semi-structured $(5\eps/8)$-slacked bins.

By \cref{lem:ss-to-round-cont}, the number of compartmental $(5\eps/8)$-slacked
bins needed to fractionally pack $\widetilde{I}$ is at most
$\left( 1 + 2\eps/(1-\eps) \right) m + 2
\le \left( 1 + 2\eps/(1-\eps) \right)^2 (a\opt(I) + b) + 4/(1-\eps)$.
\end{proof}

\begin{table}[!ht]
\centering
\caption{Upper bound on the number of distinct widths and heights for
compartments of different types}
${\displaystyle \begin{array}{*3{|>{\displaystyle}c}|}
\hline & \textrm{no. of widths} & \textrm{no. of heights}
\\ \hline \textrm{Division-1 wide compartments}
    & \frac{4}{\epsLarge^2}
    & \frac{12}{\eps\epsLarge^5}
\\ \hline \textrm{Division-1 tall compartments}
    & 1
    & \left(\frac{(d+1)\nwwcHyp}{\eps\epsLarge} + 1\right)^{1/\epsLarge-1}
\\ \hline \textrm{Division-2 wide compartments}
    & \left(\frac{(d+1)\nwwcHyp}{\eps\epsLarge} + 1\right)^{1/\epsLarge-1}
    & 1
\\ \hline \textrm{Division-2 tall compartments}
    & \frac{12}{\eps\epsLarge^5}
    & \frac{4}{\epsLarge^2}
\\ \hline
\end{array} }$
\end{table}

Since division-1 tall items have $(d+1)n_{\wwc}/\eps\epsLarge$ possible heights,
the number of possible heights of division-1 tall sparse compartments is
$\left((d+1)n_{\wwc}/\eps\epsLarge + 1\right)^{1/\epsLarge-1}$.
This is a huge number of possible heights, and it is possible to reduce it by partitioning
tall compartments into weight classes and using linear grouping.
We will not perform this improvement here.

\subsection{Packing Algorithm}

Let $I$ be a subset of the rounded items. Formally,
let $(\widetilde{I}', D) \in \round(I', \eps)$ and $I \subseteq \widetilde{I}'$.

We will first present a polynomial-time algorithm
$\fpack(I, m)$ that takes as input $I$ and an integer $m$ and either outputs
a fractional packing of $I$ into $m$ compartmental $\mu$-slacked bins
(where $\mu \le \eps$) or claims that fractionally packing $I$ into $m$
compartmental $\mu$-slacked bins is impossible.

We can use $\fpack(I, m)$ to find the optimal compartmental $\mu$-slacked
fractional packing of $I$ by using binary search on $m$.
With slight abuse of notation, let $\fpack(I)$ denote this algorithm.

Then, we will present an algorithm that finds a $\mu$-slacked
(non-fractional) packing of $I$ by using $\fpack(I)$ as a subroutine.
Note that we're interested in getting a non-fractional $\mu$-slacked packing of $I$,
but that packing need not be compartmental.

\subsubsection{Guessing Bin Configurations}

A $\mu$-slacked bin $J$ can have one of 4 possible slack types:
\begin{enumerate}
\item Normal: $\forall k \in [d], v_k(J) \le 1-\mu$.
\item Big single: $|J| = 1$ and $J$ contains a big item
    and $\forall k \in [d], v_k(J) \in (1-\mu, 1]$.
\item Dense single: $|J| = 1$ and $J$ contains a dense item
    and $\forall k \in [d], v_k(J) \in (1-\mu, 1]$.
\item Dense double: $|J| = 2$ and $J$ contains two dense items
    and $\forall k \in [d], v_k(J) \in (1-\mu, 1]$
    and $\forall i \in J, v_{\max}(i) \le 1/2$.
\end{enumerate}
Note that these slack types are disjoint, i.e., a bin cannot have more than one slack types.

A configuration of a bin is defined to be all of this information:
(i) The division type,
(ii) The slack type,
(iii) Whether the bin has a dense compartment,
(iv) A packing of big items, heavy items and compartments into the bin.

We will now enumerate all possible configurations that a bin can have.

For a division-1 bin,
there are $4(d+1)n_{\bwc}/\eps\epsLarge^3$ different types of big items,
$n_{\hwc}$ different types of heavy items,
$48/\eps\epsLarge^7$ different types of wide compartments and
$((d+1)n_{\wwc}/\eps\epsLarge + 1)^{1/\epsLarge-1}$ different types of tall compartments.
A bin can pack less than $1/\epsLarge^2$ big items, less than $1/\epsLarge$ heavy items
at most $\nWCont$ wide compartments and at most $\nTCont$ tall compartments.
Therefore, by iterating over
\[ \nNConfs' \defeq
\left(\frac{4(d+1)n_{\bwc}}{\eps\epsLarge^3} + 1\right)^{1/\epsLarge^2-1}
\left(\frac{48}{\eps\epsLarge^7} + 1\right)^{\nWCont}
\left(\left(\frac{(d+1)n_{\wwc}}{\eps\epsLarge}+1\right)^{1/\epsLarge-1}
    \hspace{-1em}+ 1\right)^{\nTCont} \]
values (a large constant), we can guess the set of big items, heavy items and compartments
in a division-1 bin of normal slack type that does not have a dense compartment.
For a bin that has a dense compartment, the number configurations to consider is
$\nNConfs'(n_{\hwc} + 1)^{1/\epsLarge-1}$.

For a bin of big-single slack type, there are at most $4(d+1)n_{\bwc}/\eps\epsLarge^3$
configurations. For a bin of dense-single slack type, there are at most
$n_{\hwc}$ configurations. For a bin of dense-double slack type, there are at most
$n_{\hwc}^2$ configurations.
Double the number of configurations to also account for division-2 items.
Therefore, the total number of configurations is at most
\[ \nNConfs \defeq 2\left(\nNConfs'\left(1 + (n_{\hwc} + 1)^{1/\epsLarge-1}\right)
    + \frac{4(d+1)n_{\bwc}}{\eps\epsLarge^3} + n_{\hwc} + n_{\hwc}^2\right). \]

There can be at most $\nWCont + \nTCont$ items and compartments in a bin
(see the proof of \cref{lem:ss-to-round-cont}).
Since the $x$-coordinate of these items and compartments is a multiple of $\epsLarge^2/4$,
we can use brute-force to check if they can be packed into a bin.
For each item, guess its $x$-coordinate and its `layer' number.
This will take time
$O\left(\left(4(\nWCont + \nTCont)/\epsLarge^2\right)^{\nWCont + \nTCont}\right)$.

For $m$ bins, we can have at most $\binom{m+\nNConfs-1}{\nNConfs-1} \in O(m^{\nNConfs-1})$
possible combinations of configurations. Now for each combination,
we will check if the remaining items can fit into the bins.

\subsubsection{Fractionally Packing Wide, Tall, Small and Light items}

We will use a linear program to fractionally pack wide, tall, small and light items.
Note that bins of non-normal slack type can only have big and heavy items,
which have already been packed, so we won't consider them any further.

\begin{definition}[Length and Breadth]
For an item $i$, let the length $\ell(i)$ be the longer geometric dimension
and the breadth $b(i)$ be the shorter geometric dimension
(so for a wide item $i$, $\ell(i) \defeq w(i)$ and $b(i) \defeq h(i)$,
and for a tall item $i$, $\ell(i) \defeq h(i)$ and $b(i) \defeq w(i)$).
Similarly define $\ell$ and $b$ for wide and tall compartments.
\end{definition}

In any packing of items in a wide compartment, move a line horizontally upwards,
as if scanning the compartment. The line, at each point, will intersect some wide items.
The set of such wide items is called a 1D configuration.
Similarly define 1D configuration for items in tall compartments.
Any fractional packing of items inside a compartment can be described
by mentioning the breadths of all 1D configurations in the compartment.
The number of types of items is a constant and there can be at most $1/\epsLarge-1$
items in a 1D configuration. Therefore, the number of 1D configurations is a constant.

\begin{tightemize}
\item Let $M_1$ and $M_2$ be the set of division-1 bins and division-2 bins respectively
    of normal slack type. Let $M \defeq M_1 \cup M_2$.
    Let $M_D \subseteq M$ be the set of bins that have a dense compartment.
\item Let the set of wide and tall non-dense item types be $L$. For $t \in L$,
    let $b(t)$ be the sum of breadths of items of type $t$
    and each item has length $\ell(t)$.
\item Let $J_i$ be the set of compartments in bin $i$.
\item For compartment $j$, let $\mathcal{C}_j$ be the set of feasible 1D configurations
    and let $b(j)$ be the breadth of the compartment.
\item Let $\ell_C$ be the sum of lengths of items in 1D configuration $C$.
\item Let $n_{t,C}$ be the number of items of type $t$ in 1D configuration $C$.
\item Let $\alpha_{k,C}$ be the weight-to-breadth ratio of 1D configuration $C$
    in the $k\Th$ vector dimension.
\item Let $\beta_{k,p}$ be the weight-to-area ratio of weight class $p$ for
    small non-dense items.
\item Let $\gamma_{k,p} \defeq v_k(i)/v_{\max}(i)$ for a light dense item $i$ in weight class $p$.
\item Let $B_i$ be the set of big items in bin $i$.
\item Let $H_i$ be the set of heavy items in bin $i$.
\item Let $D_i$ be 1 for bin $i$ if it contains a dense compartment and 0 otherwise.
\end{tightemize}

We will fractionally pack items using a linear program that only has constraints
and has no objective function.
We call it the fractional packing feasibility program $\FP$.
It has variables $x$, $y$ and $z$, where
\begin{tightemize}
\item $x_{j,C}$ is the breadth of 1D configuration $C$ in compartment $j$.
\item $y_{i,p}$ is the area of small non-dense items of weight class $p$ in bin $i$.
\item $z_{i,p}$ is the $v_{\max}$ of light dense items of weight class $p$ in bin $i$.
When $i \not\in M_D$, $z_{i,p}$ is the constant 0 instead of being a variable.
\end{tightemize}
Feasibility program $\FP$:
\begin{longtable}{LL}
\label{fpack-fp}
\sum_{C \in \mathcal{C}_j} x_{j,C} = b(j)
& \forall i \in M, \forall j \in J_i
\\ & \textrm{(compartment breadth constraint)}
\\
\\ \multicolumn{2}{L}{
        \sum_{j \in J_i} \sum_{C \in \mathcal{C}_j} \ell_C x_{j,C}
        + \sum_{p=1}^{n_{\swc}} y_{i,p} \le 1 - \frac{\epsLarge}{2}D_i - a(B_i)}
\\ & \forall i \in M
\\ & \textrm{(bin area constraint)}
\\
\\ \multicolumn{2}{L}{
        \sum_{j \in J_i} \sum_{C \in \mathcal{C}_j} \alpha_{k,C} x_{j,C}
        + \sum_{p=1}^{n_{\wwc}} \beta_{k,p} y_{i,p}
        + \sum_{p=1}^{n_{\lwc}} \gamma_{k,p} z_{i,p} \le 1 - \mu - v_k(B_i) - v_k(H_i)}
\\ & \forall i \in M, \forall k \in [d]
\\ & \textrm{(bin weight constraint)}
\\
\\ \sum_{i \in M} \sum_{j \in J_i} \sum_{\substack{C \in \mathcal{C}_j \\ C \ni t}}
        n_{t,C} x_{j,C} = b(t)
& \forall t \in L
\\ & \textrm{(conservation of wide and tall items)}
\\
\\ \sum_{i \in M} y_{i,p} = a(S_p)
& \forall p \in [n_{\swc}]
\\ & \textrm{(conservation of small items)}
\\
\\ \sum_{i \in M_1} z_{i,p} = v_{\max}(D_p^{l,w})
& \forall p \in [n_{\lwc}]
\\ & \textrm{(conservation of division-1 light items)}
\\
\\ \sum_{i \in M_2} z_{i,p} = v_{\max}(D_p^{l,h})
& \forall p \in [n_{\lwc}]
\\ & \textrm{(conservation of division-2 light items)}
\\
\\ x_{j,C} \ge 0 & \forall i \in M, \forall j \in J_i, \forall C \in \mathcal{C}_j
\\ y_{i,p} \ge 0 & \forall i \in M, \forall p \in [n_{\swc}]
\\ z_{i,p} \ge 0 & \forall i \in M_D, \forall p \in [n_{\lwc}]
\\ & \textrm{(non-negativity constraints)}
\end{longtable}

The number of constraints in $\FP$ (other than the non-negativity constraints) is at most
\begin{align*}
n_c &\defeq m(\nWCont + \nTCont + d + 1)
+ 2n_{\wwc}\left(\frac{4}{\epsLarge^2} + \frac{d+1}{\eps\epsLarge}\right)
+ n_{\swc} + 2n_{\lwc}
\\ &\le \left(\frac{16}{\epsLarge^3} + d\right) m
+ 2n_{\wwc}\left(\frac{4}{\epsLarge^2} + \frac{d+1}{\eps\epsLarge} + 1\right)
\tag{$2n_{\lwc} \le n_{\wwc}$ and $n_{\swc} \le n_{\wwc}$}
\\ &\le \left(\frac{16}{\epsLarge^3} + d\right) m + 2(d+6)\frac{n_{\wwc}}{\epsLarge^2}.
\end{align*}

The number of variables and constraints in $\FP$ are linear in $m$.
Therefore, $\FP$ can be solved in time polynomial in $m$.
Furthermore, if $\FP$ is feasible, we can obtain an extreme-point solution to $\FP$.

Therefore, $\fpack(I, m)$ guesses all combinations of configurations of $m$ bins
and for each such combination of configurations solves the feasibility program to check
if the remaining items can be packed into the bins according to the bin configurations.
Furthermore, $\fpack(I, m)$ runs in time polynomial in $m$.
$\fpack(I)$ makes at most $O(\log n)$ calls to $\fpack(I, m)$ with $m \le n$.
Therefore, $\fpack(I)$ runs time polynomial in $n$.

\subsubsection{Getting Containers from a Fractional Packing Solution}
\label{sec:algo:make-cont}

Suppose $\fpack(I)$ outputs a fractional packing that uses $m$ bins.
In each compartment $j$, we create a slot of breadth $x_{j,C}$
for each 1D configuration $C$.
Since we're given an extreme-point solution to $\FP$, the number of slots is
at most the number of constraints $n_c$ by rank lemma.
In each slot, we create $n_{t,C}$ containers of type $t \in L$
having length $\ell(t)$ and breadth $x_{j,C}$.

Now we (non-fractionally) pack a large subset of wide and tall items into containers
and we pack a large subset of small non-dense and light dense items outside containers.
In each wide container, items will be stacked one-over-the-other.
In each tall container, items will be stacked side-by-side.
We pack the remaining unpacked items into a small number of new bins.
This will give us a non-fractional packing that uses close to $m$ bins.

\subsubsection{Packing Light Dense Items}
\label{sec:algo:pack-light}

For light dense items, for each bin $i$ and each weight class $p \in [n_{\lwc}]$,
keep adding items of the same division as the bin and from weight class $p$ to the bin
till the total $v_{\max}$ of the items exceeds $z_{i,p}$.
Then discard the last item that was added.

As per the conservation constraints for light dense items, all items will
either be packed or discarded. The $v_{\max}$ of items from weight class $p$
that are packed into bin $i$ is at most $z_{i,p}$.

The number of discarded items is at most the number of $z$-variables in $\FP$,
which is at most $m n_{\lwc}$. Let $D$ be the set of discarded items. Then
$\Span(D) = v_{\max}(D) \le (\epsSmall n_{\lwc}) m$.
We choose $\epsSmall \le \eps/n_{\lwc}$.
Since $\epsSmall \le \epsLarge^2(1-\eps)$, each item's weight is at most $1-\eps$.
Therefore, we can use Next-Fit to pack $D$ into
$2\Span(D)/(1-\mu) + 2 \le 2\eps m/(1-\mu) + 2$ number of
$\mu$-slacked bins (by scaling up each item's weight by $1/(1-\mu)$ before packing
and scaling it back down after packing),
where tall and small dense items are packed separately from wide dense items.

The time taken to pack these items is $O(|D^{l,*}_{*}|)$.

\subsubsection{Packing Wide and Tall Non-Dense Items}
\label{sec:algo:pack-wide}

For each item type $t$, iteratively pack items of type $t$ into a container of type $t$
till the total breadth of items in the container exceeds $x_{j,C}$. Then discard
the last item and move to a new container and repeat. As per the conservation constraints
for wide and tall items, all items will either be packed or discarded.

Treat the items discarded from each slot $C$ as a single (composite) item of breadth $\epsSmall$,
length $\ell_C$ and weight $\epsSmall/\epsLarge^2$ in each dimension.
We will pack these composite items into bins, where wide and tall items are packed separately.

Let $D$ be the set of all such discarded items. Then $|D| \le n_c$.
For a composite item $i$, let
\[ \Span'(i) \defeq \max\left(b(i), \min\left(\frac{v_{\max}(i)}{1-\mu}, 1\right)\right). \]
Treat each composite item as a 1D item of size $\Span'(i)$.
Then by \cref{lem:bot-slack-pack}, we can use Next-Fit to pack these items
into $2\Span'(D) + 2$ $\mu$-slacked bins, where wide and tall items are packed separately.
\[ \Span'(i) \le \max\left(b(i), \frac{v_{\max}(i)}{1-\mu}\right)
\le \max\left(\epsSmall, \frac{\epsSmall}{\epsLarge^2(1-\mu)}\right)
\le \frac{\epsSmall}{\epsLarge^2(1-\mu)}. \]
Therefore, the number of bins needed is
\[ 2\Span'(D) + 2 \le 2\frac{\epsSmall}{\epsLarge^2(1-\mu)} n_c + 2
\le \frac{2\epsSmall(16 + d\epsLarge^3)}{\epsLarge^5(1-\mu)}m
    + \frac{4(d+6)}{1-\mu} \frac{\epsSmall n_{\wwc}}{\epsLarge^4} + 2. \]
Therefore, we choose
$\epsSmall \le \eps\epsLarge^5/(16 + d\epsLarge^3)$
so that the number of new bins needed is at most
\[ \frac{2\eps}{1-\mu}m + 2 + \frac{4(d+6)}{1-\mu} \frac{\epsSmall n_{\wwc}}{\epsLarge^4}. \]
The time taken to pack these items is $O(|L|)$.

\subsubsection{Packing Small Non-Dense Items}
\label{sec:algo:pack-small}

In each bin, there are at most $\nWCont + \nTCont$ compartments and big items
(see the proof of \cref{lem:ss-to-round-cont}).
By \cref{lem:empty-to-rects}, the free space outside compartments can be partitioned into
\[ 3(\nWCont + \nTCont) + 1
\le 3\left(\frac{16}{\epsLarge^2}\left(\frac{1}{\epsLarge} - 1\right) + 1\right) + 1
\le \frac{48}{\epsLarge^3} \]
rectangular regions.
Suppose there are $p_i$ slots in bin $i$ for which $x_{j,C} > 0$.
The sum of $p_i$ over all bins is equal to $n_c$, the number of constraints in $\FP$.
Each slot may give us a free rectangular region in the compartment.
Therefore, the free space in bin $i$ can be partitioned into $p_i + 48/\epsLarge^3$
rectangular regions. Let $R_i$ denote the set of these free rectangular regions.

Let $M$ be the set of bins of normal slack type. Let $m \defeq |M|$. Then
\[ \sum_{i \in M} |R_i| \le n_c + \frac{48}{\epsLarge^3} m
\le \left(\frac{64}{\epsLarge^3} + d\right) m. \]
This observation forms the basis of our algorithm $\packSmall$ (\cref{algo:pack-small})
for packing small non-dense items.

\begin{algorithm}[!ht]
\caption{$\packSmall(I, M, y)$:
Here $I$ is a set of items and $M$ is the fractional packing output by $\fpack(I)$.
$(x, y, z)$ is a feasible solution to $\FP$ output by $\fpack(I)$.}
\label{algo:pack-small}
\begin{algorithmic}[1]
\State Let $S_p$ be the $p\Th$ coarse partition of small non-dense items in $I$.
\State $D_1 = D_2 = \{\}$  \Comment{sets of items to discard}
\For{each bin $i \in M$ of normal slack type}
    \State Let $R_i$ be the set of free rectangular regions in bin $i$.
    \LineComment{Select a set of items to pack}
    \State $\widehat{S} = \{\}$
    \For{$p \in [n_{\swc}]$}
        \State $S_{i,p} = \{\}$
        \While{$a(S_{i,p}) < y_{i,p}$ and $|S_p| > 0$}
            \State Move an item from $S_p$ to $S_{i,p}$.
        \EndWhile
        \If{$a(S_{i,p}) > y_{i,p}$} \label{alg-line:pack-small:select}
            \State Move the last item in $S_{i,p}$ to $D_1$.
        \EndIf
        \State $\widehat{S} = \widehat{S} \cup S_{i,p}$
    \EndFor
    \LineComment{Pack those items into $R_i$}
    \While{$|R_i| > 0$ and $|\widehat{S}| > 0$}
        \State Remove a rectangle $r$ from $R_i$.
        \State Let $T$ be the smallest prefix of $\widehat{S}$ such that
            $a(T) \ge a(r) - 2\epsSmall$ or $T = \widehat{S}$
        \State \label{alg-line:pack-small:nfdh}Pack $T$ into $r$ using NFDH.
            \Comment{We will prove that this is possible}
        \State $\widehat{S} \texttt{ -= } T$
    \EndWhile
    \State $D_2 = D_2 \cup \widehat{S}$
\EndFor
\State \Return $D_1 \cup D_2$  \Comment{discarded items}
\end{algorithmic}
\end{algorithm}

\begin{lemma}
Algorithm $\packSmall(I, M)$ doesn't fail at line \ref{alg-line:pack-small:nfdh}.
\end{lemma}
\begin{proof}
Since all items have area at most $\epsSmall^2$,
\[ a(T) < a(r) - 2\epsSmall + \epsSmall^2
\le a(r) - (w(r) + h(r))\epsSmall + \epsSmall^2
= (w(r) - \epsSmall)(h(r) - \epsSmall). \]
Therefore, by \cref{lem:nfdh-small}, we get that $T$ can be packed into $r$.
\end{proof}

\begin{lemma}
In $\packSmall(I, M)$, every small non-dense item is either packed or discarded.
\end{lemma}
\begin{proof}[Proof by contradiction]
Assume $\exists p \in [n_{\swc}]$ such that there are items in $S_p$
that are neither packed nor discarded.
Therefore, for each bin $i$, at \cref{alg-line:pack-small:select},
$a(S_{i,p}) \ge y_{i,p}$. Therefore, the total area of all items from $S_p$
that are either packed or discarded is at least
\[ \sum_{i \in M} y_{i,p} = a(S_p)  \tag{by conservation constraint in $\FP$} \]
which means that all items have been packed or discarded, which is a contradiction.
\end{proof}

\begin{lemma}
Let $D_1 \cup D_2 = \packSmall(I, M)$. Then
$a(D_1 \cup D_2) \le \epsSmall\left(n_{\swc} + \frac{128}{\epsLarge^3} + 2d\right) m$.
\end{lemma}
\begin{proof}
During $\packSmall(I, M)$, for each bin, $a(D_1)$ increases by at most $\epsSmall n_{\swc}$.
Therefore, $a(D_1) \le \epsSmall n_{\swc} m$.

We know that $a(\widehat{S}) \le \sum_{p=1}^{n_{\swc}} y_{i,p} \le a(R_i)$.
The first inequality follows from the way we chose $\widehat{S}$.
The second inequality follows from the area constraint in $\FP$.

\textbf{Case 1}: We used up all items in $\widehat{S}$ during bin $i$:\\
We didn't discard any items during bin $i$.

\textbf{Case 2}: We used up all rectangles in $R_i$ during bin $i$:\\
Then the used area is at least $a(R_i) - 2\epsSmall|R_i|$.
Therefore, the items discarded during bin $i$ have area at most
$a(\widehat{S}) - a(R_i) + 2\epsSmall|R_i| \le 2\epsSmall|R_i|$.
Therefore,
\[ a(D_2) \le 2\epsSmall\sum_{i \in M}|R_i|
\le 2\epsSmall\left(\frac{64}{\epsLarge^3} + d\right) m. \qedhere \]
\end{proof}

\begin{lemma}
Let $D \defeq \packSmall(I, M)$. Then we can pack $D$ into
$2\eps m/(1-\mu) + 1$ number of
$\mu$-slacked bins, where $\mu \le \eps$.
\end{lemma}
\begin{proof}
Since $\epsSmall^2 \le (1-\eps)\epsLarge^2$, we get that
$\forall i \in D, v_{\max}(i) \le 1-\eps \le 1-\mu$.
Let $\Span'(i) = a(i)/\epsLarge^2(1-\mu)$. Then by interpreting each $i \in D$ as a 1D item
of size $\Span'(i)$, we can pack them into $2\Span'(D) + 1$ bins using Next-Fit.
In each bin $J$,
$v_{\max}(J) \le a(J)/\epsLarge^2 = (1-\mu)\Span'(J) \le 1-\mu$.
Also, $a(J) \le \epsLarge^2(1-\mu)\Span'(J) \le (1-\epsSmall)^2$.
Therefore, the bin is $\mu$-slacked and by \cref{lem:nfdh-small},
we can pack items $J$ in the bin using NFDH.

We choose
$\epsSmall \le \eps\epsLarge^2/\left(n_{\swc} + 128/\epsLarge^3 + 2d\right)$.
Therefore, the number of bins needed is at most
\[ 2\Span'(D) + 1 \le \frac{2a(D)}{\epsLarge^2(1-\mu)} + 1
\le \frac{2\eps}{1-\mu} m + 1. \qedhere \]
\end{proof}

The time taken to pack small items is $O(n_S \log n_S)$, where $n_S$ is the
number of small items, because we need to sort items by height for NFDH.

\subsubsection{The Algorithm and its Approximation Factor}

We give an algorithm $\ipack_{\mu}$ (\cref{algo:ipack}) for packing a subset $I$ of rounded items.

\begin{algorithm}[!ht]
\caption{$\ipack_{\mu}(I)$: Computes a non-fractional $\mu$-slacked packing of $I$.
Here $\mu \le \eps$ and $I \subseteq \widetilde{I}'$ and $(D, \widetilde{I}') \in \round(I')$
for some set $I'$ of items.}
\label{algo:ipack}
\begin{algorithmic}[1]
\State Let $(J, x, y, z) \defeq \fpack_{\mu}(I)$. Here $J$ is a fractional packing of $I$ into $m$ bins
and $(x, y, z)$ is a feasible solution to the feasibility program $\FP$.
\State Create containers inside compartments in $J$ using $x$ as per \cref{sec:algo:make-cont}.
\State Pack light dense items into dense compartments using $z$ as per \cref{sec:algo:pack-light}.
\State Pack wide and tall non-dense items into containers as per \cref{sec:algo:pack-wide}.
\State $\packSmallHyp(I, J, y)$ \Comment{Pack small non-dense items outside containers.}
\end{algorithmic}
\end{algorithm}

\begin{theorem}
\label{thm:ipack}
Let $(\Itild', D) \in \round(I', \eps)$ and $I \subseteq \Itild'$.
Let there be $m$ bins in the optimal $\mu$-slacked compartmental fractional packing of $I$,
where $\mu \le \eps$. Then $\ipack_{\mu}(I)$ runs in polynomial time and outputs
a $\mu$-slacked packing of $I$ into
\[ \left( 1 + \frac{6\eps}{1-\mu}\right)m + 5
    + \frac{4(d+6)}{1-\mu} \frac{\epsSmall\nwwcHyp}{\epsLarge^4} \]
bins, where each bin satisfies either \cref{prop:hrnd}\ref{item:hrnd:dense}
or \cref{prop:vrnd}\ref{item:vrnd:dense}.
\end{theorem}
\begin{proof}
$\fpack(I)$ finds the optimal $\mu$-slacked compartmental fractional packing
of $I$ in polynomial time. Given the output of $\fpack(I)$, $\ipack$ can, in $O(n\log n)$ time,
compute a packing of $I$ into
$\left( 1 + 6\eps/(1-\mu)\right)m + 5
+ 4(d+6)\epsSmall n_{\wwc}/\epsLarge^4(1-\mu)$ number of $\mu$-slacked bins.

Each bin satisfies either \cref{prop:hrnd}\ref{item:hrnd:dense}
or \cref{prop:vrnd}\ref{item:vrnd:dense}.
This is because the $m$ bins output by the fractional packing are compartmental,
and the extra bins either contain only dense tall and small items or
only dense wide items or only non-dense items.
\end{proof}

We choose
$\epsSmall \defeq \ceil{\max\left(n_{\lwc}/\eps, (16 + d^3)/\eps\epsLarge^5,
    (128 + \epsLarge^3(n_{\swc} + 2d))/\eps\epsLarge^5\right)}^{-1}$.
Hence,
$\epsSmall^{-1} \in O(\epsLarge^{-5} + \eps^{-d}\epsLarge^{-(2d+2)})$.
So, the parameter $f$ in $\remMedHyp(I, \eps, f, \delta_0)$ is
\begin{equation} \label{eqn:remmed-f} f(x) = \ceil{\max\left(
    \frac{\nlwcHyp}{\eps}, \frac{16 + d^3}{\eps x^5},
    \frac{128 + x^3(\left(8/x^2\eps\right)^d + 2d)}{\eps x^5}\right)}^{-1}.
    \tag*{}
\end{equation}

\begin{theorem}
\label{thm:ipack-simple-approx}
Let $I$ be a \geomvec{2}{$d$} bin packing instance.
For some $(\Itild, D) \in \round(I, \eps)$, if we pack $D$ using $\simplePack$
and pack $\Itild$ using $\ipack$, we get a packing of $I$ into at most
\[ \left((1 + 17\eps)a + 6\eps(d+1)\right)\opt(I) + O(1)
+ \frac{4(10d+51)}{1-\eps} \frac{\epsSmall\nwwcHyp}{\epsLarge^4} \]
bins. Here $a$ is the constant defined in \cref{table:slack-pack-ab}
in \cref{sec:rbbp-2d-extra:slack}.
\end{theorem}
\begin{proof}
By \cref{thm:struct-theorem}, $\exists (\Itild, D) \in \round(I, \eps)$ such that
$\Itild$ has a compartmental $(5\eps/8)$-slacked fractional packing into at most
$m \defeq \left( 1 + 2\eps/(1-\eps) \right)^2 a\opt(I) + O(1)$ bins.

By \cref{thm:span-pack,lem:round-discard,lem:span-lb-opt}, we get
\begin{align*}
|\simplePack(D)| &\le 6\Span(D) + 3
\le 6\left(\eps\Span(I) + 6(d+5)\frac{\epsSmall\nwwcHyp}{\epsLarge^4}\right) + 3
\\ &\le 6\eps(d+1)\opt(I) + \left(3 + 36(d+5)\frac{\epsSmall\nwwcHyp}{\epsLarge^4}\right).
\end{align*}

Let $J \defeq \ipack_{\mu}(\Itild)$, where $\mu \defeq 5\eps/8$.
Then the bins in $J$ are $(5\eps/8)$-slacked. The number of bins in $J$ is at most
\begin{align*}
& \left( 1 + \frac{6\eps}{1-\mu}\right)m + 5
    + \frac{4(d+6)}{1-\mu} \frac{\epsSmall n_{\wwc}}{\epsLarge^4}
\\ &\le \left(1 + \frac{2\eps}{1-\eps}\right)^2\left(1 + \frac{6\eps}{1-\eps}\right)a\opt(I)
    + O(1) + \frac{4(d+6)}{1-\eps} \frac{\epsSmall n_{\wwc}}{\epsLarge^4}
\\ &\le (1 + 17\eps)a\opt(I)
    + O(1) + \frac{4(d+6)}{1-\eps} \frac{\epsSmall n_{\wwc}}{\epsLarge^4}.
\tag{$\eps \le 1/8$}
\end{align*}

To get a packing of $I-D$ from a packing of $\Itild$, we need to revert
the transformations done by $\round$. The only transformation in $\round$ where we
rounded down instead of rounding up is \cref{trn:round-to-0}.
As per \cref{lem:round-to-0}, it is possible to undo \cref{trn:round-to-0}
because the bins are $(5\eps/8)$-slacked and each bin satisfies either
\cref{prop:hrnd}\ref{item:hrnd:dense} or \cref{prop:vrnd}\ref{item:vrnd:dense}.
\end{proof}

By \cref{thm:ipack-simple-approx}, we get that for every $\eps' > 0$, there is a
$(a + \eps')$-asymptotic-approx algorithm for \geomvec{2}{$d$} bin packing.
We can get a better approximation factor by using the Round-and-Approx Framework,
and we call the resulting algorithm $\cbPack$.

\subsection{Using the Round-and-Approx Framework}
\label{sec:rbbp-2d-extra:rna}

To implement $\cbPack$, we need to show how to implement
$\round$, $\complexPack$ and $\unround$, and we must prove the structural theorem.

\begin{enumerate}
\item \textbf{$\cLPsolve(I)$}:
Using the algorithm of \cite{gvks} for \geomvec{2}{$d$} KS
and the LP algorithm of \cite{eku-pst},
we get a $2(1+\eps)$-approximate solution to $\cLP(I)$. Therefore, $\mu = 2(1+\eps)$.

\item \textbf{$\round$}: We can use \cref{algo:round} as $\round$.
By \cref{lem:round-time}, $\round$ runs in polynomial time.
By \cref{lem:round-discard}, $\round$ has low discard.
By \cref{lem:round-homo}, $\round$ partitions items into a constant number of classes.

\item \textbf{Structural theorem}: \Cref{thm:struct-theorem}.
We call a packing structured iff it is \compartmentalHyp{} and $(5\eps/8)$-slacked.
Here $\rho = a(1+2\eps/(1-\eps))^2$.

\item \textbf{$\complexPack(\widetilde{S})$}:
Use $\ipackHyp_{\mu}$ as $\complexPack$, where $\mu = 5\eps/8$.
By \cref{thm:ipack}, $\alpha = 1 + 6\eps/(1-\eps)$.

\item \textbf{$\unround(\Jtild)$}:
Since the output of $\ipack$ is $(5\eps/8)$-slacked and each bin satisfies
either \cref{prop:hrnd}\ref{item:hrnd:dense} or \cref{prop:vrnd}\ref{item:vrnd:dense},
we can use \cref{lem:round-to-0} to undo \cref{trn:round-to-0}.
The other transformations round up, so they are trivial to undo.
Therefore, $\gamma = 1$.
\end{enumerate}

The only remaining requirement of the Round-and-Approx framework is proving
the bounded expansion lemma.

\begin{lemma}[Bounded expansion]
Let $(\Itild, D) \in \round(I, \eps)$.
Let $\Ktild \subseteq \Itild$ be a fine partition of $\Itild$.
Let $C \subseteq I$ be a set of items that fit into a bin
and let $\Ctild$ be the corresponding rounded items.
Then $\Span(\Ktild \cap \Ctild)$ is upper-bounded by
$1/\epsLarge + 1/4$.
\end{lemma}
\begin{proof}
To prove this, it is sufficient to prove that $\forall i \in I-D$, if item $i$
gets rounded to item $\itild$, then $\Span(\itild)/\Span(i)$ is upper-bounded by
$1/\epsLarge + 1/4$.

\begin{enumerate}
\item \textbf{Big items}: We will consider division-1 big items.
The analysis for division-2 big items is analogous.
\begin{align*}
& w(\itild) \le w(i) + \epsLarge^2/4 \wedge h(\itild) \le 1
\tag{by \cref{trn:round-geom}}
\\ &\implies \frac{w(\itild)}{w(i)} \le 1 + \frac{\epsLarge}{4}
\wedge \frac{h(\itild)}{h(i)} \le \frac{1}{\epsLarge}
\\ &\implies \frac{a(\itild)}{\Span(i)} \le \frac{a(\itild)}{a(i)}
    \le \frac{w(\itild)}{w(i)}\frac{h(\itild)}{h(i)}
    \le \frac{1}{\epsLarge} + \frac{1}{4}
\end{align*}
\begin{align*}
\frac{v_{\max}(\itild)}{\Span(i)}
&\le \frac{v_{\max}(i) + \epsLarge^2\eps/8}{\max(a(i), v_{\max}(i))}
\tag{by \cref{trn:wround-non-dense}}
\\ &\le \min\left(\frac{v_{\max}(i)}{\epsLarge^2} + \frac{\eps}{8},
    1 + \frac{\epsLarge^2\eps}{8v_{\max}(i)}\right)
\\ &\le 1 + \min\left(\frac{v_{\max}(i)}{\epsLarge^2},
    \frac{\eps}{8}\frac{\epsLarge^2}{v_{\max}(i)}\right)
\\ &\le 1 + \sqrt{\frac{\eps}{8}} \le \frac{3}{2}
\le \frac{1}{\epsLarge} + \frac{1}{4}
\tag{$\epsLarge \le 2/3$}
\end{align*}

\item \textbf{Non-dense wide items}:
\[ \frac{a(\itild)}{\Span(i)} \le \frac{a(\itild)}{a(i)}
= \frac{w(\itild)}{w(i)} \le \frac{1}{\epsLarge} \]
\begin{align*}
\frac{v_{\max}(\itild)}{\Span(i)}
&\le \frac{v_{\max}(i) + h(i)(\epsLarge\eps/8)}{\max(h(i)\epsLarge, v_{\max}(i))}
\tag{by \cref{trn:wround-non-dense}}
\\ &\le \min\left(\frac{v_{\max}(i)}{\epsLarge h(i)} + \frac{\eps}{8},
    1 + \frac{\eps}{8} \frac{\epsLarge h(i)}{v_{\max}(i)} \right)
\\ &\le 1 + \min\left(\frac{v_{\max}(i)}{\epsLarge h(i)},
    \frac{\eps}{8} \frac{\epsLarge h(i)}{v_{\max}(i)} \right)
\\ &\le 1 + \sqrt{\frac{\eps}{8}} \le \frac{3}{2} \le \frac{1}{\epsLarge}
\tag{$\epsLarge \le 2/3$}
\end{align*}

\item \textbf{Non-dense tall items}: Similar to the non-dense wide items case.

\item \textbf{Non-dense small items}: $a(\itild) = a(i)$.
\begin{align*}
\frac{v_{\max}(\itild)}{\Span(i)}
&\le \frac{v_{\max}(i) + a(i)(\eps/8)}{\max(a(i), v_{\max}(i))}
\tag{by \cref{trn:wround-non-dense}}
\\ &\le \min\left(\frac{v_{\max}(i)}{a(i)} + \frac{\eps}{8},
    1 + \frac{\eps}{8} \frac{a(i)}{v_{\max}(i)} \right)
\\ &\le 1 + \min\left(\frac{v_{\max}(\itild)}{a(i)},
    \frac{\eps}{8} \frac{a(i)}{v_{\max}(i)} \right)
\\ &\le 1 + \sqrt{\frac{\eps}{8}} \le \frac{3}{2} \le \frac{1}{\epsLarge}
\tag{$\epsLarge \le 2/3$}
\end{align*}

\item \textbf{Heavy dense items}: (See \cref{trn:round-up-heavy})
\[ \frac{\Span(\itild)}{\Span(i)} \le \frac{v_{\max}(\itild)}{v_{\max}(i)}
\le 1 + \frac{\epsLarge\eps}{8v_{\max}(i)} \le 1 + \frac{\eps}{8} \le \frac{1}{\epsLarge} \]

\item \textbf{Light dense items}: (See \cref{trn:round-up-light})
\[ \frac{\Span(\itild)}{\Span(i)} \le \frac{v_{\max}(\itild)}{v_{\max}(i)}
\le 1 + \frac{\eps}{8} \le \frac{1}{\epsLarge}  \qedhere \]
\end{enumerate}
\end{proof}

The asymptotic approximation factor given by the Round-and-Approx framework is
\[ \mu\left(1 + \frac{\alpha\rho\gamma}{\mu}\right) + \Theta(1)\eps
= 2\left(1 + \ln\left(\frac{a}{2}\right)\right) + \Theta(1)\eps \]
Using \cref{table:slack-pack-ab}, for $d=1$, we get $2.919065 + \eps'$ when item rotations
are forbidden and $2.810930 + \eps'$ when item rotations are allowed.

\end{document}